\documentclass[a4paper,UKenglish,cleveref, autoref, thm-restate]{lipics-v2021}

\hideLIPIcs  

\usepackage[all,defaultlines=3]{nowidow}
\title{Model Checking for Low Monodimensionality Fragments of \texorpdfstring{\CMSO}{CMSO} on Topological-Minor-Free Graph Classes} 

\titlerunning{Low Monodimensionality Fragments of \texorpdfstring{\CMSO}{CMSO} } 

\author{Ignasi Sau}{LIRMM, Univ Montpellier, CNRS, Montpellier, France.}{ignasi.sau@lirmm.fr}{}{National Research Agency (ANR) project ELIT (ANR-20-CE48-0008-01).}

\author{Nicole Schirrmacher}{University of Bremen, Bremen, Germany.}{schirrmacher@uni-bremen.de}{0000-0002-1740-7478}{French-German Collaboration ANR/DFG Project UTMA (ANR-20-CE92-0027, DFG 446200270).}

\author{Sebastian Siebertz}{University of Bremen, Bremen, Germany.}{siebertz@uni-bremen.de}{0000-0002-6347-1198}{}

\author{Giannos Stamoulis}{Université Paris Cité, CNRS, IRIF, F-75013, Paris, France.}{stamoulis@irif.fr}{}{}

\author{Dimitrios M. Thilikos}{LIRMM, Univ Montpellier, CNRS, Montpellier, France.}{sedthilk@thilikos.info}{0000-0003-0470-1800}{French-German Collaboration ANR/DFG Project UTMA (ANR-20-CE92-0027, DFG 446200270), Franco-Norwegian project PHC AURORA 2024-2025 (Projet n°\! 51260WL), and French National Research Agency (ANR) under project GODASse ANR-24-CE48-4377 and under the France 2030 grant reference number
ANR-24-RRII-0002 operated by the Inria Quadrant Program.}

\author{Alexandre Vigny}{LIMOS, Univ Clermont Auvergne, Clermont Auvergne INP.}{alexandre.vigny@uca.fr}{0000-0002-4298-8876}{French National Research Agency under France 2030 project ANR-23-IACL-0006.}

\authorrunning{
I. Sau, N. Schirrmacher, S. Siebertz, G. Stamoulis, D.M. Thilikos, A. Vigny} 

\Copyright{Jane Open Access and Joan R. Public} 



\category{} 

\relatedversion{} 

\nolinenumbers 

\EventEditors{Claudia Faggian and Joost-Pieter Katoen}
\EventNoEds{2}
\EventLongTitle{Forty-First Annual Symposium on Logic in Computer Science}
\EventShortTitle{LICS 2016}
\EventAcronym{LICS}
\EventYear{2016}
\EventDate{July 20--23, 2016}
\EventLocation{Lisbon, Portugal}
\EventLogo{}
\SeriesVolume{380}
\ArticleNo{48}

\usepackage{latexsym,
amsthm,
color,
amssymb,
url,
bm,
cite,
amssymb,
microtype,
amsmath,
amsfonts,
mathtools}
\usepackage{verbatim,listings}
\usepackage{todonotes}
\usepackage{float}

\usepackage{caption}
\usepackage{setspace}
\usepackage{subcaption}
\usepackage{setspace}

\usepackage[inline]{enumitem} 
\usepackage{amsthm}
\usepackage{nicefrac}
\usepackage{dsfont}
\usepackage{tikz}

\RequirePackage[T1]{fontenc}
\RequirePackage[utf8]{inputenc}

\usetikzlibrary{math}

\makeatletter
\input{colordvi}

\usepackage{aliascnt}

\hypersetup{
    colorlinks,
    linkcolor={blue!75!black},
    citecolor={green!75!black},
    urlcolor={blue!75!black},
}

\definecolor{MidnightBlack}{rgb}{0.1,0.1,.34}
\definecolor{MidnightBlue}{rgb}{0.1,0.1,0.43}
\definecolor{Black}{rgb}{0,0, 0}
\definecolor{Blue}{rgb}{0, 0 ,1}
\definecolor{Red}{rgb}{1, 0 ,0}
\definecolor{White}{rgb}{1, 1, 1}
\definecolor{grey}{rgb}{.6, .6, .6}
\definecolor{Mygreen}{rgb}{.0, .7, .0}
\definecolor{Yellow}{rgb}{.55,.55,0}
\definecolor{Mustard}{rgb}{1.0, 0.86, 0.35}
\definecolor{applegreen}{rgb}{0.55, 0.71, 0.0}
\definecolor{darkturquoise}{rgb}{0.0, 0.81, 0.82}
\definecolor{celestialblue}{rgb}{0.29, 0.59, 0.82}
\definecolor{green_yellow}{rgb}{0.68, 1.0, 0.18}
\definecolor{crimsonglory}{rgb}{0.75, 0.0, 0.2}
\definecolor{darkmagenta}{rgb}{0.30, 0.0, 0.30}
\definecolor{magenta}{rgb}{0.50, 0.0, 0.50}
\definecolor{internationalorange}{rgb}{1.0, 0.31, 0.0}
\definecolor{darkorange}{rgb}{1.0, 0.55, 0.0}
\definecolor{ao}{rgb}{0.0, 0.5, 0.0}
\definecolor{awesome}{rgb}{1.0, 0.13, 0.32}
\definecolor{darkcyan}{rgb}{0.0, 0.50, 0.50}
\definecolor{violet}{rgb}{0.93, 0.51, 0.93}
\definecolor{brown}{rgb}{0.65, 0.16, 0.16}
\definecolor{orange}{rgb}{1.0, 0.65, 0.0}
\definecolor{cornflowerblue}{rgb}{0.39, 0.58, 0.93}

\newcommand{\remove}[1]{}

\newcounter{func}
\newcommand{\funref}[1]{\hyperref[#1]{f_{\ref*{#1}}}} 

\tikzset{black node/.style={draw, circle, fill = black, minimum size = 5pt, inner sep = 0pt}}
\tikzset{white node/.style={draw, circlternary_treese, fill = white, minimum size = 5pt, inner sep = 0pt}}
\tikzset{normal/.style = {draw=none, fill = none}}
\tikzset{lean/.style = {draw=none, rectangle, fill = none, minimum size = 0pt, inner sep = 0pt}}
\usetikzlibrary{decorations.pathreplacing}
\usetikzlibrary{arrows.meta}
\usetikzlibrary{calc}

\usetikzlibrary{shapes}
\tikzset{diam/.style={draw, diamond, fill = black, minimum size = 7pt, inner sep = 0pt}}

\tikzset{
	position/.style args={#1:#2 from #3}{
		at=($(#3)+(#1:#2)$)
	}
}

\tikzset{
  v:main/.style = {draw, circle, scale=0.8, thick,fill=black,inner sep=0.7mm},
  v:ghost/.style = {inner sep=0pt,scale=1},
  v:marked/.style = {circle, scale=1.3, fill=DarkGoldenrod,opacity=0.4},
  >={latex},
  e:main/.style = {line width=1pt}
}

\newcommand{\Acal}{\mathcal{A}}
\newcommand{\Bcal}{\mathcal{B}}

\newcommand{\Gcal}{\mathcal{G}}
\newcommand{\Hcal}{\mathcal{H}}

\newcommand{\Ocal}{\mathcal{O}}

\newcommand{\Zcal}{\mathcal{Z}}

\newcommand{\Nbbb}{\mathbb{N}}

\newcommand{\Oh}{\mathcal{O}}

\newcommand{\Tt}{\mathcal{T}}

\RequirePackage{stmaryrd}
\usepackage{textcomp}
\DeclareUnicodeCharacter{2286}{\subseteq}
\DeclareUnicodeCharacter{2192}{\ifmmode\to\else\textrightarrow\fi}
\DeclareUnicodeCharacter{2203}{\ensuremath\exists}
\DeclareUnicodeCharacter{183}{\cdot}
\DeclareUnicodeCharacter{2200}{\forall}
\DeclareUnicodeCharacter{2264}{\leq}
\DeclareUnicodeCharacter{2265}{\geq}
\DeclareUnicodeCharacter{8614}{\mathbin{\mapsto}}
\DeclareUnicodeCharacter{8656}{\Leftarrow}
\DeclareUnicodeCharacter{8657}{\Uparrow}
\DeclareUnicodeCharacter{8658}{\Rightarrow}
\DeclareUnicodeCharacter{8659}{\Downarrow}
\DeclareUnicodeCharacter{8669}{\rightsquigarrow}
\newcommand{\eqdef}{\stackrel{{\scriptsize\rm def}}{=}}
\DeclareUnicodeCharacter{8797}{\eqdef}
\DeclareUnicodeCharacter{8870}{\vdash}
\DeclareUnicodeCharacter{8873}{\Vdash}
\DeclareUnicodeCharacter{22A7}{\models}
\DeclareUnicodeCharacter{9121}{\lceil}
\DeclareUnicodeCharacter{9123}{\lfloor}
\DeclareUnicodeCharacter{9124}{\rceil}
\DeclareUnicodeCharacter{2208}{\in}
\DeclareUnicodeCharacter{9126}{\rfloor}
\DeclareUnicodeCharacter{9655}{\triangleright}
\DeclareUnicodeCharacter{9665}{\triangleleft}
\DeclareUnicodeCharacter{9671}{\diamond}
\DeclareUnicodeCharacter{9675}{\circ}
\DeclareUnicodeCharacter{10178}{\bot}
\DeclareUnicodeCharacter{10214}{} 
\DeclareUnicodeCharacter{10215}{} 
\DeclareUnicodeCharacter{10229}{\longleftarrow}
\DeclareUnicodeCharacter{10230}{\longrightarrow}
\DeclareUnicodeCharacter{10231}{\longleftrightarrow}
\DeclareUnicodeCharacter{10232}{\Longleftarrow}
\DeclareUnicodeCharacter{10233}{\Longrightarrow}
\DeclareUnicodeCharacter{10234}{\Longleftrightarrow}
\DeclareUnicodeCharacter{10236}{\longmapsto}
\DeclareUnicodeCharacter{10238}{\Longmapsto} 
\DeclareUnicodeCharacter{10503}{\Mapsto}    
\DeclareUnicodeCharacter{10971}{\mathrel{\not\hspace{-0.2em}\cap}}
\DeclareUnicodeCharacter{65294}{\ldotp}
\DeclareUnicodeCharacter{65372}{\mid}

\definecolor{Red}{rgb}{1, 0 ,0}
\definecolor{Blue}{rgb}{0, 0 ,1}

\newcommand{\hh}{\end{document}}

\newcommand{\p}{{\sf p}}

\newcommand{\ext}{{\sf ext}\xspace}
\newcommand{\rtw}{\mathsf{ad\textsf{-}brg}\xspace}%
\newcommand{\bog}{{\sf bog}\xspace}%
\newcommand{\adeg}{{\sf adeg}\xspace}

\renewcommand{\phi}{\varphi}

\newcommand{\atw}{{\sf atw}\xspace}

\newcommand{\aall}{\mathcal{A}_{{\text{\rm  \textsf{all}}}}}

\newcommand{\brg}{{\sf brg}\xspace}
\newcommand{\tw}{{\sf tw}\xspace}
\newcommand{\size}{{\sf size}\xspace}
\newcommand{\td}{{\sf td}\xspace}

\newcommand{\diss}{\mathsf{diss}}
\newcommand{\ap}{{\sf ap}\xspace}
\newcommand{\tp}{{\sf tp}\xspace}

\newcommand{\adh}{\ensuremath{\mathsf{adh}}}

\newcommand{\torso}{\ensuremath{\mathsf{torso}}}

\newcommand{\detail}{\text{$\mathsf{detail}$}\xspace}

\usepackage{mathrsfs}
\newcommand{\Cc}{\mathscr{C}}

\newcommand{\cc}{{\sf cc}\xspace}%
\newcommand{\ttw}{{\sf ttw}\xspace}%
\newcommand{\numen}[1]{\ifthenelse{\not\equal{#1}{1}}{#1}{}}

\usepackage{color}

\definecolor{vagelisColour}{RGB}{0, 65, 130}

\newcommand{\extfolio}{\textnormal{\textsf{-ext-folio}}\xspace}%

\newcommand{\abtmpair}{\textnormal{\textsf{abtm-pair}}\xspace}%

\newcommand{\abtm}{\textnormal{\textsf{abtm}}\xspace}%

\newcommand{\folio}{\textnormal{\textsf{-folio}}\xspace}%

\usepackage[utf8]{inputenc}

\usepackage{amsthm}
\usepackage{amssymb}
\usepackage{amsmath}
\usepackage{mathtools}

\usepackage{tikz}
\usepackage{graphicx}
\usepackage{caption}
\usepackage{subcaption}
\usepackage{xspace}

\newcommand{\DP}{\mathsf{dp}}

\newcommand{\FO}{\ensuremath{\mathsf{FO}}\xspace}
\newcommand{\FOconn}{\ensuremath{\mathsf{FO}\hspace{0.8pt}\raisebox{-0.3pt}{\texttt{+}}\hspace{0.8pt}\mathsf{conn}}\xspace}
\newcommand{\MSO}{\ensuremath{\mathsf{MSO}}\xspace}
\newcommand{\CMSO}{\ensuremath{\mathsf{CMSO}}\xspace}

\newcommand{\CMSOtw}{\ensuremath{\CMSO\hspace{-0.3pt}\texttt{/}\hspace{0.3pt}\ttw}\xspace}

\newcommand{\FODP}{\ensuremath{\FO\hspace{0.8pt}\raisebox{-0.3pt}{\texttt{+}}\hspace{0.8pt}\DP}\xspace}

\newcommand{\CMSOtwDP}{\ensuremath{\CMSO\hspace{-0.3pt}\texttt{/}\hspace{0.3pt}\ttw\hspace{0.8pt}\raisebox{-0.3pt}{\texttt{+}}\hspace{0.8pt}\DP}\xspace}

\newcommand{\CMSObrg}{\ensuremath{\CMSO\hspace{-0.3pt}\texttt{/}\hspace{1pt}\brg}\xspace}

\newcommand{\CMSObrgDP}{\ensuremath{\CMSObrg\hspace{0.8pt}\raisebox{-0.3pt}{\texttt{+}}\hspace{0.8pt}\DP}\xspace}

\newcommand{\CMSObog}{\ensuremath{\CMSO\hspace{-0.3pt}\texttt{/}\hspace{0.3pt}\bog}\xspace}

\newcommand{\CMSOsize}{\ensuremath{\CMSO\hspace{-0.3pt}\texttt{/}\hspace{0.3pt}\size}\xspace}
\newcommand{\CMSOsizeDP}{\ensuremath{\CMSOsize\hspace{0.8pt}\raisebox{-0.3pt}{\texttt{+}}\hspace{0.8pt}\DP}\xspace}

\newcommand{\CMSOatw}{\ensuremath{\CMSO\hspace{-0.3pt}\texttt{/}\hspace{0.3pt}\atw}\xspace}
\newcommand{\CMSOatwDP}{\ensuremath{\CMSOatw\hspace{0.8pt}\raisebox{-0.3pt}{\texttt{+}}\hspace{0.8pt}\DP}\xspace}

\newcommand{\CMSOttw}{\ensuremath{\CMSO\hspace{-0.3pt}\texttt{/}\hspace{0.3pt}\ttw}\xspace}

\newcommand{\CMSOttwDP}{\ensuremath{\CMSOttw\hspace{0.8pt}\raisebox{-0.3pt}{\texttt{+}}\hspace{0.8pt}\DP}\xspace}

\newcommand{\CMSObogDP}{\ensuremath{\CMSObog\hspace{0.8pt}\raisebox{-0.3pt}{\texttt{+}}\hspace{0.8pt}\DP}\xspace}

\newcommand{\CMSOp}{\ensuremath{\CMSO\hspace{-0.3pt}\texttt{/}\hspace{0.3pt}\p}\xspace}
\newcommand{\CMSOpDP}{\ensuremath{\CMSOp\hspace{0.8pt}\raisebox{-0.3pt}{\texttt{+}}\hspace{0.8pt}\DP}\xspace}
\newcommand{\CMSOppDP}{\ensuremath{\CMSOp'\hspace{0.8pt}\raisebox{-0.3pt}{\texttt{+}}\hspace{0.8pt}\DP}\xspace}

\newcommand{\N}{\mathbb{N}}

\newcommand{\bag}{\mathsf{bag}}
\newcommand{\cone}{\mathsf{cone}}
\newcommand{\cmp}{\mathsf{comp}}
\newcommand{\mrg}{\mathsf{mrg}}
\newcommand{\parent}{\mathsf{parent}}
\newcommand{\children}{\mathsf{children}}

\newcommand{\bound}[1]{\mathbf{#1}}
\newcommand{\pretp}{\preceq}
 
\newcommand{\type}{\mathsf{type}}
\newcommand{\exttype}{\mathsf{ext}\textsf{-}\mathsf{type}}

\begin{document}

\maketitle

\begin{abstract}
\noindent Algorithmic meta-theorems explain the tractability of large classes of computational problems by linking logical expressibility with structural graph properties. While extensions of first-order logic such as \textsf{FO+dp} admit efficient model checking on graph classes excluding a fixed topological minor, comparable results for richer fragments of \textsf{CMSO} were previously unknown.
We further develop the framework of Sau, Stamoulis, and Thilikos [SODA 2025] for fragmenting \textsf{CMSO} via \emph{annotated graph parameters}, which restrict set quantification to vertex sets satisfying bounded structural conditions.
Following this approach, we identify a fragment of \textsf{CMSO}, namely the one defined by allowing quantification only over sets having what we call \emph{low monodimensionality}, that generalizes several previously-known logics and we show that model checking for this fragment, enhanced with the disjoint-paths predicate, is fixed-parameter tractable on topological-minor-free graph classes. Such classes essentially delimit the tractability for this logic on subgraph-closed classes.
As a consequence, our results lift several known algorithmic meta-theorems beyond first-order logic to the topological-minor-free setting.\end{abstract}

\section{Introduction}
 

\subparagraph{Algorithmic meta-theorems.}
Algorithmic meta-theorems (AMTs) aim to explain, in a uniform and conceptually clean way, 
\emph{why} large families of computational problems admit efficient algorithms on large classes of inputs.
Rather than designing a new algorithm for each individual problem and each individual graph class,
AMTs identify general principles showing that \emph{all} problems expressible in a certain logic
can be solved efficiently on \emph{all} graphs from a certain structural class.
In this sense, meta-theorems do not merely provide algorithms; they unify, systematize, and
conceptually explain tractability.

Formally, an AMT typically asserts that for a logic $\mathcal{L}$ and a graph class $\Cc$,
the model checking problem, that is, deciding whether a given graph $G\in \Cc$ satisfies a given formula $\varphi\in \mathcal{L}$, in symbols $G \models \varphi$, can be solved in time $$
f(|\varphi|,\Cc) \cdot |G|^{c}, $$
for some function $f$, 
where the exponent $c$ is independent of both the formula $\varphi$ and the class~$\Cc$.
Such results show that once the structural complexity of the input and the expressive power of the specification language are bounded,
the combinatorial explosion inherent in many problems can be controlled.
Over the past decades, AMTs have become a central paradigm in algorithmic graph theory and finite model theory,
unifying a wide range of algorithmic techniques across many problems and graph classes
\cite{GroheK09,Kreutzer11,Grohe08,SiebertzV24}.

\subparagraph{Monadic second-order logic with modulo counting.}
The classical and most influential example of an algorithmic meta-theorem is Courcelle’s theorem
\cite{Courcelle1990,Courcelle97,Courcelle92}.
It states that every graph property definable in monadic second-order logic with modulo counting (\CMSO)
can be decided in linear time on graphs of bounded treewidth.
This theorem abstracts dynamic programming over tree decompositions and provides a single,
uniform explanation for the tractability of a vast number of problems on tree-like graph classes.
It has been extended to related width parameters, most notably cliquewidth \cite{courcelle2000linear}.
At the same time, matching lower bounds show that these results are essentially tight.
Indeed, extending efficient \CMSO model checking to substantially more general graph classes is not to be expected
\cite{ganian2014lower,kreutzer2010lower}.

\subparagraph{First-order logic.}
A different and highly successful line of research focuses on first-order logic~(\FO),
which is significantly less expressive than \CMSO, but enjoys much broader algorithmic tractability.
Model checking for \FO is tractable on a wide range of sparse graph classes,
including bounded degree graphs \cite{Seese96line},
bounded expansion and nowhere dense classes \cite{GroheKS17dec,DvorakKT13firs}. 
These algorithms rely on locality principles and structural decompositions
\cite{DawarGK07loca,FrickG01decid}, and are optimal for subgraph-closed classes, that is, \FO model checking on a subgraph-closed class is fixed-parameter tractable if and only if the class is nowhere dense~\cite{dvovrak2010deciding}.
Recently, the tractability results have also been extended to dense graph classes, e.g.~to classes of bounded twin-width \cite{BonnetKTW22twinI}.
A main conjecture from the area states that \FO model checking is fixed-parameter tractable on a class of graphs closed under induced subgraphs if and only if the class is monadically dependent, see~\cite{bonnet2022twin}, and partial results as well as a beautiful structure theory for monadically dependent classes is being developed, see e.g.~\cite{braunfeld2025decomposition,dreier2024first,DreierMS23,dreier2024flip,GajarskyMMOPPSS23,dreier_et_al23}. 

From a high-level perspective, these results reveal a clear tradeoff:
\FO is tractable on very broad graph classes, but its expressive power is limited.
In particular, \FO cannot express basic counting properties, reachability, or global connectivity.
Thus, many fundamental algorithmic problems lie beyond its expressive scope.

\subparagraph{The gap between \FO and \CMSO.}
For a long time, there was a striking gap between the two extremes:
on the one hand, \CMSO, which is very expressive but only tractable on rather restricted graph classes,
and on the other hand, \FO, which is widely tractable but too weak to express many important properties.
It was natural to ask whether there exist logics in between these two that combine
\emph{meaningful expressive power} with \emph{broad algorithmic tractability}.

Classical candidates for closing this gap include (monadic) transitive closure logic, (monadic) least fixed-point logic,
and related extensions of first-order logic, see~\cite{Libkin04,ebbinghaus1999finite}.
However, these logics turn out to be intractable already on very simple graph classes,
such as planar graphs of maximum degree three.
As a consequence, their study did not lead to robust algorithmic meta-theorems,
and they failed to provide a satisfying explanation of tractability beyond very restricted settings~\cite{Grohe08}.
Other extensions are those with counting or other numerical predicates, see e.g.~\cite{dreier2021approximate,grohe2018first,kuske2017first,kuske2018gaifman}.

\subparagraph{New algorithmically meaningful extensions of \FO.}
Recently, a new and promising direction emerged.
Instead of adding general recursion mechanisms, recent work has begun to directly enrich first-order logic
with carefully chosen operators that are strong enough to express important algorithmic properties,
yet sufficiently tame to preserve tractability on broad graph classes.
The guiding principle is to extend \FO in a way that aligns with known algorithmic techniques
and structural decompositions, thereby enabling genuine algorithmic meta-theorems
that go beyond pure first-order expressiveness.

One approach enhances \FO with restricted connectivity predicates, leading to logics such as \FOconn, also called \textsl{separator logic}~\cite{Bojanczyk21separ,SchirrmacherSV23}.
Model checking for \FOconn-expressible formulas can be performed
in time $f(|\varphi|,|H|) \cdot |G|^{3}$ on $H$-topological-minor-free graphs~\cite{PilipczukSSTV22}.
A stronger extension, \FODP~\cite{SchirrmacherSV23}, allows predicates expressing the existence of vertex-disjoint paths.
More precisely, \FO is enhanced with the $2k$-ary \emph{disjoint-paths predicate}
$\DP_k(x_1, y_1, \ldots, x_k, y_k)$,
where $x_1, y_1, \ldots, x_k, y_k$ are first-order terms,
which expresses the existence of $k$ pairwise vertex-disjoint paths between the valuations
of $x_i$ and $y_i$, for $i \in \{1,\ldots,k\}$.
With the help of the disjoint-paths predicate, \FODP can define global properties, including connectivity and (topological) minor exclusion.
As proved in~\cite{GolovachST22model},
model checking for \FODP-expressible formulas can be performed
in time $f(|\varphi|,|H|) \cdot |G|^{2}$ on $H$-minor-free graphs, and even more generally, as shown in~\cite{SchirrmacherSSTV24mode}, in time $f(|\varphi|,|H|) \cdot |G|^{3}$ on $H$-topological-minor-free graphs.
Just as nowhere density is the limit of tractability for~\FO on subgraph-closed classes, it turns out that for \FODP on subgraph-closed classes, the limit of tractability are classes with excluded topological minors, as it is already the case for \FOconn~\cite[Theorem 1.3]{PilipczukSSTV22}.

In an attempt to make the logics even more expressive and useful for algorithmic purposes, another enhancement of \FO was introduced by Fomin, Golovach, Sau, Stamoulis, and Thilikos in~\cite{fominGSST25compound}, the so-called \emph{compound logic} $\tilde{\Theta}^{\DP}$, which combines syntactic fragments of \FODP and \CMSO  
and is essentially a language tailored to express general families of graph modification problems via a modification/target machinery.
As proved in~\cite{fominGSST25compound}, model checking for $\tilde{\Theta}^{\DP}$-expressible formulas can be done in time $f(|\varphi|,|H|) \cdot |G|^{2}$ on $H$-minor-free graphs.

\subparagraph{Fragmenting \CMSO by annotated graph parameters.}
The aforementioned approaches are based on extensions of $\FO$.
In~\cite{SauST25Parameterizing}, a different perspective is proposed, which allows
these and other enhanced logics to be viewed as fragments of \CMSO by restricting the use of set quantification.
In particular, \cite{SauST25Parameterizing}
introduces \emph{annotated treewidth logic}, denoted $\CMSOatw$, which is a fragment of
\CMSO in which second-order quantification is constrained to sets of restricted treewidth {(formally, we do not measure the treewidth of the graph induced by the quantified set $X$, but the treewidth of $X$-minors to obtain a robust theory, the formal definition is given shortly)}.
This logic is tractable on $H$-minor-free graph classes.

The idea of structurally restricting set quantification also appears in recent work on \emph{low-rank $\MSO$},  which is defined as the fragment of $\MSO$ where set quantification is limited to vertex sets of bounded cut-rank~\cite{bojanczyk2025lowrankmso}.
Somewhat surprisingly, over sparse graph classes, this fragment has the same expressive power as separator logic.
Moreover, this logic is also applicable to dense graph classes, where it corresponds to related formalisms such as \emph{flip-connectivity logic} and \emph{flip-reachability logic}, which adapt the separator paradigm to settings in which vertex separators are no longer useful.
Crucially, interpreting these formalisms as low-rank fragments of $\MSO$ yields a unifying semantic framework:
rather than relying on ad-hoc extensions of $\FO$, one obtains a hierarchy of \emph{\MSO fragments} whose structural restrictions simultaneously govern expressiveness and algorithmic tractability.

\subparagraph{Finding elegant logics.}
From our perspective, both first-order logic and monadic second-order logic are extremely natural formalisms for reasoning about graphs and combinatorial structures. 
Over the years, they have proven to be robust, expressive, and conceptually clean, and a large body of algorithmic meta-theorems is built on them. 
Consequently, we also perceive extensions of \FO, such as enrichments by connectivity or the disjoint-paths predicate, as natural logical operations. 
Similarly, we view restrictions of \CMSO, in particular, restrictions on the sets over which one is allowed to quantify, as conceptually clean and easy to understand. 
In contrast, highly engineered formalisms such as compound logic are extremely powerful and very successful from an algorithmic point of view, but from a logical perspective, they appear more as programming languages than as elegant logics.

The aim of our present work is to study \emph{elegant} restrictions of \CMSO quantification that preserve both conceptual clarity and algorithmic power. 
We propose to organize such restrictions via \emph{annotated graph parameters} and, in particular, via minor-monotone parameters, which, as we will demonstrate, enjoy strong compositionality properties and are well-aligned with dynamic programming techniques. 
We believe that the search for suitable parameters is, to a large extent, a search for elegance: different applications may suggest different parameters, and our framework is intended to support this flexibility. 
In this paper, we focus on two concrete parameters and show that they are functionally
equivalent; one of them is particularly convenient to use, while the other is better suited for tractability arguments. 
More generally, our goal is to provide a uniform framework in which such parameters can be studied and compared.

%

\subparagraph{Getting more technical: parametric fragments of $\CMSO$.}
To present our approach and our results in full generality, let us introduce some basic notions related to annotated graphs.

An \emph{annotated graph} is a pair $(G,X)$, where $G$ is a graph and $X$ is a subset of its vertex set.
An \emph{annotated graph parameter} is a function $\p$ that maps annotated graphs to non-negative integers. 

Given an annotated graph parameter $\p$, the logic
$\CMSOp$ is defined by restricting existential quantification $\exists X\, \phi$
to $\exists_{\p\le k} X\, \phi$, and universal quantification $\forall X\phi$ to $\forall_{\p\le k} X \varphi$, for any positive integer~$k$. 
In such formulas, the set variable $X$ is only allowed to be interpreted with vertex sets $S$ such that $\p(G,S)\le k$. 
The enhancement of $\CMSOp$
by the disjoint-paths predicate is denoted by \mbox{$\CMSOpDP$}.
Note that this formalism provides an extremely easy to use fragment of $\CMSO$ (\CMSOpDP). 
We emphasize that the bounds appearing in parameter-bounded quantifiers are part of
the formula. In particular, the size~$|\phi|$ of a formula~$\phi$ depends on the
largest value of $k$ used in its quantified set variables.

A very simple example of an annotated graph parameter is the size of $X$, that is, \mbox{$\size(G,X)=|X|$}.
It is easy to see that $\CMSOsizeDP$ has the same expressive power as $\FODP$, as size-restricted set quantification can be expressed by explicit vertex quantification. 
The size of the resulting formula depends on the integers $k$ that appear in the $\CMSOsizeDP$ formulas. 
This equivalence allows us to state the main result of~\cite{SchirrmacherSSTV24mode} as follows. 

\begin{proposition}\label{prop_mntghrro}
There exist a computable function $f\colon  \mathbb{N}^2\to\mathbb{N}$ and an algorithm that, given a formula $\varphi\in\textup{\CMSOsizeDP}$, and an $H$-topological-minor-free graph~$G$,
decides whether $G\models \varphi$ in time $f(|\varphi|,|H|)\cdot |G|^3$.
\end{proposition}

An annotated graph $(H,Y)$
is a \emph{minor} of $(G,X)$, denoted $(H,Y)\leq (G,X)$,
if there exists a collection $\Bcal=\{B_{v}\mid v\in V(H)\}$
of pairwise disjoint connected subsets of $V(G)$ such that
(i) for every edge $\{x,y\}\in E(H)$, the set $B_{x}\cup B_{y}$ is connected in $G$, and
(ii) for every $y\in Y$, we have $B_{y}\cap X\neq \emptyset$. We refer to $\Bcal$ as a \emph{minor model} of $(H,Y)$ in $(G,X)$.
We refer to $B_{v}$ as the \emph{branch set} of $v$. 
We call $H$ an \emph{$X$-minor} of $(G,X)$ if $(H,V(H))\leq (G,X)$. 
An annotated graph parameter is \emph{minor-monotone} if
$(H,Y)\leq (G,X)$ implies $\p(H,Y)\leq \p(G,X)$.

A much more relevant structural measure for $X$ than its size is its treewidth, which in the context of annotated graphs leads to the
notion of \emph{annotated treewidth}: the maximum treewidth of any $X$-minor of $G$. 
Intuitively, annotated treewidth measures how much global
structural complexity of $G$ is witnessed by the set~$X$. 
We write $\atw$ for the parameter annotated treewidth. 

%
%

In graph minor theory, treewidth is characterized, up to polynomial equivalence, by the size of the
largest grid minor. The $(k\times k)$-grid is the canonical model of two-dimensionality and the
fundamental obstruction to bounded treewidth. Measuring the largest grid that can be obtained as a
minor with all its vertices witnessed by $X$ therefore captures, in a robust and invariant way, how
``two-dimensionally'' the set $X$ sits inside $G$. 
Formally, the bidimensionality of $X$ in $G$ is the maximum integer $k$ such that $(\Gamma_{k},V(\Gamma_{k}))≤(G,X)$, where~$\Gamma_{k}$ is the $(k\times k)$-grid 
(see \cite{ThilikosW25TheGraph}).
The \emph{biggest rainbow grid}  of~$(G,X)$, denoted by $\brg(G,X)$,
is also called the \emph{bidimensionality} of $X$ in $G$. 
Because of the grid exclusion theorem \cite{ChuzhoyT21towa,RobertsonS86GMV}, the parameters~$\atw$ and~$\brg$ turn out to be functionally equivalent, that is, they functionally bound each other.
It is easy to see that for every~$(G,X)$, $\atw(G,X)\le \size(G,X)$, and, in fact, it can be arbitrarily  smaller.  Therefore, 
$\CMSOatwDP$ is much more expressive than $\CMSOsizeDP$ (that has the same expressive power as  $\FODP$). 

Both $\atw$ and $\brg$ are minor-monotone annotated parameters. 
In \cref{lemm_anotopomin_formula}, 
we show that if~$\p$ is a minor-monotone annotated graph parameter, then for
every fixed $k$, the property $\p(G,X)\le k$ is definable in~$\FODP$ using an $\FODP$ formula that treats the set $X$ as a given (interpretation of) a unary relation symbol. 
As a consequence, the exact choice of the parameter $\p$ does not matter, as we can replace it by any functionally equivalent parameter  and restrict the quantified sets appropriately by $\FODP$ formulas. 
This, applied to~$\atw$ and~$\brg$,  implies that $\CMSOatwDP$ and $\CMSObrgDP$ have the same expressive power. 
The aforementioned translation relies on the well-quasi-ordering
theorem for annotated graph minors \cite{ProtopapasTW2025colorfulminors} and, in general,  yields a non-constructive translation. 

It is important to note that we do \emph{not} say here that $\CMSOpDP$ collapses to $\FODP$.
For general instantiations of $\p$, the logic $\CMSOpDP$ allows quantification over sets of \textsl{unbounded} 
size, subject only to a bound on the parameter $\p$, a feature that cannot be simulated in $\FODP$ alone.
In other words, while the property $\p(G,X)\le k$ for a given set~$X$ is $\FODP$-expressible, the set $X$ itself cannot be quantified using $\FODP$.
We view this restricted form of set quantification as a natural and very convenient restriction of $\CMSO\!+\!\mathsf{dp}$, which is very well suited for practical use in expressing parameter-bounded combinatorial structures. 

In this work, we will take the perspective of treewidth and in the following, restate results that were proved for $\atw$. 
In particular, we can state the result of Sau, Stamoulis, and Thilikos in~\cite{SauST25Parameterizing} for \CMSObrgDP as follows. 

\begin{proposition}\label{prop_mntsro}
There exist a computable function $f\colon \mathbb{N}^2\to\mathbb{N}$ and an algorithm that, given a formula $\varphi\in\textup{\CMSOatwDP}$ (or, equivalently, $\phi\in\textup{\CMSObrgDP}$), and an $H$-minor-free graph $G$,
decides whether $G\models \varphi$ in time $f(|\varphi|,|H|)\cdot |G|^2$.
\end{proposition}

As observed in \cite{SauST25Parameterizing}, the formulas of the compound logic $\tilde{\Theta}^{\DP}$
defined in \cite{fominGSST25compound} can be expressed in $\CMSOatwDP$. 
Therefore, \cref{prop_mntsro} subsumes all previous results on model checking for $H$-minor-free graphs.
Moreover,  in~\cite{SauST25Parameterizing}, it was also argued that the choice of~$\brg$, that is, quantifying over sets
of bounded bidimensionality,
is optimal in \cref{prop_mntsro}, in the sense
that, for every choice of $\p$, that is not parametrically ``weaker than $\brg$'',
$\CMSOp$ is as expressive as $\CMSO$ on planar graphs. 
As \CMSO can define \textsf{NP}-hard problems (such as $3$-colorability) on planar graphs, we cannot expect tractability in this case. 
On the other hand, it remains an interesting question whether $H$-minor-free graphs constitute the combinatorial horizon of model checking for $\CMSOatwDP$.

Notice that \cref{prop_mntsro} and \cref{prop_mntghrro} are incomparable.
Indeed, \cref{prop_mntsro} applies to the more expressive logic $\CMSOatwDP$,
whereas \cref{prop_mntghrro} holds under the more general combinatorial assumption that a fixed graph is excluded as a topological minor.
To the best of our knowledge, prior to our main result (presented below), $\CMSOsizeDP$  was the most expressive logic currently known to admit tractable model checking on graph classes excluding a fixed topological minor.
Furthermore, it is known that this result cannot be extended to more general subgraph-closed graph classes~\cite{SchirrmacherSSTV23}.


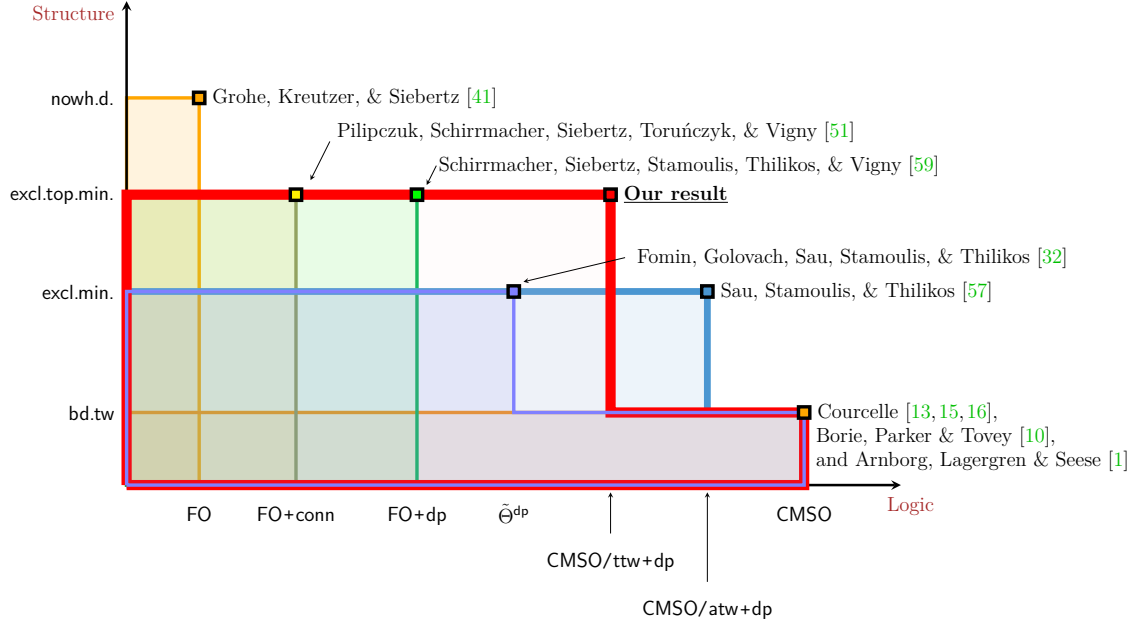
\begin{figure}[ht]
	\begin{center}
\scalebox{0.64}{
\begin{tikzpicture}
	\filldraw[line width=2pt,draw=orange,fill=orange, fill opacity = .13]
    		(0, 0) rectangle (1.5, 8);
  	\draw[line width=1.5pt,-stealth]
    		(0, 0) -- (0,10);
  	\draw[line width=1.5pt,-stealth]
    		(0, 0) -- (16,0);
  	\filldraw[line width=2pt,draw=darkorange,fill=darkorange, fill opacity = .12]
    		(0, 0) rectangle (14, 1.5);	
	\filldraw[line width=2pt,draw=yellow!50!black,fill=orange, fill opacity = .1]
    		(0, 0) rectangle (3.5, 6);
	\filldraw[line width=2pt,draw = blue!30!green,fill=green, fill opacity = 0.1]
		(0, 0) rectangle (6, 6);
  	\filldraw[line width=4pt,draw=celestialblue,fill=celestialblue, fill opacity = .10]
    		(0,0) -- (0,4) -- (12,4) -- (12,1.5) -- (14,1.5) -- (14,0) -- (0,0);

    \filldraw[line width=6pt,draw=red,fill=red!10!white, fill opacity = .1]
    		(0, 0) -- (0,6) -- (10, 6) -- (10,1.5) -- (14,1.5) -- (14,0) -- (0,0);

	\node[label={left:\Large\textcolor{brown}{Structure}}]
   		at (0,9.75) {};
   	\node[label={left:\Large\textsf{nowh.\!\!\! d.}}]
   		at (-0.0,8) {};
   	\node[label={left:\Large\textsf{excl.\!\!\! top.\!\!\! min.}}]
   		at (-0.0,6) {};
   	\node[label={left:\Large\textsf{excl.\!\!\! min.}}]
   		at (-0.0,4) {};
   	\node[label={left:\Large\textsf{bd.\!\!\! tw}}]
   		at (-0.0,1.5) {};
	
   	\node[label={below:\Large\FO}]
   		at (1.5,-0.2) {};	
   	\node[label={below:\Large\FOconn}]
   		at (3.5,-0.2) {};
   	\node[label={below:\Large\FODP}]
   		at (6,-0.2) {};
   	\node[label={below:\Large{$\tilde{\Theta}^\DP$}}]
   		at (8,-0.15) {};
    \node[label={below:\Large{$\CMSOttwDP$}}]
   		at (10,-1.15) {};
         \draw[-stealth] (10,-1) -- (10,-0.15);
    \node[label={below:\Large{$\CMSOatwDP$}}]
   		at (12,-2.15) {};
        \draw[-stealth] (12,-2) -- (12,-0.15);
   	\node[label={below:\Large\CMSO}]
   		at (14,-0.2) {};

    \filldraw[line width=2pt,draw=blue!50!white,fill=blue!50!white, fill opacity = .1]
    		(0, 0) -- (0,4) -- (8, 4) -- (8,1.5) -- (14,1.5) -- (14,0) -- (0,0);

	\node[rectangle, minimum size=2pt,draw=black,line width=2pt, fill=orange,
		label={right:\Large Grohe, Kreutzer, \& Siebertz~\cite{GroheKS17dec}}]
			at (1.5,8) {};
	\node[rectangle, minimum size=2pt,draw=black,line width=2pt, fill=yellow]
			at (3.5,6) {};
		\node[label={right:\Large Pilipczuk, Schirrmacher, Siebertz, Toruńczyk, \& Vigny~\cite{PilipczukSSTV22}}]
			at (4.1,7.3) {};
		\draw[-stealth] (4.2,7.1) -- (3.8,6.3);
	\node[rectangle, minimum size=2pt,draw=black,line width=2pt, fill=green]
    at (6,6) {};
    \node[label={right:\Large {Schirrmacher, Siebertz, Stamoulis, Thilikos, \& Vigny~\cite{SchirrmacherSSTV24mode}}}]
			at (6.2,6.6) {};
    \draw[-stealth] (6.4,6.6) -- (6.2,6.2);
	\node[rectangle, minimum size=2pt,draw=black,line width=2pt, fill=white!50!blue] at (8,4) {};
    \node[label={right:\Large {Fomin, Golovach, Sau, Stamoulis, \& Thilikos \cite{fominGSST25compound}}}]
			at (10.3,4.7) {};
        \draw[-stealth] (10.3,4.7) -- (8.2,4.2);    
    \node[rectangle, minimum size=2pt,draw=black,line width=2pt, fill=red,
		label={right:\Large \underline{\bf Our result}}]
			at (10,6) {};
    \node[rectangle, minimum size=2pt,draw=black,line width=2pt, fill=celestialblue,
		label={right:\Large Sau, Stamoulis, \& Thilikos~\cite{SauST25Parameterizing}}]
			at (12,4) {};
	\node[rectangle, minimum size=2pt,draw=black,line width=2pt, fill=orange,
		label={right:\Large Courcelle \cite{Courcelle1990,Courcelle97,Courcelle92},}]
			at (14,1.5) {};
	\node[label={right:\Large Borie, Parker \& Tovey \cite{BoriePT92auto},}]
			at (14,1) {};
	\node[label={right:\Large and Arnborg, Lagergren \&  Seese \cite{ArnborgLS91easy}}]
			at (14,0.5) {};	
	\node[label={below:\Large\textcolor{brown}{Logic}}]
			at (16.2,0) {};	
			\end{tikzpicture}}
	\end{center}

	\caption{Our result compared to the algorithmic meta-theorems on monotone graph classes mentioned in the introduction. Throughout the paper, we focus on \emph{monotone} graph classes, that is, classes closed under taking subgraphs. 
    This is the natural regime for algorithmic meta-theorems based on (topological) minor exclusion. 
    For this reason, we do not include the new tractability results for hereditary classes, such as~\cite{braunfeld2025decomposition,dreier2024first,DreierMS23,dreier2024flip,GajarskyMMOPPSS23,dreier_et_al23}. The colorful surfaces in the figure depict the regions where the model checking for the corresponding logic is fixed-parameter tractable, and note that Courcelle's theorem~\cite{Courcelle1990,Courcelle97,Courcelle92} is included in a region whenever \CMSO collapses to the corresponding logic on graphs of bounded treewidth (for the logic introduced in this paper, cf. \cref{lem_MSOcollapse_bdtw}).}
	\label{fig_amts_full}
\end{figure}


\subparagraph{Torso treewidth: an annotated parameter between annotated treewidth and size.}
In our present work, we study the following parameter called \emph{torso treewidth}. 
Given an annotated graph $(G,X)$, the \emph{torso}
of~$X$ in $G$ is the graph obtained from $G[X]$ by adding, for every connected component $C$ of $G-X$, all edges that make the neighbors of~$V(C)$ in $X$ a clique.
The \emph{torso treewidth} of $(G,X)$, denoted by $\ttw(G,X)$, is defined as the minimum~$k$
such that $X$ is contained in a vertex set whose torso has treewidth at most~$k$.
Clearly, $\atw(G,X)\leq \ttw(G,X)\leq \size(G,X)$, hence $\CMSOttwDP$
is a fragment of $\CMSOatwDP$.

This parameter for annotated graphs has
already been proposed by Jansen and Swennenhuis in \cite{JansenS2024Steiner}
under the name \emph{$X$-free treewidth}
of $G$, and by Hodor, La, Micek, and Rambaud in
\cite{HodorLMR24quick}
under the notation $\tw(G,X)$. Moreover, a
``multi-annotated'' version of torso treewidth has been recently introduced by Protopapas,  Thilikos, and  Wiederrecht in \cite{ProtopapasTW2025colorfulminors} 
for colorful graphs (that is, graphs with multiple annotated sets)
under the name \emph{restricted treewidth}.

\subparagraph{Our results.}
We study the annotated parameter $\ttw$ and the resulting logic $\CMSOttwDP$. 
We first identify an annotated  parameter that relates to $\ttw$ as $\brg$ relates to $\atw$. 
%
Given a graph $G$ and a set $X\subseteq V(G)$, we define the \emph{monodimensionality}
of $X$ in~$G$ as the maximum~$k$ such that 
$(\Gamma_{k},P(\Gamma_{k}))≤(G,X)$, where $P(\Gamma_{k})$ is the perimeter of~$\Gamma_{k}$.
We define the \emph{biggest outer-annotated grid} of $(G,X)$, denoted by  $\bog(G,X)$, as the  monodimensionality of $X$ in~$G$. 
It was proved in \cite{HodorLMR24quick}, that
$\bog$ and $\ttw$ are functionally equivalent parameters.

%

We then show our main result, stating that we can lift the model checking result for $\CMSOsizeDP$ to $\CMSOttwDP$  on classes with excluded topological minors.

\begin{theorem}\label{th_our_g_res_tp_bog}
There exist a computable function $f\colon  \mathbb{N}^2\to\mathbb{N}$ and an algorithm that, given a formula $\varphi\in\textup{\CMSOttwDP}$ (or, equivalently, $\phi\in\textup{\CMSObogDP}$), and an $H$-topological-minor-free graph~$G$,
decides whether $G\models \varphi$ in time $f(|\varphi|,|H|)\cdot |G|^3$.
\end{theorem}

We also show (cf. \cref{sec_hardness})
that \cref{th_our_g_res_tp_bog}, even for \CMSOttw alone, is optimal in the sense that it cannot be extended to any monotone graph class that contains all graphs as topological minors.
On the other hand, it is an open question whether 
$\ttw$ is the most ``general'' annotated graph parameter for which 
an AMT as in \cref{th_our_g_res_tp_bog} exists.

At this point, we emphasize that
torso treewidth is closer to the mechanism underlying the definition of the compound
logic $\tilde{\Theta}^{\DP}$. Indeed, a formula of
$\tilde{\Theta}^{\DP}$ roughly expresses the existence of
a modulator $X$ whose torso has bounded treewidth and satisfies some $\CMSO$ formula $\beta$, while $G-X$
satisfies some $\FODP$ formula~$\gamma$.
In this way, formulas in $\tilde{\Theta}^{\DP}$
are compound formulas in which $\beta$ specifies
a modification of the graph $G$, whereas $\gamma$ specifies
the target property after this modification.
In \cite{fominGSST25compound}, numerous algorithmic applications
are presented for graph modification
problems expressible within this modulator/target framework. 
It follows that every formula
in $\tilde{\Theta}^{\DP}$ can be expressed
using a bounded torso treewidth quantification
over $X$, which readily  implies that $\tilde{\Theta}^{\DP}$
can be viewed as a fragment of $\CMSOttwDP$.
Consequently, all results of \cite{fominGSST25compound},
together with their numerous applications to graph modification
problems, also hold in the more general setting of
$H$-topological-minor-free graph classes.

When comparing $\CMSOttwDP$ and $\FODP$, a typical example to consider is the parameter of $\Hcal$-treewidth
introduced in \cite{EibenGHK21}. For every minor-closed graph class $\Hcal$, checking whether $\Hcal$-treewidth$(G)≤k$ can be expressed in $\CMSOttwDP$ (even in $\tilde{\Theta}^{\DP}$) while  there is evidence that this is not possible in $\FODP$ (see also \cite{fominGSST25compound} for more examples).

On the other hand, $\CMSOttwDP$ can express
problems that fall outside the modulator/target scheme of~$\tilde{\Theta}^{\DP}$.
For instance, one may ask whether the modulator $X$ for which
$G-X$ satisfies the target property is \emph{unique},
as the logic also allows universal quantification over $X$.
Moreover, universal quantification
enables us to require the existence of multiple modulators~$X$
with prescribed diversity, for instance,
satisfying specific Hamming-distance constraints, as suggested by
the ``solution diversity’’ framework introduced by
Baste, Fellows, Jaffke, Masarík,
de Oliveira Oliveira, Philip, and Rosamond
in \cite{BasteFJMOPR22Diversity}.
This suggests that $\CMSOttwDP$
provides a conceptually simpler tool for expressing broad families
of algorithmic problems.
Compared to approaches that define expressibility languages
via problem-definition schemes, such as the modulator/target scheme
proposed in \cite{SauST25Parameterizing}, it offers a more direct
and uniform framework. 

\subparagraph{Our proof and the meta-theoretic value.}
The proof of \cref{th_our_g_res_tp_bog} follows the general strategy developed in the model checking meta-theorems for $\FODP$ on topological-minor-free graph classes, but requires substantial new ingredients to handle restricted \CMSO quantification.
We start by computing a strongly unbreakable decomposition of the input graph using the framework of Cygan, Lokshtanov,  Pilipczuk,  Pilipczuk,  
Saurabh~\cite{cygan2019minimum},
which reduces the problem to understanding the logic on highly connected parts.

We then show that on unbreakable graphs, formulas of $\CMSOttwDP$ collapse to $\CMSOsizeDP$ (which is equivalent to $\FODP$). 
More precisely, when the graph is highly connected, every set quantification over sets $X$ in $\CMSOttwDP$ is restricted to sets of bounded size, and the torso treewidth of $X$ can be expressed in $\FODP$. 

We then establish that $\CMSOttwDP$ satisfies the \textsl{compositionality property}, which is a key property also of \CMSO and \FO. 
Informally, compositionality means that when gluing two graphs along a small interface, the type of the obtained graph depends only on the types of the two graphs that are glued. 
In fact, this follows from a more general result that this property 
is satisfied by $\CMSOpDP$
for \textsl{every} minor-monotone $\p$.
This is proved in \cref{sec_compositionality} and is based on the fact 
that every minor-monotone 
annotated parameter is completely characterized by a \textsl{finite} set of obstructions (observed in  \cite{ProtopapasTW2025colorfulminors} and based on the results in \cite{RobertsonS10GraphminorsXXIII}).
While, in general, this compositionality  result is not constructive, the definitions 
of $\brg$ and $\bog$ (which are equivalent to $\atw$ and $\ttw$ respectively)
make it constructive as both 
parameters are defined in terms of the exclusion of 
a unique parametric annotated graph, that is $(\Gamma_{k},V(\Gamma_{k}))$
and $(\Gamma_{k},P(\Gamma_{k}))$ respectively.
The compositionality property suggests 
that $\CMSOpDP$ is a robust  
logical framework for expressing the algorithmic
potential of fragments of $\CMSO$
and investigating what is the  combinatorial horizon where 
they permit the proof of AMTs.

Finally, we combine these ingredients via dynamic programming over the unbreakable
decomposition, using boundaried graphs and folios, as inspired by the techniques of~\cite{SchirrmacherSSTV24mode},
to compute and store small 
type-preserving representative graphs.
In particular, in order to compute a small type-preserving representative, we rewrite the types in $\FODP$, since we know that our logic collapses to this one on unbreakable graphs, and we use the model checking algorithm of~\cite{SchirrmacherSSTV24mode} to compute this representative. The unbreakability property of our decomposition does not guarantee that the graph ``hanging'' from each node is unbreakable, but we can guarantee unbreakability (with worse bounds) assuming that we have already computed small-size representatives for the children nodes. Let us stress here that the collapse result together with the use of the unbreakable decomposition allows for a clean reduction to the main result of~\cite{SchirrmacherSSTV24mode}. Interestingly, our algorithm bypasses the explicit distinction into the ``clique-minor-free'' and the ``large-clique-minor'' cases of the algorithm of~\cite{SchirrmacherSSTV24mode} by doing so implicitly when using the model checking algorithm of~\cite{SchirrmacherSSTV24mode} as a blackbox. Consequently, it also avoids the direct use of the algorithm of~\cite{SauST25Parameterizing} for minor-free graphs.

Overall, this yields the claimed fully constructive fixed-parameter tractable model checking algorithm for
$\CMSOttwDP$ on $H$-topological-minor-free graph classes.

\medskip 
From a broader perspective, our result provides a new and conceptually clean algorithmic meta-theorem principle for fragments of \CMSO beyond the classical bounded-treewidth regime.
It shows that structural restrictions on set quantification, formulated via minor-monotone annotated graph parameters, lead to robust and tractable logics on large graph classes, in particular, on classes excluding a fixed topological minor.
This unifies and subsumes a number of previously disparate formalisms, including \FODP\ and compound logic, within a single, logically natural framework.
We view this as evidence that restricting \CMSO quantification by elegant structural parameters is a powerful and flexible methodology for designing expressive and yet tractable logics, and that the search for suitable parameters is a fruitful direction for future research on algorithmic meta-theorems.


\subparagraph{The use of the disjoint-paths predicate.}
As a final note, we would like to discuss the presence of the disjoint-paths predicate in our logic. This is motivated principally by two reasons. The first one is related to previous work: as we mentioned, part of our motivation is to generalize the meta-theorems of~\cite{SchirrmacherSSTV24mode,fominGSST25compound} and provide a unified way of expressing all problems captured by these meta-theorems. For this, and in order to compare our logic with \FODP, $\tilde{\Theta}^{\DP}$, (and even $\CMSOatwDP$), it seems natural to include the disjoint-paths predicate as a feature of our logical framework. From a more technical point of view, the disjoint-paths predicate seems to play an essential role in the proof of compositionality of our logic (\cref{lem_compositionality}) and, therefore, in our model checking algorithm. In fact, even for the restricted fragment $\CMSOttw$, our approach for model checking seems to need the stronger assumption of working with $\CMSOttwDP$ to go through. Namely, for $\CMSOttw$ it seems essential to compute type-preserving representatives for the richer fragment \CMSOttwDP, because only these encode the information needed to show correctness of our replacement step. These are a couple of reasons why the disjoint-paths predicate seems to be an integral part of our approach. However, it is quite natural to ask whether its use comes with the cost of imposing hardness of model checking on monotone classes that do not exclude a topological minor. We answer in~\cref{sec_hardness} that this is not the case: hardness also holds even in the absence of the disjoint-paths predicate.

\subparagraph{Organization of the paper.}
After the introduction of needed notation for the technical results in \cref{sec_prelims}, \cref{sec_param} and \ref{sec_collapse} are devoted to further discussions on the graph parameters and their collapse on unbreakable graphs.
Then in \cref{sec_prelims_folio}, \ref{sec_compositionality}, and \ref{sec_represent}, we are providing key ingredients for our model checking algorithm, presented in~\cref{sec_mc}.
Before concluding the paper, we finally discuss the expressive power of our logic in \cref{sec_power} and the hardness of the model checking problem on classes admitting efficient encoding of topological minors in \cref{sec_hardness}. We conclude the article in \cref{sec_conclusions} with some directions for further research.

\section{Preliminaries}
\label{sec_prelims}

\subparagraph{Sets and integers.} We denote by $\mathbb{N}$ the set of non-negative integers.
For an integer $p \geq 1,$ we set $[p] = \{1,\ldots,p\}$ and $\mathbb{N}_{\geq p} = \mathbb{N} \setminus [0, p - 1].$
Also, given a function $f\colon A\to B$, we always consider its extension $f\colon 2^A\to B$ such that for every $X\subseteq A,$ $f(X)=\{f(a)\mid a \in X\}.$
Given two functions $\chi,\psi\colon \mathbb{N}\rightarrow \mathbb{N},$ 
some $r\in\Nbbb_{≥1}$, and a sequence $x_{1},\ldots,x_{r}$ of variables, 
we write $\chi(n)=\mathcal{O}_{x_{1},\ldots,x_{r}}(\psi(n))$ in order to denote that there exists a  function $f\colon\mathbb{N}^{r} \rightarrow \mathbb{N}$
such that $\chi(n)=\mathcal{O}(f(x_{1},\ldots,x_{r})\cdot \psi(n)).$ 

\subsection{Graphs and graph parameters}
\label{subsec_prelim_graphs}

\subparagraph{Graphs and annotated graphs.}
All graphs in this paper are finite, undirected graphs without loops and without multi-edges.
We do not distinguish between isomorphic graphs and write $\Gcal_{\mathsf{all}}$ for the class of all graphs.  
We write $V(G)$ for the vertex set and $E(G)$ for the edge set of a graph $G$. We write $\|G\|$ for $|V(G)|+|E(G)|$. 
For a vertex subset $X\subseteq V(G)$, we write $G[X]$ for the subgraph of $G$ induced by $X$. 


An \emph{annotated graph} is a pair $(G,X),$ where $X\subseteq V(G)$. 
Two annotated graphs $(G,X)$ and $(G',X')$ are isomorphic 
if there is an isomorphism $\theta\colon  V(G)\to V(G')$ from $G$ to $G'$ such that 
$\theta(X)=X'$.
An annotated graph $(H,Y)$ is an \emph{annotated subgraph}
of an annotated graph~$(G,X)$ if $H$ is a subgraph of $G$
and  $Y\subseteq X$.
We  denote by~$\Acal_\mathsf{all}$ the class of all annotated graphs.

Throughout the paper, when we speak about annotated graphs and formulas (from any logic), we have a signature extended by a unary predicate encoding the annotated set. We will not make the signature explicit every time but assume it tacitly.

\subparagraph{Connectivity and disjoint paths.}
Let $G$ be a graph and $u,v\in V(G)$. 
A $u$-$v$-path~$P$ in $G$ is a sequence $v_1,\dots, v_\ell$ of pairwise different vertices such that $\{v_i,v_{i+1}\}\in E(G)$ for all $1\leq i<\ell$ and $v_1=u$ and $v_\ell=v$. 
The vertices $v_2,\ldots, v_{\ell-1}$ are the \emph{internal vertices} of $P$ and the vertices~$u$ and~$v$ are its \emph{endpoints}. 
Two vertices~$u,v$ are \emph{connected} if there exists a path with endpoints $u,v$. 
A graph is \emph{connected} if every two of its vertices are connected and a set $X\subseteq V(G)$ is \emph{connected} in $G$ if $G[X]$ is connected.
Two paths $P$, $Q$ are \emph{internally vertex-disjoint} if no vertex of one path appears as an internal vertex of the other path. 
They are \emph{disjoint} if their vertex sets are disjoint. 

Let $G$ be a graph, $r\in\mathbb{N}$, and $s_1, t_1,\ldots, s_r, t_r$ pairwise distinct vertices of $G$. 
Then $\DP_r(s_1, t_1,\ldots, s_r, t_r)$ is true in $G$ if and only if there are disjoint paths between~$s_i$ and $t_i$ for $1\leq i\leq r$.
We use $\DP$ instead of $\DP_r$ when $r$ is clear from the context.

\subparagraph{Minors and topological minors.}
Given a graph $G$ and some $S\subseteq V(G)$, we say that $S$ is a \emph{connected} set if $G[S]$ is a connected graph.
A graph $H$ is a \emph{minor} of a graph $G$ if $H$ can be obtained from a subgraph of~$G$ by contracting edges. 
A \emph{minor model} of $H$ in $G$ is collection $\Bcal=\{B_{v}\mid v\in V(H)\}$ of pairwise-disjoint 
connected subsets of  $V(G)$ where, 
  for every $\{u,v\}\in E(H)$, $B_{u}\cup B_{v}$ is connected. Each set $B_{v}$ is called the \emph{branch set} of $v$. 


Given two annotated  graphs $(H,Y)$ and $(G,X)$, we say that 
$(H,Y)$ is a \emph{minor} of~$(G,X)$, denoted 
by $(H,Y)≤(G,X)$,  
if there exists a minor model $\Bcal$ of $H$ in $G$ such that
for every $u \in Y$, we have $B_{u} \cap X \neq \emptyset$. 

The \emph{dissolution} of a vertex $v$ of degree two 
whose two neighbors are $x$ and $y$ 
is the operation of removing~$v$ and adding the edge $\{x,y\}$ (if it does not already exist).
A graph~$H$ is a \emph{topological minor} of a graph $G$ if~$H$ can be obtained from a subgraph of $G$ by dissolving edges. 

A {\em{topological minor model}} of a graph~$H$ in a graph $G$ is an 
injective mapping $\eta$ that maps vertices of $H$ to vertices of $G$ and edges of~$H$ to pairwise internally vertex-disjoint paths in~$G$ so that for every $\{u,v\}\in E(H)$ the path $\eta(\{u,v\})$ has the endpoints $\eta(u)$ and~$\eta(v)$.
The vertices in  $\{\eta(v)\mid v\in V(H)\}$ are called \emph{principal vertices} of the model in $G$.
A graph~$H$ is a \emph{topological minor} of $G$ if there is a topological minor model of $H$ in $G$.
We say that an annotated graph $(H,Y)$ is a \emph{topological minor} of 
an annotated graph $(G,X)$, denoted by $(H,Y)\pretp (G,X)$,
if there exists a topological minor model $\eta$ of $H$ in $G$ such that $\eta(Y)\subseteq X$. 

We call a graph $G$ {\em{$H$-minor-free}}, and {\em{$H$-topological-minor-free}}, respectively, if $H$ is not a minor, or topological minor of $G$. 
We call a class $\Cc$ of graphs {\em{(topological-)minor-free}} if there exists a graph~$H$ such that every member of $\Cc$ is $H$-(topological-)minor-free. 

Given an annotated graph  $(H,Y)$, we define $\ext(H,Y)$ as the 
set of all topological-minor minimal annotated graphs $(G,X)$ that contain $(H,Y)$ as an annotated minor.
We need the following easy observation.

\begin{observation}
\label{obs_topominormany}
Let $(H,Y)$ be an $n$-vertex annotated graph. Then, every graph in $\ext(H,Y)$ has at most~$\Ocal(n^2)$
vertices.

\end{observation}
\begin{proof}
Let $(H,Y)$ be an $n$-vertex annotated graph and let $(G,X)\in \ext(H,Y)$.
Thus, $(H,Y)$ is an annotated \emph{minor} of $(G,X)$ for some $X\subseteq V(G)$, and $(G,X)$ is
topological-minor minimal with this property. Fix a witnessing minor-model $\Bcal=\{B_{v}\mid v\in V(H)\}$.

We may assume that $(G,X)$ is annotated-subgraph minimal (that is $|V(G)|+|E(G)|+|X|$ is minimized) containing the model $\Bcal$.
Note that $V(G)$ is the union of all branch sets together with a set of edges between branch sets witnessing all edges of $H$. Also, again by minimality, for every $v\in Y$, exactly one vertex of~$B_v$ lies in $X$.
Furthermore, by topological-minor minimality, each branch set $B_v$ 
induces a tree~$T_v$ in~$G$.
An \emph{attachment} of $T_{v}$ is 
a vertex of $T_{v}$ that either has a (unique, according to the annotated-subgraph minimality) neighbor outside $B_v$ in $G$ or, in case $v\in Y$, 
is the  (unique, due to annotated-subgraph minimality) vertex of $B_v\cap X$ witnessing the annotation.
We use $A_{v}$ to denote the attachment vertices of~$v$
and observe that $|V(T_{v})|≤|V(H)|-1+1=|V(H)|$. 
Notice now that, by topological  minor minimality, 
every vertex of $T_{v}$ that is not an attachment 
should be an internal vertex of $T_{v}$ of degree at least three in~$T_{v}$. This in turn 
implies that $|V(T_{v})\setminus A_{v}|≤|A_{v}|$ which 
implies that $|V(T_{v})|≤2 \cdot |V(H)$.
We conclude that $|V(G)|≤|V(H)|\cdot \max\{|V(T_{v})|\mid v\in V(H)\}≤2\cdot |V(T_{v})|^2=2\cdot |V(H)|^2$, as required.
 \end{proof}

When any definition for annotated (topological) minors can be applied in the case of empty annotation, i.e., $X=\emptyset$,
then, for simplicity, we drop the term ``annotated''.






\subparagraph{Annotated graph parameters.}
An \emph{annotated graph parameter} is a function \mbox{$\p\colon \aall\rightarrow \N$},  mapping annotated graphs to non-negative integers.

\begin{observation}
\label{obs_mosubm}
Let   $\p\colon \aall\to\Nbbb$ be a minor-monotone annotated graph parameter
and let $(G,X)$ be an annotated graph. Then, for all $Y\subseteq X$, it 
holds that $\p(G,Y)≤\p(G,X)$.
\end{observation}

\subparagraph{Obstructions.}
Given a minor-monotone annotated graph parameter  $\p\colon \aall\to\Nbbb$,
we define~$\Ocal_{k}^{\p}$ as the set of all $\leq$-minimal annotated graphs 
where the value of $\p$ exceeds $k$.

\begin{lemma}
    \label{prop_minor_obs}
For every minor-monotone annotated graph parameter $\p\colon \aall\to\Nbbb$, there is a function 
$f_{\p}\colon \Nbbb\to\Nbbb$ and a collection of finite sets $\{\Ocal_{k}^{\p}\mid k\in\Nbbb\}$
such that, for every $k\in\Nbbb$,   $|\Ocal_{k}^{\p}|\leq f_{\p}(k)$
and for every annotated graph $(G,X)$, $\p(G,X)\leq k$ if and only if $(G,X)$ excludes all graphs in $\Ocal_{k}^{\p}$ as topological minors. 
\end{lemma}

\begin{proof}
Consider the class $\Gcal_{\p,k}=\{(G,X)\mid \p(G,X)\leq k\}$ and observe that, for every $k\in\Nbbb,$ 
$\Gcal_{\p,k}$ is 
closed under annotated minors.
We define $\Zcal_{k}^{\p}$ as the set  of the minor-minimal annotated graphs that do not belong to $\Gcal_{\p,k}$ and observe that, for every $k$,
$\p(G,X)\leq k$ if and only if $(G,X)$ minor-excludes all annotated graphs in $\Zcal_{k}^{\p}$.
As observed in~\cite[Theorem~1.1]{ProtopapasTW2025colorfulminors}, based on the results of Robertson and Seymour in \cite{RobertsonS10GraphminorsXXIII}, the set 
$\Zcal_{k}^{\p}$
is finite, and therefore, there is a function 
$g_{\p}\colon \Nbbb\to\Nbbb$ such that 
every annotated graph in~$\Zcal_{k}^{\p}$ has at most $g_{\p}(k)$ vertices.
We now set $\Ocal_{k}^{\p}=\bigcup_{(Z,Y)\in\Zcal_{k}^{\p}}\ext(Z,Y)$.
By the definition of $\ext$,
we have that an annotated graph minor-excludes all graphs in $\Zcal_{k}^{\p}$ if and only 
if $(G,X)$ topological-minor-excludes all graphs in $\Ocal_{k}^{\p}$.
From \cref{obs_topominormany}, each of its annotated graphs contains  
$\Ocal((g_{\p}(k))^2)$ vertices. Therefore, 
there is some function 
$f_{\p}\colon \Nbbb\to\Nbbb$
as the one required for  
the lemma to hold.
\end{proof}

\begin{lemma}
\label{lemm_anotopomin_formula}
For every minor-monotone annotated graph parameter  $\p\colon \aall\to\Nbbb$ and 
every $k\in\Nbbb$, there is a formula $\sigma_{k}^{\p}\in\textup{\FODP}$
such that $(G,X)\models\sigma_{k}^{\p}$ iff $\p(G,X)\leq k$.
\end{lemma}

\begin{proof}
Because of \cref{prop_minor_obs}, 
$\p(G,X)\leq k$ iff $\forall (H,Y)\in\Ocal_{k}^{\p}$, $(H,Y)\not\pretp (G,X)$.
Therefore, it remains to express $(H,Y)\pretp (G,X)$
using some $\FODP$ formula $\phi_{H,Y}$. 
The existence of a topological minor model $\eta\colon V(H)\to V(G)$ 
of $(H,Y)$ in $(G,X)$ can be expressed by quantifying existentially 
on the $|V(H)|$ distinct images of $\eta$  where $\forall x\in Y, \eta(x)\in X$
and then asking whether 
there is some edge set  $\{\{x_1,y_1\},\ldots,\{x_{r},y_r\}\}\subseteq E(H)$ such that 
\begin{itemize}
\item for every edge $\{x,y\}\in E(G)\setminus E$, it holds that $\{\eta(x),\eta(y)\}\in E(G)$,
and 
\item for every  $i\in [r]$,  there are vertices $s_i,t_i\in V(G)$ where 
$\{\eta(x_i),s_i\},\{\eta(y_i),t_i\}\in E(G)$ and, moreover,  there are vertex-disjoint paths 
between $s_i$ and $t_i$, for $i\in[r]$, avoiding the vertices in $\eta(V(H))$.
\end{itemize} 
The above can be expressed quantifying (existentially) on 
$\Ocal(|V(H)|)$ vertices of $G$ and using the  disjoint-paths predicate.
For this, in order to avoid the vertices in $\eta(V(H))$, 
we additionally ask for (trivial) disjoint paths from each $\eta(x), x\in V(H)$ to itself. 
Therefore, there is a $\phi_{H,Y}\in \FODP$ expressing
$(H,Y)\pretp (G,X)$ as required.
\end{proof}

%
%
%
%
%
%

\subparagraph{Equivalent (annotated) parameters.}
Let $\p$ and $\p'$ be two annotated graph parameters. We say that $\p\pretp \p'$ 
if there is a function $f \colon \mathbb{N} \to \mathbb{N}$ such that, for every annotated graph $(G,X)$, it holds that $\p(G,X) \leq  f(\p'(G,X))$.  We say that $\p$ and $\p'$ are \emph{functionally equivalent}, or simply \emph{equivalent}, which we denote by $\p \sim \p'$, if   $\p\pretp \p'$ and  $\p'\pretp \p$.
The parameter ordering relation $\pretp$
and the concept of annotated parameter equivalence 
extend to non-annotated parameters in the obvious way.

\subsection{Tree decompositions and unbreakability}

\subparagraph{Trees and tree orders.}
A \emph{tree} is an acyclic and connected graph $T$. 
By assigning a distinguished vertex $r$ as the root of a tree, we impose a tree order $\pretp_T$ on $V(T)$ by $x\pretp_T y$ if~$x$ lies on the unique path (possibly of length~$0$) from $y$ to $r$. If $x\pretp_T y$, we call $x$ an \emph{ancestor} of $y$ in $T$. Note that by this definition, every node is an ancestor of itself. We drop the subscript $T$ if it is clear from the context. 
We write $\parent(x)$ for the parent of a non-root node $x$ of $T$, and $\children(x)$ for the set of children of $x$ in~$T$. We define $\parent(r)=\bot$.

\subparagraph{Tree decompositions.}
A {\em{tree decomposition}} of a graph $G$ is a pair $\Tt=(T,\bag)$, where~$T$ is a rooted tree and $\bag\colon V(T)\to 2^{V(G)}$ is a mapping assigning to each node $x$ of $T$ its {\em{bag}} $\bag(x)$, which is a subset of vertices of $G$ such that the following conditions are satisfied:
\begin{enumerate}
  \item For every vertex $u\in V(G)$, the set of nodes $x\in V(T)$ with $u\in \bag(x)$ induces a nonempty and connected subtree of $T$.
  \item For every edge $\{u,v\}\in E(G)$, there exists a node $x\in V(T)$ with $\{u,v\}\subseteq \bag(x)$.
\end{enumerate}

The \emph{width} of a tree decomposition is the value $\max_{t\in V(T)}|\mathsf{bag}(t)|-1$ and the \emph{treewidth} of a graph $G$ is the minimum possible width of a tree decomposition of $G$. 

\medskip
Recall that if $r$ is the root of $T$, then we have $\parent(r)=\bot$. We define $\bag(\bot)=\emptyset$.
%
For a node $x\in V(T)$, we define
the {\em{adhesion}} of $x$ as
  $\adh(x)\coloneqq \bag(\parent(x))\cap \bag(x);$
  the \emph{margin} of $x$ as $\mrg(x)\coloneqq \bag(x)\setminus \adh(x)$;
  the {\em{cone at $x$}} as
  $\cone(x)\coloneqq \bigcup_{y\succeq_T x} \bag(y);$
  and the {\em{component at $x$}} as
  $\cmp(x)\coloneqq \cone(x)\setminus \adh(x).$ 
  We will see in a moment that we may assume that the component at $x$ is connected for every node $x$, which justifies the name component. 

\smallskip
The {\em{adhesion}} of a tree decomposition $\Tt=(T,\bag)$ is defined as the largest size of an adhesion, that is, $\max_{x\in V(T)}|\adh(x)|$.
Given a tree decomposition $\mathcal{T}=(T,\mathsf{bag})$, we refer to~$T$ as the \emph{underlying tree} of $\mathcal{T}$.

\smallskip
A tree decomposition $\Tt=(T,\bag)$ is {\em{regular}} if for every non-root node $x\in V(T)$ 
\begin{enumerate}
  \item\label{reg_mrg} the margin $\mrg(x)$ is nonempty;
  \item\label{reg_con} the graph $G[\cmp(x)]$ is connected; and
  \item\label{reg_nei} every vertex of $\adh(x)$ has a neighbor in $\cmp(x)$.
\end{enumerate}

\subparagraph{Unbreakability.}
A {\em{separation}} in a graph $G$ is a pair $(A,B)$ of subsets of vertices of $G$ such that $A\cup B=V(G)$ and there is no edge with one endpoint in \mbox{$A\setminus B$} and the other endpoint in $B\setminus A$.
The {\em{separator}} of $G$ associated with $(A,B)$ is the intersection \mbox{$A\cap B$} and the {\em{order}} of a separation is the size of its separator.

For $q,k\in \N$, a vertex subset $X$ in a graph $G$ is {\em{$(q,k)$-unbreakable}} if for every separation $(A,B)$ of~$G$ of order at most $k$, we have
\[|A\cap X|\leq q\qquad\textrm{or}\qquad |B\cap X|\leq q.\]
We say that a graph $G$ is \emph{$(q,k)$-unbreakable} if $V(G)$ is $(q,k)$-unbreakable in $G$.

The notion of unbreakability can be lifted to tree decompositions by requiring it from every individual bag. In~\cite{cygan2019minimum}, Cygan,  Lokshtanov,  Pilipczuk,  Pilipczuk, and 
Saurabh introduce two notions of unbreakable tree decompositions, which differ on whether we impose each bag to be unbreakable in its cone, or in the entire graph. These definitions then yield different bounds and computation time. We use the notion named {\em strongly unbreakable}.

For $q,k\in \N$, a tree decomposition $\Tt=(T,\bag)$ of a graph $G$ is {\em{strongly $(q,k)$-unbreakable}}
if for every $x\in V(T)$, $\bag(x)$ is $(q,k)$-unbreakable in $G[\cone(x)]$. 

\begin{proposition}[\!\!\cite{cygan2019minimum}]\label{th_strong_unbreakability}
  There is a function $q(k)\in 2^{\Oh(k)}$ such that for every graph $G$ and $k\in \N$, there exists a strongly $(q(k),k)$-unbreakable tree decomposition of $G$ of adhesion at most $q(k)$. Moreover, given~$G$ and~$k$, such a tree decomposition can be computed in time $2^{\Oh(k^2)}\cdot |G|^2\cdot \|G\|$. 
\end{proposition}

Given any strongly $(q,k)$-unbreakable decomposition, we can refine it so that it becomes regular (and remains strongly unbreakable), see~\cite{bojanczyk2016definability}.  
Hence, we may assume that the tree decompositions constructed by the algorithm of~\cref{th_strong_unbreakability} are regular.

\subsection{Monadic second-order logic}
\label{subsec_msol_defs}

\subparagraph{Monadic second-order logic.}
We use \emph{counting monadic second-order logic} (\CMSO) with quantification over vertex sets, which extends monadic second-order logic (\MSO)
by modulo counting predicates on set cardinalities. 
For conceptual clarity, we work over a relational signature $\sigma$ consisting of a single binary relation symbol~$E$ (the edge relation) and a finite set of unary relation symbols, called \emph{colors}.
A \emph{colored graph} is a $\sigma$-structure such that~$E$ is interpreted as a symmetric and irreflexive relation on the universe, and each unary relation symbol is interpreted as a subset of elements of the universe.
Quantification over edge sets ($\CMSO_2$) or tuples of relations (as in guarded \CMSO) can be simulated for general relational signatures, by considering the incidence encoding. 
By considering the incidence encoding of structures, we can also apply the disjoint-paths predicate in the more general setting. 

In the context of graphs, \MSO and its counting extension \CMSO were systematically studied in the work of Courcelle and Engelfriet; a standard
reference is \cite{CourcelleEngelfriet2012}. 
The original application of \MSO to graphs appears in Courcelle’s classic paper \cite{Courcelle1990}. 
We assume basic knowledge of logic and refer the reader to the textbook~\cite{Libkin04} for more background.

\newcommand{\card}{\mathsf{card}}
$\CMSO$ is defined as follows. The variables of the logic are first-order variables $x,y,z,\ldots$ ranging over vertices, and set variables $X,Y,Z,\ldots$ ranging over subsets of vertices.
The formulas of $\CMSO$ are built from the atomic formulas $x = y$, $E(x,y)$, $C(x)$ for each unary relation symbol $C \in \sigma$, $x \in X$, $\card(X) \equiv i \pmod m$ for fixed integers $m \ge 1$ and $0 \le i < m$,
using the Boolean connectives $\neg, \wedge, \vee$ and the quantifiers
$\exists x, \forall x, \exists X, \forall X$.
We use $\forall X \,\varphi$ as an abbreviation for $\neg \exists X \,\neg \varphi$. 

The semantics are the standard Tarskian semantics for relational structures. In particular, 
set variables range over all subsets of the universe, and the counting atom
$\card(X) \equiv i \pmod m$ is true if and only if the cardinality of the interpretation of $X$
is congruent to $i$ modulo $m$.



\subparagraph{Syntax and semantics of $\CMSOp$.}
Let $\p\colon \aall\to\Nbbb$ be an annotated graph parameter.  
The terms, (atomic) formulas, and free variables of $\CMSOp$ are defined as the ones of $\CMSO$
with the only exception that instead of using the quantifier $\exists X$, where $X$ is a set variable,
we use the quantifier $\exists_{\p\leq k} X$ (semantically defined below) for some $k\in\mathbb{N}$.
Also, we use $\forall_{\p\leq k} X \varphi$ as an abbreviation of $\neg (\exists_{\p\leq k} X \neg \varphi)$.
The semantics of $\CMSOp$ are as the ones of $\CMSO$, where 
a colored graph $G$ satisfies
$\exists_{\p\leq k}X \varphi$, denoted by 
$G\models \exists_{\p\leq  k}X \varphi$, if and only if 
there is a vertex subset $S\subseteq V(G)$ such that
$\p(G,S)\leq k$ and
$G$ 
satisfies $\varphi$ when $X$ is interpreted as $S$. 
We define $|\phi|$ as the encoding length of $\phi$. In particular, this takes into account the values $k$ appearing in the quantifiers of $\phi$ as well as the values $i$ and $m$ appearing in atomic formulas of the form $\card(X) \equiv i \pmod m$ in $\phi$.

\subparagraph{The logic \texorpdfstring{$\CMSOpDP$}{CMSO/\p\!+\!dp}.}
For every $k\in \N$,
we define the $2k$-ary \emph{disjoint-paths predicate} $\DP(x_1, y_1, \ldots, x_k, y_k)$,
where $x_1, y_1, \ldots, x_k, y_k$ are first-order terms.
Although the following definition makes sense for any logic $\mathcal{L}$, we will only consider
it for extensions of $\CMSOp$, which is the focus of this paper.

A formula of $\CMSOpDP$
is obtained by enhancing the syntax of $\CMSOp$ by allowing atomic formulas of the form $\DP(x_1, y_1, \ldots, x_k, y_k)$ on first-order terms $x_1, y_1, \ldots, x_k, y_k$, for $k\geq 1$.
The satisfaction relation between a 
colored graph $G$ 
and formulas $\varphi$ of $\CMSOpDP$ is as for $\CMSOp$,
where $\DP(x_1, y_1, \ldots, x_k, y_k)$ evaluates true in 
$G$ if and only if
there are pairwise vertex disjoint paths
$P_1,\ldots, P_k$ in $G$ 
between (the interpretations of) $x_i$ and $y_i$ for all $i\in[k]$.

It is not difficult to see that $\DP(x_1,y_1,\ldots, x_k,y_k)$ is $\CMSO$-expressible for every fixed $k\in \N$. 
Suppose now that  $\p$ is minor-monotone.
Then because of \cref{lemm_anotopomin_formula}, the query $\p(G,X)\leq k$ can be expressed in $\textup{\FODP}$, 
and therefore also in  $\CMSO$. 
This implies that 
for every minor-monotone $\p$, the logic 
$\CMSOpDP$ is contained in $\CMSO$.

\begin{lemma}
\label{lem_fun_min_same_logic}
If $\p$ and $\p'$ are  equivalent annotated graph parameters that are both minor-monotone, then the logics \textup{\CMSOpDP} and \textup{\CMSOppDP} have the same expressive power.
\end{lemma}

\begin{proof}
Given that $\p$ and $\p'$ are equivalent, there is 
a function $f\colon \Nbbb\to\Nbbb$ such that for every $(G,X)$,
$\p'(G,X)\leq f(\p(G,X))$.
Also, from \cref{lemm_anotopomin_formula}, there exists a  formula 
$\sigma_{k}^{\p}\in\FODP$ such that $(G,S)\models\sigma_{k}^{\p}$ iff $\p(G,S)\leq k$.
 
Let $\phi\in\CMSOpDP$.
We transform $\phi$ to a formula $\phi'\in\CMSOppDP$ by replacing every instantiation of 
$\exists_{\p\le k} X\, \psi$ by $\exists_{\p'\le k'} X\, \sigma_{k}^{\p}\wedge \psi$
where $k'=f(k)$.
Notice now that if there is some $S$ where $\p(G,S)\leq k$, then 
$\p'(G,S)\leq k'$, therefore, for the same $S$, $\p'(G,S)\leq k'$
and $(G,S)\models \sigma_{k}^{\p}$.
On the other hand, if $\p'(G,S)\leq k'$
and $(G,S)\models \sigma_{k}^{\p}$, then 
obviously $\p(G,S)\leq k$. 
This implies that the above replacement 
gives a formula $\phi'\in\CMSOppDP$ that is equivalent to $\phi$.
\end{proof}

%
%

\subparagraph{Prenex normal form.}
In the rest of this paper, we assume that for every considered formula~$\varphi$, the maximum $m$ such that a predicate of the form $\mathsf{card}(X) \equiv i \pmod m$ appears in $\varphi$ is bounded by a fixed constant.

We say that a sentence $\varphi$ of $\CMSOpDP$ is in \emph{prenex normal form} if it is of the form
\[Q_1 X_1\ \ldots Q_\chi X_\chi\ Q_{\chi+1} x_{\chi+1}\ \ldots Q_{\chi+\beta}x_{\chi+\beta}\ \psi(X_1,\ldots,X_\chi,x_{\chi+1},\ldots,x_{\chi+\beta}),\]
where $\chi,\beta\in\mathbb{N}$ and
\begin{itemize}
    \item $X_1,\ldots,X_\chi$ are vertex set variables and $x_{\chi+1},\ldots,x_{\chi+\beta}$ are vertex variables;
    \item for each $i\in[\chi]$, $Q_i\in\{\exists_{\p\le w},\forall_{\p\le w}\}$, for some $w\in\mathbb{N}$;
    \item for each $i\in[\chi+1,\chi+\beta]$, $Q_i\in\{\exists,\forall\}$; and
    \item  $\psi(X_1,\ldots,X_\chi,x_{\chi+1},\ldots,x_{\chi+\beta})$ is a quantifier-free formula of $\CMSOpDP$.
\end{itemize}

For a given $i\in[\chi+\beta]$, the \emph{$i$th variable of $\varphi$} is $X_i$ if $i\le \chi$ and $x_i$ if $i\ge \chi+1$.
We refer to $\chi+\beta$ as the \emph{quantifier rank} of $\varphi$.
We assume that every given sentence of $\CMSObogDP$ is in prenex normal form.
Let $\varphi$ be a formula of $\CMSOpDP$.
The \emph{$\DP$-rank} of $\varphi$ is the maximum integer $r$ such that the predicate $\DP_r$ appears in $\phi$, or $0$ if there is no $\DP$ predicate in $\phi$.
Also, the \emph{$\p$-rank} of $\varphi$ is the maximum integer $w$ such that $\exists_{\p\le w}$ or $\forall_{\p\le w}$ appears in $\varphi$, or $0$ if no such quantifier is used in $\phi$.

\section{Torso treewidth and equivalent parameters}
\label{sec_param}

In this section, we study the annotated parameter \emph{torso treewidth}. 
We prove that it is equivalent to two other parameters, a fact that will allow us to choose the most convenient parameter in our proofs. 

\subparagraph{Torso treewidth.}
Recall that, given an annotated graph $(G,X)$, the \emph{torso}
of $X$ in $G$ is the graph obtained from $G[X]$ by adding, for every connected component $C$ of $G-X$, all edges that make the neighbors of~$V(C)$ in $X$ a clique.
We define the \emph{torso treewidth} of an annotated graph $(G,X)$
 as   \[\ttw(G,X) = \min\{\tw(\torso(G,X'))\mid X\subseteq X'\}.\]

The minimization over all supersets $X'\supseteq X$ in the definition of $\ttw$ is crucial.
Intuitively, $X$ is intended to represent the ``interface'' to the remainder of the graph, and we are allowed to enlarge it by adding additional vertices so as to obtain a torso of small treewidth. 
Formally, this minimization is needed to ensure that $\ttw$ is robust and, in
particular, minor-monotone. 
As a simple example, let $G=K_{1,k}$ and let $X$ be the set of leaves. 
Then, $\torso(G,X)$ is a clique on $k$ vertices, hence $\tw(\torso(G,X))=k-1$. However, the intended ``complexity witnessed by~$X$'' should not depend on
whether we include the unique high-degree vertex $c$ in the interface.  Setting $X':=X\cup\{c\}$, we
obtain $\torso(G,X')=K_{1,k}$ and thus, $\tw(\torso(G,X'))=1$. 
Our definition captures this robustness by allowing us to add vertices to the interface when this simplifies the
structure of the torso, and it is this flexibility that makes $\ttw$ behave well under taking minors.

\begin{observation}\label{obs_monotone_ttw}
    $\ttw$ is a minor-monotone annotated graph parameter.
\end{observation}
\begin{proof}
Let $(H,Y)\le (G,X)$, witnessed by a minor model $\eta$ of $H$ in $G$ such that
$\eta(y)\cap X\neq \emptyset$ for all $y\in Y$.
Choose a set $X^\star\supseteq X$ such that $\ttw(G,X) = \tw(\torso(G,X^\star))$.
We will construct a superset $Y^\star\supseteq Y$ such that
$\tw(\torso(H,Y^\star)) \le \tw(\torso(G,X^\star))$,
which then implies $\ttw(H,Y)\le \ttw(G,X)$ by the definition of $\ttw$ as a minimum over supersets.

\smallskip
For each $v\in V(H)$, if $\eta(v)\cap X^\star\neq\emptyset$ choose a vertex $x_v\in \eta(v)\cap X^\star$,
otherwise pick an arbitrary vertex $x_v\in \eta(v)$.
Now, let $Y^\star := \{v\in V(H)\mid x_v\in X^\star\}$. 
Then, $Y\subseteq Y^\star$ because for each $y\in Y$, we have $\eta(y)\cap X\neq\emptyset$ and
$X\subseteq X^\star$, hence $\eta(y)\cap X^\star\neq\emptyset$ and thus $y\in Y^\star$.

\smallskip
We now show that $\torso(H,Y^\star)$ is a minor of $\torso(G,X^\star)$.
We define a minor model~$\mu$ of $\torso(H,Y^\star)$ in $\torso(G,X^\star)$ as follows.
For each vertex $v\in V(H)$, set
$\mu(v) := \eta(v)\cap X^\star$.
By construction, if $v\in Y^\star$, then $\mu(v)\neq\emptyset$.
Let $\{u,v\}\in E(H)$.
Since $\eta$ is a minor model of~$H$ in $G$, there is an edge $\{a,b\}\in E(G)$ with
$a\in \eta(u)$ and $b\in \eta(v)$. Consider the vertices~$x_u$ and~$x_v$ chosen above.
If $x_u\in \eta(u)\cap X^\star$ and $x_v\in \eta(v)\cap X^\star$ then, within $G$, there are paths
inside the connected sets $\eta(u)$ and $\eta(v)$ linking $a$ to $x_u$ and $b$ to $x_v$.
If these paths use vertices outside $X^\star$, then they traverse some components of $G-X^\star$.
By the definition of the torso, for every component $C$ of $G-X^\star$ the set $N_G(C)\cap X^\star$
is made a clique in $\torso(G,X^\star)$.
Hence, we can shortcut each traversal through $C$ by an edge between the corresponding vertices of
$N_G(C)\cap X^\star$ in the torso. In this way, the adjacency between~$\eta(u)$ and $\eta(v)$ witnessed
by $\{a,b\}$ yields an edge in $\torso(G,X^\star)$ between $\eta(u)\cap X^\star$ and $\eta(v)\cap X^\star$,
and therefore between the branch sets $\mu(u)$ and $\mu(v)$.
Hence, $\mu$ is a minor model of $\torso(H,Y^\star)$ in $\torso(G,X^\star)$, and thus 
$\tw(\torso(H,Y^\star)) \le \tw(\torso(G,X^\star))$.

By definition of $\ttw$,
$\ttw(H,Y) \le \tw(\torso(H,Y^\star))
\le \tw(\torso(G,X^\star))
= \ttw(G,X)$,
as claimed.
\end{proof}

The definition of torso treewidth can alternatively
be derived by the  concept of $X$-free decompositions of~$G$ recently proposed by Jansen and Swennenhuis in \cite{JansenS2024Steiner}.
 Also, the same annotated graph parameter 
has been introduced and examined by Hodor, La, Micek, and Rambaud in~\cite{HodorLMR24quick}.

\subparagraph{Bidimensionality and attachment degree.} 

We denote by $\mathbf{\Gamma}^{\mathsf{\bullet}}_{k}=(\Gamma_{k},V_k)$ the \emph{rainbow $(k \times k)$-grid}, where~$\Gamma_{k}$ is the Cartesian product of $P_{k}$ with itself and 
$V_{k}$ is the vertex set of $\Gamma_{k}$. 
We define the \emph{bidimensionality} of a set $X$ in a graph $G$ 
as the maximum $k$ such that $\mathbf{\Gamma}^{\mathsf{\bullet}}_{k}\leq (G,X)$.
We write ${\brg}(G,X)$ for the bidimensionality of $X$ in $G$ ($\brg$ stands for biggest rainbow grid). 
The following observation is immediate by the definition. 

\begin{observation}\label{obs_brg_min_monotone}
    $\brg$ is a minor-monotone annotated graph parameter. 
\end{observation}

By $\cc(H)$, we denote the connected components of a graph $H$. 
The \emph{attachment degree} of an annotated graph $(G,X)$  is defined as
\begin{eqnarray}
\adeg(G,X) & = & \max\{|N_{G}(V(C))| \mid C\in\cc(G-X)\}.\label{eq_mortrorso}
\end{eqnarray}

We define the annotated graph parameter $\rtw\colon\aall\to\Nbbb$ as 
\begin{eqnarray}
\rtw(G,X) &  =  & \min\{\max\{\brg(G,X'),\adeg(G,X')\}\mid X\subseteq X'\}.\label{eq_aossrtwso}
\end{eqnarray}

Again, we take the minimum over all supersets $X'\supseteq X$ in the above definition to obtain minor-monotonicity, which is easy to verify. 
\begin{observation}
\label{obs_closed_nm7u}
$\rtw$ is a minor-monotone annotated graph parameter.
\end{observation}
\begin{proof}
Let $(H,Y)\le (G,X)$, witnessed by a minor model $\eta$ of $H$ in $G$ such that
$\eta(y)\cap X\neq\emptyset$ for all $y\in Y$.
Fix any superset $X'\supseteq X$. Define a superset $Y'\supseteq Y$ by
\mbox{$Y' := \{ v\in V(H)\mid $} $\eta(v)\cap X'\neq\emptyset\}$.
Clearly $Y\subseteq Y'$. 
We show
$
\max\{\brg(H,Y'),\adeg(H,Y')\} \le $ $\max\{\brg(G,X'),$ $\adeg(G,X')\}.
$
This will imply the claim by taking minima over supersets:
\[
\rtw(H,Y) \le \max\{\brg(H,Y'),\adeg(H,Y')\}
 \le \max\{\brg(G,X'),\adeg(G,X')\},
\]
and then minimizing over all $X'\supseteq X$ yields $\rtw(H,Y)\le \rtw(G,X)$.

\smallskip
Since $\brg$ is minor-monotone by \cref{obs_brg_min_monotone}, and $(H,Y')\le (G,X')$ holds by construction of $Y'$
(the same model $\eta$ witnesses it),
we obtain $\brg(H,Y') \le \brg(G,X')$.

\smallskip
For $\adeg$, let $D$ be a connected component of $H-Y'$. 
Consider the vertex set 
$U := \bigcup_{v\in V(D)} \eta(v)\ \subseteq V(G)$.
Because $D$ is connected in $H$ and each branch set $\eta(v)$ is connected in $G$, the set $U$ is connected in $G$.
Moreover, $U$ is disjoint from $X'$, since $D\cap Y'=\emptyset$ implies $\eta(v)\cap X'=\emptyset$ for all $v\in V(D)$.
Hence, $U$ is contained in some connected component $C$ of $G-X'$.
We claim that
\begin{equation}\label{eq_adeg_bound}
N_H(V(D)) \subseteq \{\, u\in Y' \mid \eta(u)\cap N_G(V(C))\neq\emptyset\},
\end{equation}
and therefore $|N_H(V(D))|\le |N_G(V(C))| \le \adeg(G,X')$.

To prove \eqref{eq_adeg_bound}, take any vertex $u\in N_H(V(D))$.
Then, there exists an edge $\{u,v\}\in E(H)$ with $v\in V(D)$. 
Let $a\in \eta(u)$ and $b\in \eta(v)$ such that $\{a,b\}\in E(G)$.
Since $b\in \eta(v)\subseteq U\subseteq V(C)$, we have $a\in N_G(V(C))$.
Also, $u\notin V(D)$ implies $u\in Y'$ (because $V(H)\setminus Y'$ is exactly the vertex set of $H-Y'$),
so $u\in Y'$ and $\eta(u)$ meets $N_G(V(C))$, proving~\eqref{eq_adeg_bound}.

As $D$ was an arbitrary component of $H-Y'$, we conclude
\[
\adeg(H,Y')\ =\ \max\{|N_H(V(D))|\mid D\in \cc(H-Y')\}\ \le\ \adeg(G,X').
\]

\smallskip
Altogether, we get $\rtw(H,Y)\le \rtw(G,X)$, as desired.
\end{proof}


\subparagraph{Monodimensionality.}
We introduce a final annotated graph parameter, which we already mentioned in the introduction. 
The \emph{outer-annotated $(k \times k)$-grid} is $\mathbf{\Gamma}^{\circ}_{k}=(\Gamma_{k},P_k)$, where~$P_{k}$ are the vertices of the perimeter (that is, the outer boundary) of $\Gamma_{k}$.
We define the \emph{monodimensionality} of a set~$X$ in a graph $G$ 
as the maximum $k$ such that $\mathbf{\Gamma}^{\circ}_{k}\leq (G,X)$.
We write ${\bog}(G,X)$ for the monodimensionality of $X$ in $G$ ($\bog$ stands for biggest outer-annotated grid). 
Again, the following observation is immediate by definition. 

\begin{observation}
    $\bog$ is a minor-monotone annotated graph parameter.
\end{observation}

\subparagraph{Equivalence of the parameters.}
We now prove that the three defined parameters are equivalent. 
%
%
%
%
The definition of $\ttw$ directly implies the following.
\begin{observation}
 \label{obs_adeg_lop}
 For every annotated graph  $(G,X)$, we have $\adeg(G,X)\leq \ttw(G,X)+1$. 
\end{observation}

Using the main result of Marx, Seymour, and Wollan in~\cite{MarxSW2017Rooted}, Hodor, La, Micek, and Rambaud \cite{HodorLMR24quick}
proved the following (using the terminology introduced in \cref{sec_prelims}).

\begin{proposition}\label{prop_elm_equ_v}
$\ttw\sim \bog$.
\end{proposition}

 Let $z\in\{0,\ldots,\lceil k/2\rceil\}$.
 The \emph{$z$-trimmed subgrid} of $\Gamma_{k},$
 is the unique $(k-2z\times k-2z)$-subgrid of $\Gamma_{k}$ whose 
 vertices are all within distance at least $z$ from all the vertices of the perimeter of $\Gamma_{k}$. In our proofs, we  make use of the following result of Demaine, Fomin, Hajiaghayi, and Thilikos~\cite[Lemma~3.1]{DemaineFHT04bidim}.

\begin{proposition}[Demaine, Fomin, Hajiaghayi, and Thilikos~\cite{DemaineFHT04bidim}] \label{prop_subjection}
 If $S$ is a subset of the vertices in the $\lfloor \sqrt[4]{|S|} \rfloor$-trimmed subgrid of $\Gamma_{k},$ then $\brg(G,S)≥\lfloor \sqrt[4]{|S|} \rfloor$.
\end{proposition}

\begin{lemma}
\label{lem_low_bed_y}
$\rtw(\mathbf{\Gamma}^{\circ}_{k})=\Omega(k)$.
\end{lemma}

\begin{proof}
Recall that we use the notation $\mathbf{\Gamma}^{\circ}_{k}=(\Gamma_{k},P_{k})$ where $P_{k}$ is the perimeter of $\Gamma_{k}$. 
Let~$X'$ be a vertex set of $\Gamma_{k}$ where $P_{k}\subseteq X'$.  
Given some $r\in\{0,\ldots,\lceil k/2\rceil\}$, we denote by~$\Gamma_{k}^{\circ}$
the $\lfloor\sqrt[4]{r}\rfloor$-trimmed subgrid of $\Gamma_{k}$ and we set $X'_{r}=X'\cap V(\Gamma_{k}^{\circ})$.  If $|X_{r}'|≥\lfloor\sqrt[4]{r}\rfloor$, then consider a subset $S$ of   $X_{r}'$ of exactly $\lfloor\sqrt[4]{r}\rfloor$ vertices.
As  $S\subseteq X_{r}'\subseteq V(\Gamma_{k}^{\circ})$, \cref{obs_mosubm} and \cref{prop_subjection} imply that $\brg(\Gamma_{k},X')≥\brg(\Gamma_{k},X'_{r})≥\brg(\Gamma_{k},S)≥\lfloor\sqrt[4]{r}\rfloor$.
Otherwise, if $|X'_{r}|<\lfloor\sqrt[4]{r}\rfloor$, then 
at least $k-2\lfloor\sqrt[4]{r}\rfloor-\lfloor\sqrt[4]{r}\rfloor$ of the vertical 
paths of $\Gamma_{k}^{r}$ and at least $k-2\lfloor\sqrt[4]{r}\rfloor-\lfloor\sqrt[4]{r}\rfloor$ of the horizontal  
paths of $\Gamma_{k}^{r}$ avoid the vertices of $X'$, therefore 
their union, say $U$,  forms a connected subgraph of some  connected component, say $C$, of $\Gamma_{k}-X'$ where $N_{\Gamma_{k}}(V(C))≥4\cdot (k-3\lfloor\sqrt[4]{r}\rfloor)$.
Indeed, this is  because $V(U)\cap X'=\emptyset$ and every vertical/horizontal path 
in $\Gamma_{k}^{r}$ that belongs to $U$ is a subpath of a vertical/horizontal path 
in $\Gamma_{k}$ whose two endpoints are vertices of $P$, thus also vertices of $X'$. 
This implies that $\adeg(\Gamma_{k},X')≥4\cdot (k-3\lfloor\sqrt[4]{r}\rfloor)$. 
Therefore, for every $X', P\subseteq X'$ it holds that $\rtw(\Gamma_{k},X')\geq\min\{\max\{\lfloor\sqrt[4]{r}\rfloor,4\cdot (k-3\lfloor\sqrt[4]{r}\rfloor)\}\mid r\leq k\}=\Omega(k)$ and the lemma follows.
\end{proof}

We are now in position to also prove the following equivalence.
\begin{lemma}\label{lem_equivalence}
$\bog\sim\rtw$.
\end{lemma}

\begin{proof}
 To prove that $\bog\pretp \rtw$, consider an annotated  graph $(G,X)$, and let $k=\bog(G,X)$. Then, the definition of $\bog$ implies that  $\mathbf{\Gamma}^{\circ}_{k} \leq (G,X)$. By \cref{obs_closed_nm7u}, we get that 
$\rtw(\mathbf{\Gamma}^{\circ}_{k}) \leq \rtw(G,X)$, 
and \cref{lem_low_bed_y} implies  that there exists a positive constant~$c$ such that, assuming that $k$ is large enough,
\[c \cdot \bog(G,X) = c\cdot k \leq \rtw(\mathbf{\Gamma}^{\circ}_{k}) \leq \rtw(G,X),\]
proving that, indeed, $\bog\pretp \rtw$. 

For the other direction, from \cref{prop_elm_equ_v} we have that $\ttw\pretp \bog$, and therefore for proving $\rtw\pretp \bog$, it is enough 
to prove that $\rtw\pretp \ttw$. To this end, observe first that $\Gamma^{\circ}_{k}\leq \Gamma^{\bullet}_{k}$, therefore $\brg\pretp\bog$. Also from \cref{prop_elm_equ_v}, $\bog\pretp \ttw$.
We deduce that $\brg\pretp \ttw$, therefore there is a function 
 $f\colon \Nbbb\to\Nbbb$ such that for every annotated graph $(G,X')$, 
 $\brg(G,X')\leq f(\ttw(G,X'))$. Also, from \cref{obs_adeg_lop}, $\adeg(G,X')\leq \ttw(G,X')+1$.
Using the two last inequalities and the definition of $\rtw$, we obtain that $\rtw\pretp \ttw$, as required.
\end{proof}

\section{Collapse on unbreakable graphs}
\label{sec_collapse}


Let $\mathcal{L}_1$, $\mathcal{L}_2$ be two logics and let $\Cc$ be a class of graphs.
We say that $\mathcal{L}_1$ \emph{collapses to $\mathcal{L}_2$ on $\Cc$} if for every formula of $\mathcal{L}_1$, there is a formula of $\mathcal{L}_2$ such that the two formulas are equivalent on all graphs from $\Cc$.

In this section, we prove the collapse of \CMSOttwDP and \CMSOtw to \FODP and \FO respectively, on unbreakable graphs. The first ingredient is the following lemma, which states that on unbreakable graphs, only small sets can have bounded torso treewidth.

\begin{lemma}\label{lem_small_unbreakable}
    Let~\mbox{$k,q\in\mathbb{N}$} and let $G$ be a $(q,k+1)$-unbreakable graph with at least $3q$ vertices. Also, let $X$ be a vertex subset of $V(G)$ with $\ttw(G,X)\le k$. 
    Then, for every $X'\supseteq X$ with $\ttw(G,X)=\mathsf{tw}(\mathsf{torso}(G,X'))$, it holds that $|X'|\le q$.
\end{lemma}
\begin{proof}
    Let $q,k$ be integers, and $G$ a $(q,k+1)$-unbreakable graph with at least $3q$ vertices.
    Let $X\subseteq V(G)$ be a subset of vertices of $G$ with $\ttw(G,X)\leq k$, and $X'\supseteq X$ with $\tw(\torso(G,X'))\leq k$.
    Then, there is a tree decomposition $\mathcal{T}$ of $\torso(G,X')$ of width $k$ (i.e.,~with at most $k+1$ vertices per bag).
    
    Assume towards a contradiction that $|X'|>q$.
    Then, fix a connected component $C$ of $G-X'$. Since $N(C)\cap X'$ forms a clique in $\torso(G,X')$, we have that $S_C:= N(C)\cap X'$ contains at most $k+1$ many vertices. As the removal of $S_C$ separates $X'$ from $C\cup S_C$, and since $|X'|>q$, the unbreakability of $G$ implies that~$|C\cup S_C| \le q$.

    We can now construct a tree decomposition $\mathcal{T}'$ of $G$ of width at most $q-1$ whose underlying tree $T'$ is a copy of $T$ where, for every connected component $C$ of $G-X'$, we create a bag $C\cup S_C$, attached (as a leaf) to the bag of $\mathcal{T}$ containing the vertices of $S_C$.
    (There must be a bag containing the vertices $N(C)\cap X'$ because they form a clique in $\torso(G,X')$.)
    
    Then, there is an inner-bag of at most $k+1$ vertices whose removal splits the vertices of~$G$ in a balanced way such that both parts contain at least $|G|/3\ge q$ vertices of $G$.
    (Remark that this is a property of tree decompositions.)
    Hence, there is a separation of $G$ of order~$k+1$ with at least $q$ vertices on both parts.
    This is a contradiction to the assumption that~$G$ is $(q,k+1)$-unbreakable.
\end{proof}

As proved in \cref{lemm_anotopomin_formula}, the fact that an annotated graph has small torso treewidth can be expressed in \FODP. Instead of the annotated graph $(G,X)$, we talk here about the graph~$G$ together with free variables $x_1,\ldots,x_\ell$, which is only as powerful when the set $X$ is known to have small size, see \cref{lem_small_unbreakable}.
\begin{lemma}\label{lem_ttw_FODP}
    For all integers $k,\ell$, there is an \textup{\FODP} formula $\sigma_{k}^{\textup{\FODP}}(x_1,\ldots,x_\ell)$ such that for every graph~$G$ and set $A = \{a_1,\ldots,a_\ell\}\subseteq V(G)$, we have that $(G,\bar a) \models \psi_k$ if and only if $\ttw(G,A)\le k$. 
\end{lemma}

Using the previous two lemmas, we can now show that on unbreakable graphs (with well-chosen parameters), $\CMSOtwDP$ collapses to $\FODP$.

\begin{lemma}\label{lem_collapse_to_FODP}
    There is a function $f_{\mathsf{rank}}\colon \mathbb{N}^3\to\mathbb{N}$ such that the following holds.
    For all~\mbox{$d,k,q\in\mathbb{N}$} and $\varphi$ a \textup{\CMSOtwDP}-sentence in prenex normal form that has quantifier rank~$d$ and $\ttw$-rank~$k$, there is an $\textup{\FODP}$-sentence $\psi$ of quantifier rank at most $f_{\mathsf{rank}}(d,k,q)$, such that for every $(q,k+1)$-unbreakable graph $G$, it holds that $G\models\varphi$ if and only if $G\models \psi$.
\end{lemma}
\begin{proof}
    Let $d, k, q$ be integers, and $G$ be a $(q,k+1)$-unbreakable graph. We can assume w.l.o.g.~that $G$ has at least $3q$ vertices, as the \FO formula that we now build can check that $|G|\ge 3q$ and on graphs of small size, one can hardcode in \FO that $(G,X)$ has small torso treewidth.

    For every given $\CMSOtwDP$-sentence $\varphi$ with quantifier rank~$d$ and $\ttw$-rank~$k$, there is an equivalent $\FODP$-sentence $\psi$ that is obtained from $\varphi$ by replacing each subformula of~$\varphi$ of the form $\exists_{\ttw\le k}X~\xi(X)$
    with a formula $\exists x_1,\ldots,x_q ~~\sigma_{k}^{\FODP}(x_1\ldots,x_q) \wedge \xi'(x_1,\ldots,x_q)$, where $\sigma_{k}^{\FODP}(x_1\ldots,x_q)$ is given by \cref{lem_ttw_FODP} and $\xi'(x_1,\ldots,x_q)$ is obtained from $\xi(X)$ by replacing each atomic formula $y\in X$ with the formula $\lor_{i\in[q]}y=x_i$.
    
    Thanks to \cref{lem_small_unbreakable}, for every unbreakable graph $G$ that satisfies the sentence\linebreak 
    $\exists_{\ttw\le k}X~\xi(X)$, the witness set $A$ must be of size at most $q$, and therefore, we have that\linebreak $G\models \exists x_1,\ldots,x_q ~~\sigma_{k}^{\FODP}(x_1\ldots,x_q) \wedge \xi'(x_1,\ldots,x_q)$ with assignments $a_i$ to $x_i$ where\linebreak $A= \{a_1,\ldots, a_q \}$.

    Conversely, if a graph $G$ (not necessarily unbreakable) satisfies\linebreak $\exists x_1,\ldots,x_q ~~\sigma_{k}^{\FODP}(x_1\ldots,x_q) \wedge \xi'(x_1,\ldots,x_q)$, then it also clearly satisfies $\exists_{\ttw\le k}X~\xi(X)$. Hence, the formulas are indeed equivalent.
\end{proof}

It is more involved, but we can actually prove a similar collapse from $\CMSOtw$ to $\FO$ on unbreakable graphs. The proof of \cref{lem_collapse_to_FODP} is the exact same, with the only exception that $\sigma^{\FODP}_k$ now needs to be written in \FO. Note that although \FODP does not collapse to \FO on unbreakable graphs, this can be achieved for these specific formulas.

\begin{lemma}\label{lem_ttw_FO}
    For every integers $q,k,\ell$, there is a \FO formula $\sigma_{k}^{\FO}(x_1,\ldots,x_\ell)$ such that for every $(q,k)$-unbreakable graph $G$, and set $A = \{a_1,\ldots,a_\ell\}\subseteq V(G)$, we have that $(G,\bar a) \models\sigma_{k}^{\FO}$ iff $\ttw(G,A)\leq k$.
\end{lemma}
\begin{proof}    
    Let $q,k,\ell$ be integers. 
    First note that there if an \FO sentence that checks whether a graph $G$ of size at most $\ell$ has treewidth at most $k$. Similarly, we can express in \FO that the graph induced by the assignments of free variables $x_1,\ldots,x_\ell$ has treewidth at most $k$. In what remains of the proof, we show that the edge relation of a torso is definable in the logic \FOconn (assuming attachment degree at most $k$, which we can also check in \FOconn). Since this logic collapses to \FO on unbreakable graphs (for well-chosen parameters), this will conclude the proof.

    In order to apply this collapse result of $\FOconn$, we need to give a more detailed definition of $\FOconn$. The logic $\FOconn$ extends $\FO$ by allowing the use of the \emph{connectivity predicate} $\mathsf{conn}_k(x,y,z_1,\ldots,z_k)$ which evaluates to true in a graph $G$ with some distinguished vertices $s,t$ and $w_1,\ldots,w_k$ (interpreting $x,y,z_1,\ldots,z_k$, respectively) if and only if there is a path from $s$ to $t$ in $G\setminus\{w_1,\ldots,w_k\}$.
    The logic $\FOconn_k$ allows the use of connectivity predicates $\mathsf{conn}_{i}(x,y,z_1,\ldots,z_i)$ for every $i\le k$. It is proven in \cite{PilipczukSSTV22} that for every integers $q,k\in\mathbb{N}$, every $\FOconn_k$ formula is equivalent to an \FO formula on every $(q,k)$-unbreakable graphs. 

    Note that there is an $\FOconn_k$ formula with free variables $x_1\ldots,x_\ell$ such that
    $(G,\bar{a}) \models \psi(\bar x)$ if and only if the attachment degree of $(G,\bar{a})$ is at most $k$. This is achieved by saying that for every $y$ other than $\bar x$ there is a subset $\bar z$ of $\bar x$ of size at most $k$, that separate $y$ from~$\bar x$.

    Also, there is $\FOconn_k$ formula with free variables $y_1,y_2,x_1\ldots,x_\ell$, such that $(G,b_1,b_2,\bar{a}) \models \psi(y_1,y_2,\bar x)$ if and only if $b_1$ and $b_2$ are among $\bar a$, the attachment degree of $(G,\bar{a})$ is at most~$k$, and $b_1,b_2$ are adjacent in $\torso(G,\bar a)$. This can be achieved by saying that $y_1$ and $y_2$ each have a neighbor $z_1,z_2$ not in $\bar x$ that are connected in $G-\{w_1,\ldots,w_k\}$ for every possible choice of $(w_i)_{i\le k}$ in $\bar x$.

    Thanks to the collapse of $\FOconn_k$ to \FO on $(q,k)$-unbreakable graphs~\cite{PilipczukSSTV22}, this last formula can be rewritten in \FO for $(q,k)$-unbreakable graphs.

    Finally combining everything, there is an \FO formula $\sigma_{k}^{\FO}(x_1,\ldots,x_\ell)$ that guesses the set~$X'$ of size at most $q$ extending $\bar x$, checks that it has small attachment degree, computes the edges of the torso, and finally, checks that the result has treewidth at most $k$.   
\end{proof}

Thanks to \cref{lem_ttw_FO}, we can now adapt the proof of \cref{lem_collapse_to_FODP} to yield the following result.
\begin{lemma}\label{lem_collapse_to_FO}
    There is a function $f_{\mathsf{rank}}\colon \mathbb{N}^3\to\mathbb{N}$ such that the following holds.
    For all~\mbox{$d,k,q\in\mathbb{N}$} and $\varphi$ a \textup{\CMSOtw}-sentence in prenex normal form that has quantifier rank~$d$ and $\ttw$-rank~$k$, 
    there is an $\textup{\FO}$-sentence~$\psi$ of quantifier rank at most $f_{\mathsf{rank}}(d,k,q)$, such that for every $(q,k+1)$-unbreakable graph $G$, it holds that $G\models\varphi$ if and only if $G\models \psi$.
\end{lemma}

\section{Preliminaries on folios}
\label{sec_prelims_folio}

In this section, we define  boundaried (annotated) graphs, which are (annotated) graphs with a distinguished set of vertices and a linear order among these vertices.
We also define when two annotated boundaried graphs are \textsl{compatible} and the notion of \textsl{folio} of an annotated boundaried graph. Finally, we prove \cref{lem_repl_compos}, which is the combinatorial 
engine of the proof of the compositionality of $\CMSOpDP$
in \cref{sec_compositionality}.

\subsection{Boundaried annotated graphs}
\label{sec_bounded_annotated_graphs}



An important combinatorial ingredient of our proofs is boundaried graphs. 
This subsection is dedicated to their introduction and to the extension of the containment relations that we gave in previous sections to the boundaried setting.

\subparagraph{Boundaried graphs.}
Let $t\in\mathbb{N}.$
An  \emph{annotated $t$-boundaried graph} is a quadruple $\bound{G} = (G,X,B,\rho)$ where~$(G,X$) is an annotated  graph, $B \subseteq V(G)$, $|B| = t,$ and
$\rho\colon B \to [t]$ is a bijective function (see~\cite[Definition 12.9.1]{DowneyF13fund}).
We call $B$ the \emph{boundary} of $\bound{G}$ and  the vertices of $B$ \emph{the boundary vertices} 
of~$\bound{G}$. 
 We also
call $(G,X)$ \emph{the underlying annotated graph} of~$\bound{G}$
and we call $G$ \emph{the underlying  graph} of $\bound{G}$.
The quadruple~$(G,X,B,\rho)$ is an  \emph{annotated boundaried graph} if it is an annotated  $t$-boundaried graph, for some~$t\in\mathbb{N}.$

We say that an annotated $t$-boundaried graph $\bound{G}'=(G',X',B',\rho')$ is an (annotated boundaried) \emph{subgraph}  of $\bound{G}$ if
$G'$ is a subgraph of $G$, $B'=B$, $B\cap X=B'\cap X'$, $X'\subseteq X\cap V(G')$,  and $\rho'=\rho$. 
Two $t$-boundaried graphs $\bound{G}_{1}=(G_{1},X_{1},B_{1},\rho_{1})$ and
$\bound{G}_{2}=(G_{2},X_2,B_{2},\rho_{2})$ are \emph{isomorphic}  if  $G_{1}$ is isomorphic to $G_{2}$ 
via a bijection $\phi\colon  V(G_{1})\to V(G_{2})$ such that $\rho_{1}=\rho_{2}\circ\phi|_{B_{1}} ,$ i.e.,~the vertices of $B_{1}$ are mapped via $\phi$ to equally indexed vertices  of $B_{2},$
and, moreover,  $\rho_{2}|_{X_{2}}=\rho_{2}\circ\phi|_{X_{1}}$, i.e.,~boundary vertices in $X$ are equally indexed in both~$\bound{G}_{1}$ and~$\bound{G}_{2}$.
As in~\cite{RobertsonS95XIII} (see also~\cite{BasteST20acom}), 
we define the \emph{detail} of an annotated  boundaried graph
$\bound{G} = (G,X,B,\rho)$ as  $\mathsf{detail}(\bound{G}):=\max\{|E(G)|,|V(G)\setminus B|\}.$
We denote by ${\mathcal{A}}^{(t)}$ the set of all (pairwise non-isomorphic)  annotated $t$-boundaried graphs.
We also set ${\mathcal{A}}=\bigcup_{t\in\mathbb{N}}{\mathcal{A}}^{(t)}.$

\subparagraph{Compatible annotated boundaried graphs.}
Let $\bound{G}_{1}=(G_1,X_{1},B_1,\rho_1)$ and $\bound{G}_{2}=(G_2,X_{2},$ $B_2,\rho_2)$ be two annotated boundaried graphs.
We say that $\bound{G}_{1}$ and $\bound{G}_2$ are \emph{compatible} if $\rho_{2}^{-1}\circ \rho_{1}$  is an isomorphism from $G_{1}[B_{1}]$ to $G_{2}[B_{2}]$
such that $\rho_{1}(X_{1})=\rho_{2}(X_{2})$.
Given two compatible annotated boundaried graphs $\bound{G}_{1}=(G_1,X_{1},B_1,\rho_1)$
and  $\bound{G}_{2}=(G_2,X_{2},B_2,\rho_2),$  we
define $\bound{G}_{1}\oplus\bound{G}_{2}$ as the annotated graph obtained
if we take the disjoint union of $G_{1}$ and $G_{2}$
and, for every $i\in[|B_{1}|],$ we identify the vertices $\rho_{1}^{-1}(i)$ and $\rho_{2}^{-1}(i)$, agreeing that 
after each such identification the vertices of the boundary of $\bound{G}_{1}$ prevail.

\medskip

In case $X=\emptyset$, instead of $\bound{G} = (G,X,B,\rho)$, for simplicity, we write  $\bound{G} = (G,B,\rho)$, and we call $\bound{G}$ a \emph{boundaried graph},
and all the aforementioned definitions naturally apply to boundaried graphs as well. We also use  
${\Bcal}^{(t)}$ for the set of all  (pairwise non-isomorphic) $t$-boundaried graphs and  we set $\Bcal=\bigcup_{t\in\mathbb{N}}{\mathcal{B}}^{(t)}$. 

\subparagraph{Topological minors of annotated boundaried graphs.}
If $\textbf{M}=(M,X,B,\rho)\in{\mathcal{A}}$ and   $T\subseteq V(M)$ with $B\subseteq T,$ and such that 
every vertex in $V(M)\setminus T$ has degree 2,
we  call  $(\textbf{M},T)$ an \emph{\abtm-pair}
and we  define  $\mathsf{diss}(\textbf{M},T)=(\mathsf{diss}(M, T),X\cap T,B,\rho),$ where $\mathsf{diss}(M, T)$ is the result of the dissolution in $M$ of all vertices in~$V(M)\setminus T$.
We refer to the set $T$ as the set of \emph{branch vertices} of $(\mathbf{M},T)$.
Note that we do not permit dissolution of boundary vertices, as we demand all of them to be branch vertices. 
If $\bound{G}=(G,X,B,\rho)$ is an annotated  boundaried graph, we say that $({\bf M},T)$ is an \emph{\abtmpair of $\bound{G}$} if it is an \abtm-pair  where $\bound{M}$ is a subgraph of $\bound{G}$ and $B\subseteq T$. 
Let  $\bound{G}_{1}=(G_1,X_{1},B_1,\rho_1)$ and $\bound{G}_{2}=(G_2,X_{2},B_2,\rho_2)$  be two annotated  boundaried graphs.
We say that $\textbf{G}_{1}$ is an \emph{annotated topological minor}
of~$\bound{G}_{2},$ denoted by $\bound{G}_{1}\pretp\bound{G}_{2},$ if $\bound{G}_{1}$ and~$\bound{G}_{2}$ are compatible and~$\bound{G}_{2}$ has an \abtm-pair $(\textbf{M},T)$
such that  $\mathsf{diss}(\textbf{M},T)$ is isomorphic to $\textbf{G}_{1}.$

\subsection{Folios}

We next present the definition of folios and extended folios of (annotated) boundaried graphs,
which encode the set of (annotated) topological minors they contain. These notions will be used
several times in the remainder of the paper. At the end of this subsection, we prove ~\cref{lem_repl_compos},
which (informally) states that the extended folio captures the partial behavior of any fixed
minor-monotone annotated graph parameter on the set $X$ of an annotated boundaried graph
$\bound{G}=(G,X,B,\rho)$.

\subparagraph{Folios and extended folios.}
Given $\bound{G}\in\mathcal{A}$ and $\ell\in\mathbb{N}$, we define the \emph{$\ell$-folio} of $\bound{G}$ as
\[
\ell\folio(\bound{G}) := \bigl\{\, \bound{G}'\in\mathcal{A}\ \bigm|\ \bound{G}'\pretp \bound{G}
\text{ and }\detail(\bound{G}')\le \ell \,\bigr\}.
\]

The number of distinct $\ell$-folios of annotated $t$-boundaried graphs is bounded in the following result, proved first in~\cite{BasteST20hittI} and used also in~\cite{BasteST20acom}.
\begin{proposition}[\!\!\cite{BasteST20hittI}]\label{prop_number_folios}
	There exists a function $f_{\ref{prop_number_folios}}\colon  \mathbb{N}^{2} \to \mathbb{N}$ such that for every $t,\ell\in \mathbb{N},$ it holds that, for every $\bound{G}\in\Acal^{(t)}$, 
    \[|\{\ell\folio(\bound{G}) \mid \bound{G}\in\Acal^{(t)}\}|\leq f_{\ref{prop_number_folios}}(t,\ell).\]
\end{proposition}


\subparagraph{Extended folios.}
Let $\bound{G} = (G,X,B,\rho)$ be an  annotated boundaried graph.
Given $I\subseteq \binom{[|B|]}{2}$, we write~$\bound{G}^I$ 
for the (annotated) boundaried graph obtained from $\bound{G}$ by adding 
in its underlying graph the edges in 
\[
\bigl\{\, \{\rho^{-1}(i),\rho^{-1}(j)\} \ \bigm|\ \{i,j\}\in I \,\bigr\}.
\]

We define the \emph{extended  $\ell$-folio of $\bound{G} = (G,X,B,\rho)$}
to be 
\begin{eqnarray}
\ell\extfolio(\mathbf{G})& = & \{(I,\ell\folio(\textbf{G}^I))\mid 
I\in \binom{[|B|]}{2}\}. 
\label{eq_ext_f}
\end{eqnarray}

From \cref{prop_number_folios}, and the definition in \eqref{eq_ext_f}, we easily obtain the following.

\begin{observation}
\label{obs_number_anns_folios}
	There exists a function $f_{\ref{obs_number_anns_folios}}\colon  \mathbb{N}^{2} \to \mathbb{N}$ such that for every $t,\ell\in \mathbb{N},$ and every $\bound{G}\in\Bcal^{(t)}$, it holds that
    \[|\ell\extfolio(\mathbf{G})|\leq f_{\ref{obs_number_anns_folios}}(t,\ell).\]
\end{observation}

%
We say that two (annotated) boundaried graphs $\bound{G}_{1},\bound{G}_{2}$  are \emph{$\ell$-equivalent} if they are compatible and 
\[\ell\extfolio(\mathbf{G}_1)=\ell\extfolio(\mathbf{G}_2).\]

After presenting the above definitions, we are now ready to state and prove the following result.

\begin{lemma}
\label{lem_repl_compos}
For every minor-monotone annotated graph parameter $\p\colon \aall\to\Nbbb$, there exists a function $f_{\ref{lem_repl_compos}}^{\p}\colon \Nbbb\to\Nbbb$ such that the following holds:
Let $w\in\Nbbb$.
If two annotated boundaried graphs $\mathbf{G}_{1}$, $\mathbf{G}_{2}$ are  $f_{\ref{lem_repl_compos}}^{\p}(w)$-equivalent, then for  every   compatible  annotated boundaried graph~$\bound{C}$, it holds that 
$$\p(\mathbf{C}\oplus\mathbf{G}_{1})\leq w\iff 
\p(\mathbf{C}\oplus\mathbf{G}_{2})\leq w.$$
\end{lemma}

\begin{proof}
Fix $w\in\N$.
By \cref{prop_minor_obs}, there exists a finite set $\Ocal_{w}^{\p}$ of annotated graphs such that, for every
annotated graph $(H,Y)$,
\[
\p(H,Y)\le w \quad\Longleftrightarrow\quad (H,Y)\ \text{excludes every } \mathbf{O}\in \Ocal_{w}^{\p}\ \text{as an annotated topological minor.}
\]

Let
\[
\ell \ :=\ \max\{\detail(\mathbf{O}) \mid \mathbf{O}\in \Ocal_{w}^{\p}\},
\qquad\text{and set}\qquad
f_{\ref{lem_repl_compos}}^{\p}(w):=\ell .
\]

We claim that for every compatible annotated boundaried graph $\bound{C}$ and every annotated graph
$\mathbf{O}$ with $\detail(\mathbf{O})\le \ell$,
\begin{equation}\label{eq_tm_preserved_by_extfolio}
\mathbf{O}\pretp (\bound{C}\oplus \mathbf{G}_1)
\quad\Longleftrightarrow\quad
\mathbf{O}\pretp (\bound{C}\oplus \mathbf{G}_2).
\end{equation}

To see this, assume $\mathbf{O}\pretp (\bound{C}\oplus \mathbf{G}_1)$ and fix an \abtm-pair
$(\mathbf{M},T)$ witnessing this, i.e.,~$\mathbf{M}$ is an annotated boundaried subgraph of
$\bound{C}\oplus \mathbf{G}_1$ (with the boundary contained in $T$) and
$\diss(\mathbf{M},T)\cong \mathbf{O}$.
Let $B:=B(\mathbf{G}_1)=B(\mathbf{G}_2)$ and let $I\subseteq \binom{[|B|]}{2}$ be the set of all pairs
$\{i,j\}$ such that $\mathbf{M}$ contains an $i$--$j$ path whose internal vertices lie in
$V(\bound{C})\setminus B$ (equivalently: $\mathbf{M}$ uses~$\bound{C}$ to connect the $i$th and $j$th boundary
vertices). Consider the induced annotated boundaried subgraph $\mathbf{M}_1$ of $\mathbf{G}_1^{I}$ obtained by
taking all vertices and edges of $\mathbf{M}$ that lie in $\mathbf{G}_1$, together with the boundary $B$
(and keeping the inherited annotation).
Then $\detail(\mathbf{M}_1)\le \detail(\mathbf{O})\le \ell$, and by construction,
\[
\mathbf{M}_1 \pretp \mathbf{G}_1^{I},
\qquad\text{hence}\qquad
\mathbf{M}_1 \in \ell\folio(\mathbf{G}_1^{I}).
\]

Since $\mathbf{G}_1$ and $\mathbf{G}_2$ are $\ell$-equivalent, we have
$\ell\extfolio(\mathbf{G}_1)=\ell\extfolio(\mathbf{G}_2)$, and therefore
$\ell\folio(\mathbf{G}_1^{I})=\ell\folio(\mathbf{G}_2^{I})$.
So $\mathbf{M}_1 \in \ell\folio(\mathbf{G}_2^{I})$, i.e.,~$\mathbf{M}_1 \pretp \mathbf{G}_2^{I}$.
Gluing the same part of~$\mathbf{M}$ that lies in $\bound{C}$ to a witnessing model of
$\mathbf{M}_1 \pretp \mathbf{G}_2^{I}$ (and using the edges encoded by $I$ exactly where $\bound{C}$ provides the
corresponding boundary-to-boundary connections) yields an \abtm-pair in $\bound{C}\oplus \mathbf{G}_2$ whose
dissolution is isomorphic to $\mathbf{O}$.
Hence, $\mathbf{O}\pretp (\bound{C}\oplus \mathbf{G}_2)$.
The reverse implication in \eqref{eq_tm_preserved_by_extfolio} is symmetric, proving the claim.

\smallskip
Now we conclude the lemma.
By \cref{prop_minor_obs}, for $i\in\{1,2\}$ we have
\[
\p(\bound{C}\oplus \mathbf{G}_i)\le w
\quad\Longleftrightarrow\quad
(\bound{C}\oplus \mathbf{G}_i)\ \text{excludes every }\mathbf{O}\in\Ocal_{w}^{\p}\ \text{as a topological minor.}
\]

But every $\mathbf{O}\in\Ocal_{w}^{\p}$ satisfies $\detail(\mathbf{O})\le \ell$, so by
\eqref{eq_tm_preserved_by_extfolio}, $\bound{C}\oplus \mathbf{G}_1$ excludes $\mathbf{O}$ iff
$\bound{C}\oplus \mathbf{G}_2$ excludes $\mathbf{O}$, for all $\mathbf{O}\in\Ocal_{w}^{\p}$.
Therefore $\p(\bound{C}\oplus \mathbf{G}_1)\le w$ iff $\p(\bound{C}\oplus \mathbf{G}_2)\le w$, as required.
\end{proof}


\section{Compositionality}
\label{sec_compositionality}

In this section, we start by defining the annotated types of our logic. The notion of annotated types is central in our work, as our model checking algorithm computes type-preserving representatives.
The main technical contribution of this section is the proof of a ``Feferman-Vaught style'' composition lemma (\cref{lem_compositionality}), which is the main ingredient of the proof of correctness of our algorithm.

As reflected in the statement of this composition lemma, the notion of (plain) annotated types is not enough. What we show is that, in order to preserve the annotated type of a graph when replacing some part (of small interface) with a type-preserving representative, one needs to consider a stronger notion of representative, namely a representative of the same \emph{extended type}. Extended types generalize types the same way that extended folios generalize folios and are also defined in this section.

Let us note that the employment of the extended type (instead of just the type) is required in order to retain control on the way that paths may cross the interface between the two glued graphs and avoid blow-ups in the involved constants. This phenomenon already appears in the compositionality of \FODP~\cite{SchirrmacherSSTV24mode} and is inspired by the corresponding statement for folios~\cite[Lemma 2.4]{grohe2011finding}. 

Our proof follows the standard approach for \MSO (see~\cite[Lemma 2.3]{Grohe2007LogicGA}). However, one needs to show that when gluing along a small interface both the disjoint-paths predicate and the set quantification of our logic composes nicely (i.e.,~without blowups). For the disjoint-paths predicate, this was proven in~\cite{SchirrmacherSSTV24mode}, while for the set quantification, we show it here using~\cref{lem_repl_compos}.

\subparagraph{Annotated types of our logic.}
Let $d,r,w\in\mathbb{N}$.
For a given colored graph $G$ and a sequence $\bar{R}=(R_1,\ldots,R_{d})$ of subsets of $V(G)$,
we define the \emph{annotated $(d,r,w)$-type} of $(G,\bar{R})$ with respect to $\CMSOttwDP$, denoted by $\type_{d,r,w}(G,\bar{R})$, to be the set of all (up to logical equivalence) sentences $\varphi$ of $\CMSOttwDP$ (over the appropriate colored graph signature) in prenex normal form that have quantifier rank at most $d$, $\DP$-rank at most $r$, and $\ttw$-rank at most $w$ and are satisfied in $G$ when the $i$th variable of $\varphi$ is interpreted in~$R_i$, i.e.,~if the $i$th variable is a set variable, then it should be interpreted as a subset of~$R_i$, while if it is a first-order variable it should be interpreted as an element of~$R_i$.
We remark that when we refer to the annotated type of an $\ell$-boundaried graph, we assume the existence of $\ell$ colors in our signature and a bijection between the boundary vertices and the color each belongs to.

\subparagraph{Extended annotated types.}
Recall that, given a boundaried graph $\bound{G}=(G,B,\rho)$ and $I\subseteq \binom{[|B|]}{2}$, we write $\bound{G}^I$ 
for the boundaried graph obtained from $\bound{G}$ by adding 
in its underlying graph the edges in 
$\{\{\rho^{-1}(i),\rho^{-1}(j)\}\mid \{i,j\}\in I\}.$

Let $d,r,w\in\mathbb{N}$. Given a boundaried graph $\bound{G}=(G,B,\rho)$ and a $d$-tuple $\bar{R}$ of subsets of $V(G)$, we define the \emph{extended annotated $(d,r,w)$-type} of $(\bound{G},\bar{R})$
to be
\begin{eqnarray}
\exttype_{d,r,w}(\mathbf{G},\bar{R})& = & \{(I,\type_{d,r,w}(\textbf{G}^I,\bar{R}))\mid 
I\in \binom{[|B|]}{2}\}. 
\label{eq_ext_type}
\end{eqnarray}

\begin{remark}[Sufficiency of rank bounds]\label{rem_rank_sufficiency}
In the following, whenever we speak of annotated $(d,r,w)$-types, we implicitly assume that the parameters $d$ and $r$ are chosen large enough as a function of $w$ and the boundary size $|B|$.
In particular, $d$ and $r$ are sufficient to encode all finite obstruction information required to test the constraint $\p(\cdot,\cdot)\le k$ for every $k\le w$ via extended folios, as guaranteed by \cref{lem_repl_compos}.
In other words, we assume that $d$ and $r$ are large enough, so that if two annotated boundaried graphs have the same extended annotated $(d,r,w)$-type, then they are $f_{\ref{lem_repl_compos}}^{\ttw}(w)$-equivalent, where $f_{\ref{lem_repl_compos}}^{\ttw}$ is the function of~\cref{lem_repl_compos} for $\p=\ttw$.
\end{remark}

Before proceeding to the proof of our ``Feferman-Vaught style'' composition lemma, we show how our definitions of compatibility and gluing generalize to boundaried graphs of the form $(\bound{G},\bar{R})$.
Let $\bound{G}_{1}=(G_1,B_1,\rho_1)$ and $\bound{G}_{2}=(G_2,B_2,\rho_2)$ be two boundaried graphs.
Also, for given $d\in\mathbb{N}$, let $\bar{R}_1:=(R_{1}^1,\ldots,R_{d}^1)$ and $\bar{R}_2:=(R_{1}^2,\ldots,R_d^2)$ be two $d$-tuples of subsets of~$V(G_1)$ and~$V(G_2)$, respectively. 
We say that $(\bound{G}_{1},\bar{R}_1)$ and $(\bound{G}_2,\bar{R}_2)$ are \emph{compatible} if $\rho_{2}^{-1}\circ \rho_{1}$ is an isomorphism from $G_{1}[B_{1}]$ to $G_{2}[B_{2}]$
such that for every $i\in [d]$, $\rho_{1}(R_{i}^1\cap B_1)=\rho_{2}(R_{i}^2\cap B_2)$.
Given that $(\bound{G}_{1},\bar{R}_1)$ and $(\bound{G}_2,\bar{R}_2)$ are compatible, we
define $(\bound{G}_{1},\bar{R}_1)\oplus(\bound{G}_{2},\bar{R}_2)$ as the pair $(G,\bar{R})$, where
\begin{itemize}
    \item $G$ is the graph obtained if we take the disjoint union of $G_{1}$ and $G_{2}$ and, for every $i\in[|B_{1}|],$ we identify the vertices $\rho_{1}^{-1}(i)$ and $\rho_{2}^{-1}(i)$, agreeing that after each such identification the vertices of the boundary of $\bound{G}_{1}$ prevail, and
    \item $\bar{R}$ is the $d$-tuple $(R_1,\ldots,R_d)$, where for each $i\in[d]$, $R_i$ is obtained by the disjoint union of $R_i^1$ and $R_i^2$ and, for every $i\in[|B_{1}|],$ if $\rho_{1}^{-1}(i)\in R_i^1$, we identify the vertices $\rho_{1}^{-1}(i)$ and $\rho_{2}^{-1}(i)$, agreeing that after each such identification the vertices of $R_i^1\cap B_1$ prevail. 
\end{itemize}
We are now ready to state and prove the following lemma.

\begin{lemma}\label{lem_compositionality}
    Let $d,r,w\in\mathbb{N}$, let $\mathbf{G}_{1}$, $\mathbf{G}_{2}$ be two boundaried graphs and let $\bar{R}_1$,$\bar{R}_2$ be two $d$-tuples of subsets of $V(G_1)$ and $V(G_2)$, respectively.
    If 
    \begin{itemize}
        \item $(\mathbf{G}_1,\bar{R}_1)$ and $(\mathbf{G}_2,\bar{R}_2)$ are compatible, and
        \item the extended annotated $(d,r,w)$-type of $(\bound{G}_1,\bar{R}_1)$ and $(\bound{G}_2,\bar{R}_2)$ is the same,
    \end{itemize}
    then for every boundaried graph $\bound{C}$ and every $d$-tuple $\bar{Q}$ of subsets of $V(C)$
    such that $(\bound{C},\bar{Q})$ is compatible with $(\bound{G}_1,\bar{R}_1)$ (and $(\bound{G}_2,\bar{R}_2)$),
    we have that
    \begin{quote}
        $(\bound{C},\bar{Q})\oplus(\bound{G}_1,\bar{R}_1)$  and $(\bound{C},\bar{Q})\oplus(\bound{G}_2,\bar{R}_2)$ have  same annotated $(d,r,w)$-type.
    \end{quote}
\end{lemma}

\begin{proof}
Let $\bound{G}_1=(G_1,B,\rho)$ and $\bound{G}_2=(G_2,B,\rho)$ be compatible boundaried graphs
(with the same boundary~$B$ and the same boundary labeling $\rho$), and let
$\bar{R}_1,\bar{R}_2$ be the associated $d$-tuples.
Assume that
\[
\exttype_{d,r,w}(\bound{G}_1,\bar{R}_1)=\exttype_{d,r,w}(\bound{G}_2,\bar{R}_2).
\]
By definition of extended annotated type and our assumption,
for every $I\subseteq\binom{[|B|]}{2}$ we have
\begin{equation}\label{eq_Gparts_same_type}
\type_{d,r,w}\big(\bound{G}_1^I,\bar{R}_1\big)
=
\type_{d,r,w}\big(\bound{G}_2^I,\bar{R}_2\big).
\end{equation}

Let $\bound{C}=(C,B,\rho)$ be a boundaried graph compatible with both, and let $\bar{Q}$ be a $d$-tuple as in the statement. Write
\[
\mathcal{A}_i := (\bound{C},\bar{Q})\oplus(\bound{G}_i,\bar{R}_i)\qquad\text{for } i\in\{1,2\}.
\]

We show that for every $\CMSOttwDP$-sentence $\varphi$ of quantifier rank at most $d$, $\DP$-rank at most $r$,
and $\ttw$-rank at most $w$, 
\[
\mathcal{A}_1\models\varphi\quad\Longleftrightarrow\quad \mathcal{A}_2\models\varphi.
\]
Note that this is exactly the claim that they have the same annotated $(d,r,w)$-type.

We show this by induction on the structure of $\varphi$.

\medskip\noindent
\emph{Atomic predicates.}
This case is standard; see, e.g.,~\cite[Lemma~2.3]{Grohe2007LogicGA}.

\medskip\noindent
\emph{Boolean connectives.}
This is immediate from the induction hypothesis.

\smallskip\noindent
\emph{First-order quantifiers.}
Also this case is standard; see e.g.,~\cite[Lemma~2.3]{Grohe2007LogicGA}.

\smallskip\noindent
\emph{Disjoint-paths predicate $\DP_k$.}
This case is handled exactly as in~\cite{SchirrmacherSSTV23} (which is the full version of~\cite{SchirrmacherSSTV24mode}).

\smallskip\noindent
\emph{Restricted set quantifiers with counting.}
This is the only really new and technically relevant case.
Let
\[
\varphi = \exists_{\ttw\le k} X\ \psi(X),
\]
where $k\le w$ and where $\psi$ may contain counting predicates referring to $X$.
Assume that
\[
\mathcal{A}_1 \models \exists_{\ttw\le k} X\,\psi(X),
\]
and fix a witnessing set $S\subseteq V(\mathcal{A}_1)$ such that
\[
\ttw(\mathcal{A}_1,S)\le k
\qquad\text{and}\qquad
\mathcal{A}_1\models \psi(S).
\]

Let 
$S_C := S\cap V(C)$ and $S_1 := S\cap V(G_1)$.
We keep $S_C$ fixed and argue that there exists a set
$S_2\subseteq V(G_2)$ such that for
$S' := S_C \cup S_2$, we have 
\[
\ttw(\mathcal{A}_2,S')\le k
\quad\text{and}\quad
\mathcal{A}_2\models \psi(S').
\]

By the induction hypothesis applied to $\psi$, there are formulas $\theta(X)$ and $\xi(X)$ such that for each $i\in\{1,2\}$,
\[
\mathcal{A}_i \models \psi(S_C\cup S_i)
\quad\Longleftrightarrow\quad
(\bound{C},\bar{Q}) \models \theta(S_C)\ \text{ and }\ \bound{G}_i^{I}\models \xi(S_i)
\ \text{ for some } I\subseteq \binom{[|B|]}{2}.
\]

Because of~\eqref{eq_Gparts_same_type}, there is also a set $S_2\subseteq V(G_2)$ such that $(\bound{G}_2^I,\bar{R})\models\xi(S_2)$.

We set $S':=S_C\cup S_2$ and we argue that \[
\ttw(\mathcal{A}_2,S')\le k
\quad\text{and}\quad
\mathcal{A}_2\models \psi(S').
\]

To see that $\ttw(\mathcal{A}_2,S')\le k$, consider the annotated boundaried graphs $(G_1,S_1,B,\rho)$ and $(G_2,S_2,B,\rho)$ and note that (under the assumption of~\cref{rem_rank_sufficiency}) the fact that they have the same extended annotated type implies that they are $f_{\ref{lem_repl_compos}}^{\ttw}(w)$-equivalent. Therefore, from~\cref{lem_repl_compos} we get that $\ttw(\mathcal{A}_1,S)\le k$ if and only if $\ttw(\mathcal{A}_2,S')\le k$.

\smallskip
It remains to argue that $\psi(S')$ holds in $\mathcal{A}_2$.
Since $\psi$ has strictly smaller quantifier rank than $\varphi$, the induction hypothesis applies,
provided that all atomic predicates occurring in~$\psi$ have the same truth value under $S$ and $S'$.

This is immediate for first-order atoms and $\DP$-predicates by the previous cases.
For counting predicates referring to $X$, note that such predicates are always evaluated
on the witnessing set itself.
Their truth value depends only on:
(i) the intersection $S\cap B$ (a subset of the fixed finite boundary), and
(ii) the cardinalities of~$S_C$ and $S_i$ modulo the relevant moduli.
The boundary contribution is identical for $S$ and $S'$ by construction,
and the modulo contributions of $S_1$ and $S_2$ are part of the finite interface information
encoded in the extended annotated types.
Hence all counting predicates occurring in $\psi$ have the same truth value under $S$ and $S'$.
By the induction hypothesis, we conclude
\[
\mathcal{A}_1\models \psi(S)
\quad\Longleftrightarrow\quad
\mathcal{A}_2\models \psi(S'),
\]
and therefore
\[
\mathcal{A}_2\models \exists_{\ttw\le k} X\,\psi(X).
\]

The converse implication is symmetric.

\medskip
This completes the induction and shows $\mathcal{A}_1\models\varphi$ iff $\mathcal{A}_2\models\varphi$
for every sentence $\varphi$ within the rank bounds. Therefore, $\mathcal{A}_1$ and $\mathcal{A}_2$ have the same
annotated $(d,r,w)$-type, as required.
\end{proof}

\section{Computing representatives}
\label{sec_represent}

Our dynamic programming algorithm works bottom-up over an unbreakable decomposition and computes a representative of each bag. For a node $t$ of a tree decomposition,
we use~$\bound{G}_t$ to denote the boundaried graph $(G[\cone(t)],\adh(t),\rho)$, for some $\rho\colon \adh(t)\to|\adh(t)|$.



\subparagraph{$(d,r,w)$-representatives.}
Let $d,r,w\in\mathbb{N}$. Let $\mathbf{G},\mathbf{G}'$ be two boundaried graphs and let $\bar{R}$ and $\bar{R}'$ be two $d$-tuples of subsets of $V(G)$ and $V(G')$, respectively.
We say that $(\mathbf{G}',\bar{R}')$ is a \emph{$(d,r,w)$-representative} of $(\mathbf{G},\bar{R})$ if $\exttype_{d,r,w}(\mathbf{G},\bar{R})=\exttype_{d,r,w}(\mathbf{G}',\bar{R}').$


The main subroutine of our algorithm is given in the following lemma, whose proof is deferred to \cref{subsec_proof_first_routine}.

\begin{lemma}\label{lem_representatives}
    There is a function $f_{\size}\colon \mathbb{N}^4\to\mathbb{N}$ and an algorithm that, given:
    \begin{itemize}
        \item integers $\ell,q,d,r,w\in\mathbb{N}$;
        \item a graph $G$ that excludes a graph $H$ as a topological minor;
        \item a $d$-tuple $\bar{R}$ of subsets of $V(G)$;
        \item a tree decomposition $(T,\mathsf{bag})$ of $G$ of adhesion $\ell$; and
        \item a node $t$ of $T$ such that $G[\mathsf{cone}(t)]$ is $(q,w+1)$-unbreakable;
    \end{itemize}
    outputs, in time $\mathcal{O}_{|H|,r,\ell,q,d,w}(|\mathsf{cone}(t)|^3)$, a $(d,r,w)$-representative of $(\bound{G}_t,\bar{R}\cap \mathsf{cone}(t))$ of size at most $f_{\mathsf{size}}(d,r,w,\ell)$.
\end{lemma}

\subsection{Subroutines}
This subsection contains a series of definitions, observations, and lemmas that will be used in the proof of~\cref{lem_representatives} and in the proof of~\cref{th_our_g_res_tp_bog}.



\subparagraph{Reducing the instance away from the quantification.}
We need the following result about the efficient computation of representatives.

\begin{lemma}\label{lem_computing_folio}
    There is a computable function $f'\colon \mathbb{N}^3\to\mathbb{N}$ such that the following holds. Let $d,r,w,z\in\mathbb{N}$, let~$\bound{G}$ be an $n$-vertex boundaried graph, and let $\bar{R}=(R_1,\ldots,R_d)$ be a $d$-tuple of subsets of $V(G)$, such that for every $i\in[d]$, $R_i$ has size at most $z$ and contains $B(\bound{G})$ as a subset. Then, there is a $(d,r,w)$-representative of $(\bound{G},\bar{R})$ of size at most $f'(d,r,w,z)$. Moreover, such representative can be computed in time $\mathcal{O}_{d,r,w,z}(n^3)$.
\end{lemma}

Its proof is based on the following result from~\cite{grohe2011finding}.

\begin{proposition}[Lemma 2.2 and Theorem 3.1 of~\cite{grohe2011finding}]
\label{prop_bound_rep_folio}
There is a computable function $f_{\mathsf{folio}}\colon \mathbb{N}^2\to\mathbb{N}$ such that
for every $\delta,b\in\mathbb{N}$ and every extended $\delta$-folio $\mathcal{F}$, the following holds. Every minimum-size (in terms of vertices) $b$-boundaried graph whose extended $\delta$-folio is~$\mathcal{F}$ has size at most $f_{\mathsf{folio}}(\delta,b)$.
Moreover, there is an algorithm that given $\delta,b,$ and an $n$-vertex $b$-boundaried graph $\bound{G}$, computes, in time $\mathcal{O}_{\delta,b}(n^3)$, a $b$-boundaried graph~$\bound{H}$
of size at most $f_{\mathsf{folio}}(\delta,b)$ such that $\bound{G}$ and $\bound{H}$ are $\delta$-equivalent.
\end{proposition}

Using~\cref{prop_bound_rep_folio}, we can show~\cref{lem_computing_folio}.

\begin{proof}[Proof of~\cref{lem_computing_folio}]
    Let $B^\star:=\bigcup_{i\in[d]}  R_i\cup B(\bound{G})$ and $\bound{G}^\star$ be the boundaried graph $(G,B^\star,\rho^\star)$, where $\rho^\star$ is an arbitrary bijection from $B^\star$ to $|B^\star|$.
    We set $b:=|B^\star| \le z\cdot d$, $\delta:=\max\{r,f_{\ttw}(w)\}$, where $f_{\ttw}$ is the function of~\cref{prop_minor_obs} for $\p=\ttw$ (which applies because of~\cref{obs_monotone_ttw}).
    Then, we set $f'(d,r,w,z):=f_{\mathsf{folio}}(\delta,b)$, where $f_{\mathsf{folio}}$ is the function of~\cref{prop_bound_rep_folio}.
    
    We apply the algorithm of~\cref{prop_bound_rep_folio} for $\delta,b,$ and~$\mathcal{F}$. This way, we compute a boundaried graph $\widehat{\bound{G}}:=(\widehat{G},\widehat{B},\widehat{\rho})$ of size at most $f_{\mathsf{folio}}(\delta,b)$ such that $\bound{G}^\star$ and $\widehat{\bound{G}}$ are $\delta$-equivalent.
    For each $i\in[d]$, let $\widehat{R}_i$ be the vertex subset $\{\widehat{\rho}^{-1}\circ\rho^\star(v)\mid v\in R_i\}$ and $\widehat{B}^-:=\{\widehat{\rho}^{-1}\circ\rho^\star(v)\mid v\in B(\bound{G})\}$. We also let $\widehat{\bound{G}}^-:=(\widehat{G},\widehat{B}^-,\widehat{\rho}|_{\widehat{B}^-})$.
    It is easy to observe that $(\widehat{\bound{G}}^-,\widehat{R}_1,\ldots,\widehat{R}_d)$ and $(\bound{G},R_1,\ldots,R_d)$ have the same extended annotated $(d,r,w)$-type.
    Thus, $(\widehat{\bound{G}}^-,\widehat{R}_1,\ldots,\widehat{R}_d)$ is a $(d,r,w)$-representative of $(\bound{G},R_1,\ldots,R_d)$ of size at most $f'(d,r,w,z)$. The claimed running time follows directly from the one given by~\cref{prop_bound_rep_folio}.
\end{proof}

\subparagraph{Finding representatives in bounded treewidth graphs.}
It is easy to see that the extended annotated type of a graph is $\CMSO$-expressible and therefore, using Courcelle's theorem~\cite{Courcelle90} we can efficiently compute representatives of bounded treewidth graphs.
Note that the size of the representative is independent of the treewidth of the input graph.

\begin{observation}\label{lem_rep_tw}
    There is a function $p\colon \mathbb{N}^2\to\mathbb{N}$ and an algorithm that given integers $\ell,d,r,w,z\in\mathbb{N}$, an $n$-vertex $\ell$-boundaried graph $\bound{G}$, whose underlying graph $G$ has treewidth at most $z$, and a $d$-tuple $\bar{R}$ of subsets of $V(G)$, outputs, in time $\mathcal{O}_{\ell,d,r,w,z}(n)$, a $(d,r,w)$-representative of $(\bound{G},\bar{R})$ of size at most $p(\ell,d,r,w)$.
\end{observation}

\subsection{The proof of~\cref{lem_representatives}}
\label{subsec_proof_first_routine}

In order to show~\cref{lem_representatives}, we use~\cref{lem_collapse_to_FODP} to show that we can rewrite every given sentence of $\CMSOtwDP$ to a sentence of $\FODP$, assuming that we work with unbreakable graphs.  After rewriting everything to $\FODP$ (with some controlled blow-up in the quantifier rank, but not in the $\DP$-rank), we can use the model checking algorithm for $\FODP$ by Schirrmacher,  Siebertz,  Stamoulis, Thilikos, and  Vigny~\cite{SchirrmacherSSTV24mode}.

\begin{proposition}[\!\!\cite{SchirrmacherSSTV24mode}]\label{lem_mc_FODP}
  Let $\mathcal{C}_{H}$ be the class of graphs excluding a ﬁxed graph $H$ as a topological minor. Then, there is an algorithm that, given $G\in \mathcal{C}_{H}$ and an \textup{\FODP} formula~$\phi(\bar x)$, finds a $\bar v\in V(G)^{|\bar x|}$ such $G\models\phi(\bar v)$, if such exists, in time $f(\phi)\cdot |V(G)|^3$, where $f$ is a computable function depending on $H$.
\end{proposition}



We are now ready to show~\cref{lem_representatives}.
    

\begin{proof}[Proof of~\cref{lem_representatives}]
    We use $(R_1,\ldots,R_d)$ to denote the $d$-tuple $\bar{R}\cap\mathsf{cone}(t)$.
     Our first goal is to compute sets $R_1',\ldots,R_d'\subseteq V(G_t)$ such that \begin{itemize}
        \item for each $i\in[d]$, $R_i'$ is a subset of $R_i$ whose size depends only on $\ell,d,r,$ and $w$, and
        \item $\exttype_{d,r,w}(\bound{G}_t,R_1,\ldots,R_d)= \exttype_{d,r,w}(\bound{G}_t,R_1',\ldots,R_d').$
    \end{itemize}
    The rough idea is to use the fact that, under the assumptions of the lemma and because of~\cref{lem_collapse_to_FODP}, annotated types
    are $\FODP$-definable and therefore can be computed using the $\FODP$ model checking algorithm of~\cite{SchirrmacherSSTV24mode} for topological-minor-free graphs. 

We set $B:=\mathsf{adh}(t)$. We use $G_t$ to denote $G[\mathsf{cone}(t)]$ (which is the underlying graph of $\bound{G}_t$) and for every $I\in \binom{[|B|]}{2}$, we denote by $G_t^I$ the underlying graph of $\bound{G}_t^I$. By the hypothesis of the lemma, $G_t$ is $(q,w+1)$-unbreakable. Note that for every $I\in \binom{[|B|]}{2}$, $G_t^I$ is also $(q,w+1)$-unbreakable, since unbreakability is preserved when adding edges.

   We now show how to compute the sets $R_1',\ldots,R_d'$ inductively. Let $i\in[d]$ and suppose that the $i$th variable is a set variable.
   Assume that we have already computed vertex subsets $R_1',\ldots,R_{i-1}'\subseteq V(G_t)$ such that
   \begin{itemize}
        \item for each $j\in[i-1]$, $R_j'$ is a subset of $R_j$ whose size depends only on $\ell,d,r,$ and $w$, and
        \item $\exttype_{d,r,w}(\bound{G}_t,R_1,\ldots,R_d) = \exttype_{d,r,w}(\bound{G}_t,R_1',\ldots,R_{i-1}',R_i,\ldots,R_d).$
    \end{itemize}
    We fix an $I\in\binom{[|B|]}{2}$. Also, let $\bar{V}:=(V_1,\ldots,V_{i-1})$ be a tuple of $i-1$ subsets of $V(G_t)$ such that for each $j\in[i-1]$, $V_j\subseteq R_j'$ and $\ttw(G_t^I,V_j)\le w$.
    Note that since the size of each~$R_j'$ depends only on $\ell,d,r,$ and $w$, the number of different such tuples depends only on $\ell,d,r,$ and $w$. 

    Let $\Sigma$ be the (colored-graph) vocabulary of $(\bound{G}_t,\bar{R}\cap\mathsf{cone}(t))$.
    We set $\Sigma_i$ to be the colored-graph vocabulary extending $\Sigma$ by extra $i-1$ unary relation symbols $\mathsf{X}_1,\ldots,\mathsf{X}_{i-1}$, which will be interpreted as $V_1,\ldots,V_{i-1}$. 
    Let~$\Psi_i$ be the set of all (up to logical equivalence) formulas $\phi(X_{i})$ of $\CMSOtwDP[\Sigma_i]$, where $X_{i}$ is a free set variable and $\phi(X_{i})$ is of the form
    \begin{align*}
        Q_{i+1} X_{i+1}\ \ldots Q_\alpha X_\alpha\ Q_{\alpha+1} x_{\alpha+1}\ \ldots Q_{\alpha+\beta}x_{\alpha+\beta}\ & \psi(X_{i+1},\ldots,X_\alpha,x_{\alpha+1},\ldots,x_{\alpha+\beta}) \\
        & \land\bigwedge_{j=i+1}^{\alpha}X_j\subseteq R_j \land\bigwedge_{j=\alpha+1}^{\alpha+\beta}x_j\in R_j\\
        & \land \ttw(G_t^I,X_i)\le w,
    \end{align*}
where $\alpha,\beta\in\mathbb{N}$ with $\alpha+\beta=d$ and
\begin{itemize}
    \item $X_{i+1},\ldots,X_\alpha$ are vertex set variables and $x_{\alpha+1},\ldots,x_{\alpha+\beta}$ are vertex variables;
    \item for each $i\in[i+1,\alpha]$, $Q_i\in\{\exists_{\ttw\le w},\forall_{\ttw\le w}\}$, for some $w\in\mathbb{N}$;
    \item for each $i\in[\alpha+1,\alpha+\beta]$, $Q_i\in\{\exists,\forall\}$; and
    \item  $\psi(X_{i},\ldots,X_\alpha,x_{\alpha+1},\ldots,x_{\alpha+\beta})$ is a quantifier-free formula of 
    $\CMSOtwDP[\Sigma\cup\{\mathsf{X}_1,\ldots,\mathsf{X}_{i-1}\}]$.
\end{itemize}
Recall that we have fixed a set $I\in\binom{[|B|]}{2}$ and a tuple $\bar{V}$ of $i-1$ subsets of $V(G_t)$.
We say that two sets $U,U'\subseteq R_i$ are \emph{$(I,\bar{V})$-equivalent}, which we denote by $U\equiv^I_{\bar{V}}U'$ if for every formula $\phi(X_1,\ldots,X_{i})\in \Psi_i$, we have that
\[(G_t^I,\bar{R})\models \phi(V_1,\ldots,V_{i-1},U)\qquad \text{if and only if}\qquad (G_t^I,\bar{R})\models \phi(V_1,\ldots,V_{i-1},U').\]
It is easy to see that $\equiv^I_{\bar{V}}$ defines an equivalence relation and the number of equivalence classes of this relation depends only on $\ell,d,r,$ and $w$.

We next compute a ``representative'' of each equivalence class, i.e.,~a set $Q^I_{\bar{V}}\subseteq R_i$ such that for every $U\subseteq R_i$ with $\ttw(G_t^I,U)\le w$, there is a $U'\subseteq Q^I_{\bar{V}}$ with $\ttw(G_t^I,U')\le w$ such that $U\equiv^I_{\bar{V}}U'$.
First note that, because of~\cref{lem_small_unbreakable}, each set $X\subseteq V(G_t)$ with $\ttw(G_t^I,X)\le w$ has size at most\footnote{\cref{lem_small_unbreakable} assumes that the given graph contains at least $3q$ vertices. This is an assumption that we can have also here, since otherwise $(\bound{G}_t,\bar{R}\cap\mathsf{cone}(t))$ already has bounded size.} $q$. Also, because of~\cref{lem_collapse_to_FODP}, every formula $\phi(X_i)$ in~$\Psi_i$ can be rewritten in $\FODP[\Sigma^+]$ to a formula $\psi(x_1,\ldots,x_q)$ such that the following hold:
\begin{itemize}
    \item if there is a set $X\subseteq R_i$ with $\ttw(G_t^I,X)\le w$ and $(G_t^I,R_{i+1},\ldots,R_d)\models\phi(X)$, then there are vertices $x_1,\ldots,x_q\in R_i$ such that $(G_t^I,R_{i+1},\ldots,R_d)\models\psi(x_1,\ldots,x_q)$ and $X=\{x_1,\ldots,x_q\}$.
    
    \item if there are vertices $x_1,\ldots,x_q\in R_i$ such that $(G_t^I,R_{i+1},\ldots,R_d)\models\psi(x_1,\ldots,x_q)$, then $\ttw(G_t^I,X)\le w$ and $(G_t^I,R_{i+1},\ldots,R_d)\models\phi(X)$, where $X=\{x_1,\ldots,x_q\}$.
\end{itemize}
Therefore, for each such formula $\psi(x_1,\ldots,x_q)$, we use the algorithm of~\cref{lem_mc_FODP} to compute vertices $x_1,\ldots,x_q\in R_i$ such that $(G_t^I,R_{i+1},\ldots,R_d)\models\psi(x_1,\ldots,x_q)$. This can be done in time $\mathcal{O}_{\psi,|H|}(n^3)$ (recall that~$|\psi|$ depends on $\ell,d,r,w,$ and $q$).
We use~$S^{\phi_i}$ to denote the set $\{x_1,\ldots,x_q\}$ computed by this algorithm.
We set~$Q^I_{\bar{V}}$ to be the set $\bigcup\{S^{\phi_i} \mid \phi_i\in \Psi_i\}$. Note that the size of $Q^I_{\bar{V}}$ depends only on $\ell,d,r,w,$ and $q$ and $Q^I_{\bar{V}}$ contains a representative for each equivalence class of $\equiv^I_{\bar{V}}$. More formally, for every $U\subseteq R_i$ with \mbox{$\ttw(G_t^I,U)\le w$}, there is a $U'\subseteq Q^I_{\bar{V}}$ with $\ttw(G_t^I,U')\le w$ such that $U\equiv^I_{\bar{V}}U'$.
We set $Q_i^I:=\bigcup\{Q^I_{\bar{V}} \mid \bar{V}\in 2^{R_1}\times\cdots\times2^{R_{i-1}}\}$.
By construction, the set $Q_i^H$ has size depending only on $\ell,d,r,w,$ and $q$ and satisfies that
\[\type_{d,r,w}(G_t^I,R_1,\ldots,R_d) = \type_{d,r,w}(G_t^I,R_1',\ldots,R_{i-1}',Q_i^I,R_{i+1},\ldots,R_d).\]
We finally set $R_i':=\bigcup\{Q_i^I: I\in\binom{[|B|]}{2}\}$. Observe that $R_i'$
has size depending only on $\ell,d,r,w,$ and $q$ and that
\[\exttype_{d,r,w}(G_t,R_1,\ldots,R_d) = \exttype_{d,r,w}(G_t,R_1',\ldots,R_{i-1}',R_i',R_{i+1},\ldots,R_d).\]

This concludes the inductive step in case the $i$th variable is a set variable. In case it is a vertex variable, a direct application of~\cref{lem_mc_FODP} allows to compute representatives for the equivalence classes, defined as above. In the end of the induction, i.e.,~by applying the above until $i=d$, we get the claimed sets $R_1',\ldots,R_d'$ in time $\mathcal{O}_{|H|,\ell,d,r,w,q}(|\bound{G}_t|^3).$

We finally compute a $(d,r,w)$-representative of $(\bound{G}_t,R_1',\ldots,R_d')$. This is done using the algorithm of~\cref{lem_computing_folio}. This way, in time $\mathcal{O}_{\ell,d,r,w,q}(|\bound{G}_t|^3)$, we compute a $(d,r,w)$-representative $(\widehat{\bound{G}}_t,\widehat{R}_1,\ldots,\widehat{R}_d)$ of $(\bound{G}_t,R_1,\ldots,R_d)$ whose size depends only on $\ell,d,r,w,$ and~$q$.

In order to remove the dependency on $q$, we observe that $\widehat{G}$ has treewidth depending only on  $\ell,d,r,w,$ and $q$ and therefore, using the algorithm of~\cref{lem_rep_tw}, we can find a $(d,r,w)$-representative of $(\widehat{\bound{G}}_t,\widehat{R}_1,\ldots,\widehat{R}_d)$ (and also of $(\bound{G}_t,R_1,\ldots,R_d)$) whose size depends only on $\ell,d,r,$ and $w$.
\end{proof}   
    
\subsection{Unbreakability of the cone}


 \cref{lem_representatives} requires the whole cone under consideration to be unbreakable. We show that unbreakability (with worse bounds) is preserved when the subgraphs corresponding to subtrees of the current node have bounded size.

\begin{lemma}\label{lem_cone_unbreakability}
There is a  function $f_{\mathsf{unbr}}\colon\mathbb{N}^4\to\mathbb{N}$ such that the following holds.
Let $\ell,q,k,c\in\mathbb{N}$, $G$ be a graph, $\mathcal{T}=(T,\mathsf{bag})$ be a regular tree decomposition of $G$ of adhesion at most $\ell$, and $t\in V(T)$ such that:
\begin{itemize}
    \item $t$ has the strong $(q,k)$-unbreakability property;
    \item for every two distinct children $t_1,t_2$ of $t$, we have $\mathsf{adh}(t_1)\neq \mathsf{adh}(t_2)$; and
    \item  for every child $t'$ of $t$, it holds that $|\mathsf{cone}(t')|\le c$.
\end{itemize}
Then, $G[\mathsf{cone}(t)]$ is $(f_{\mathsf{unbr}}(\ell,c,q,k),k)$-unbreakable.
\end{lemma}
\begin{proof}
    We set $f_{\mathsf{unbr}}(\ell,c,q,k):= c\cdot\left(\binom{q}{\le \ell}+k\right) +q$.
    We use $G_t$ to denote $G[\mathsf{cone}(t)]$.
    Consider a separation $(A,B)$
    of $G_t$ of order at most $k$ and assume, without loss of generality, that $|A\cap \mathsf{bag}(t)|\le q$.
    We will show that $|A|\le f_{\mathsf{unbr}}(\ell,c,q,k)$.
    
    Given a child $t'$ of $t$, we say that it is \emph{$A$-rich} if $\mathsf{comp}(t')\subseteq A\setminus B$, and that it is \emph{$A$-poor} if $\mathsf{comp}(t')$ intersects both $A$ and $B$. Notice that if $A$ intersects the component of some child $t'$ of $t$, then $t'$ is either $A$-rich or $A$-poor.
    Our aim is to upper-bound the number of children of~$t$ that are $A$-rich or $A$-poor.
    
    We first bound the number of $A$-rich children. We claim that if a child $t'$ of $t$ is $A$-rich, then $\mathsf{adh}(t')\subseteq A$. To see this, note that if there is a vertex $v$ in $\mathsf{adh}(t')$ that is not in $A$, then it should belong to $B$. Also, because of regularity of $\mathcal{T}$, there should be a vertex $u\in\mathsf{comp}(t')$ that is adjacent to $v$. Since $t'$ is $A$-rich, $u$ belongs to $A\setminus B$ and therefore the existence of the edge $\{u,v\}$ contradicts the fact that $(A,B)$ is a separation of $G_t$. Therefore, we get that if $t'$ is $A$-rich, then $\mathsf{adh}(t')\subseteq A$. Also, because of the second property of $t$ in the statement, the  number of different children of $t$ whose adhesion is a subset of $A$ is upper-bounded by the number of different subsets of $A\cap\mathsf{bag}(t)$ of size at most $\ell$. Since $|A\cap \mathsf{bag}(t)|\le q$, we get that the number of $A$-rich children of $t$ is at most~$\binom{q}{\le \ell}$.

    For $A$-poor children of $t$, we argue as follows. We claim that if $t'$ is an $A$-poor child of $t$, then $\mathsf{comp}(t')$ contains a vertex of $A\cap B$. To see this, first note that because $\mathcal{T}$ is regular, the graph $G[\mathsf{comp}(t')]$ is connected. Therefore, if there is no vertex in $\mathsf{comp}(t')$ that is in both~$A$ and~$B$, there is an edge $\{u,v\}$ in $G[\mathsf{comp}(t')]$ such that $u\in A\setminus B$ and $v\in B\setminus A$ (these two vertices exist since $t'$ is $A$-poor), contradicting the fact that $(A,B)$ is a separation of $G_t$.
    Thus, we have that the number of $A$-poor children is at most $|A\cap B|$, which is upper-bounded by $k$.

    To obtain the upper bound on the size of $A$, we observe that the number of children of~$t$ whose components intersect $A$ is at most $\binom{q}{\le \ell}+k$. Since the component of every child of $t$ has at most $c$ vertices, we get that \[|A\setminus \mathsf{bag}(t)|\le c\cdot\left(\binom{q}{\le \ell}+k\right).\]
    Hence, $|A|\le c\cdot\left(\binom{q}{\le \ell}+k\right) +q  = f_{\mathsf{unbr}}(\ell,c,q,k)$.
\end{proof}

\section{The model checking algorithm}
\label{sec_mc}

In this section, we present our model checking algorithm, i.e.,~the proof of~\cref{th_our_g_res_tp_bog}.
Before presenting the algorithm, we show that when replacing a part of the graph by a representative, the combinatorial assumptions on the graph and the tree decomposition are preserved. This is used to show the correctness of our dynamic-programming procedure. 

\subsection{Preserving properties after gadget replacement}

In order to show that the combinatorial assumptions on the graph and the tree decomposition are preserved, we need to derive certain properties from the extended annotated type of a given boundaried graph. In particular, it is easy to observe that for sufficiently large values of $d$ and $r$ (depending on $\delta$ and $\ell$), the extended annotated $(d,r,w)$-type of an $\ell$-boundaried graph $\bound{G}$ can encode the graph induced by $B(\bound{G})$, the extended $\delta$-folio of $\bound{G}$, and whether $G\setminus B(\bound{G})$ is connected. This follows from the fact that all the above can be expressed in \FODP and is materialized in the following statement.

\begin{observation}\label{obs_representative_properties}
    There are two functions $g,g'\colon \mathbb{N}^2\to\mathbb{N}$ such that the following holds.
    Let $\ell,\delta,w\in\mathbb{N}$ and set $d:=g(\ell,\delta)$ and $r:=g'(\ell,\delta)$.
    Let $\mathbf{G},\mathbf{G}'$ be two $\ell$-boundaried graphs and let $\bar{R}$ and $\bar{R}'$ be two $d$-tuples of subsets of $V(G)$ and $V(G')$, respectively. If $\exttype_{d,r,w}(\mathbf{G},\bar{R}) = \exttype_{d,r,w}(\mathbf{G}',\bar{R}')$, then
    \begin{itemize}
    \item $(\mathbf{G}',\bar{R}')$ is compatible with $(\mathbf{G},\bar{R})$;
    \item the extended $\delta$-folios of $\mathbf{G}$ and $\mathbf{G}'$ are the same; and
    \item $G'\setminus B(\bound{G}')$ is connected if and only if $G\setminus B(\bound{G})$ is connected.
\end{itemize}
\end{observation}

\subparagraph{Preserving regularity and unbreakability.}
Recall that, given a graph $G$, a tree decomposition $\mathcal{T}=(T,\mathsf{bag})$, and a node $t\in V(T)$,
$\bound{G}_t$ denotes the boundaried graph $(G[\cone(t)],\adh(t),\rho)$, for some $\rho\colon \adh(t)\to|\adh(t)|$. We also use $G_t$ to denote the graph $G[\mathsf{cone}(t)]$.
We next show that in a tree decomposition of adhesion~$\ell$, if we replace $G_t$ with a graph with the same $\ell$-folio that also preserves the connectivity of the component, we have that regularity and strong unbreakability are preserved in the obtained tree decomposition.

\begin{lemma}\label{obs_maintain_regularity}
    Let $\ell\in\mathbb{N}$.
    Let $G$ be a graph and let $\mathcal{T}=(T,\mathsf{bag})$ be a tree decomposition of $G$ of adhesion~$\ell$.
    Let~$t$ be a node in $V(T)$.
    Let~$\bound{G}'$ be a boundaried graph such that~$\bound{G}_t$ and~$\bound{G}'$ have the same $\ell$-folio and  $G'\setminus B(\bound{G}')$ is connected if and only if $G_t\setminus B(\bound{G}_t)$ is connected.
    Also, let $\hat{G}$ be the graph obtained from $G$ by replacing $G[\mathsf{cone}(t)]$ with $G'$.
    Let $\hat{\mathcal{T}}$ be the tree decomposition of $\hat{G}$ obtained by $\mathcal{T}$ after removing all (strict) descendants of $t$ and setting $\mathsf{bag}(t)=V(G')$.
    Then, the following properties hold:
    \begin{itemize}
        \item $\hat{\mathcal{T}}$ has the same adhesion as $\mathcal{T}$;
        \item if $\mathcal{T}$ is regular, then $\hat{\mathcal{T}}$ is also regular; and
        \item for every node $x\neq t$ of $\hat{\mathcal{T}}$, it holds that if $x$ has the strong $(q,k)$-unbreakability property in $\mathcal{T}$, then it also has it in $\hat{\mathcal{T}}$.
    \end{itemize} 
\end{lemma}
\begin{proof}
    The adhesion and regularity properties follow directly from the construction of $\hat{\mathcal{T}}$. For the unbreakability, we argue as follows.

    Let $x$ be a node of $\hat{\mathcal{T}}$ that is not $t$ and let $G_x$ (resp.~$\hat{G}_x$) be the graph induced by the cone of $x$ in $\mathcal{T}$ (resp.~$\hat{\mathcal{T}}$).
    Note that, by definition of $\hat{G}$ and $\hat{\mathcal{T}}$, $\mathsf{bag}(x)$ is the same in both $\mathcal{T}$ and $\hat{\mathcal{T}}$.
    Our goal is to show that if $\mathsf{bag}(x)$ is $(q,k)$-unbreakable in $G_x$, then it is also $(q,k)$-unbreakable in $\hat{G}_x$.

    Assume towards a contradiction that there is a separation $(A,B)$ of $\hat{G}_x$ of order at most $k$ such that $|A\cap \mathsf{bag}(x)|>q$ and $|B\cap \mathsf{bag}(x)|>q$.
    Then, consider a minimum-order separation $(A^\star,B^\star)$ of $G_x$ with $A\cap \mathsf{bag}(x)\subseteq A^\star$ and $B\cap\mathsf{bag}(x)\subseteq B^\star$.
    
    We claim that $|A^\star\cap B^\star|\le k$. To see this, observe that $A^\star\cap B^\star$ intersects at most $\ell$ vertices of $G[\mathsf{cone}(t)]$, since otherwise we could replace the part of $A^\star\cap B^\star$ in $\mathsf{cone}(t)$ with $\mathsf{adh}(t)$ and still get a separation with the claimed properties and of smaller order. Also, recall that in $\hat{G}_x$, the graph $G[\mathsf{cone}(t)]$ is replaced with a graph of the same $\ell$-folio. Therefore, the size of the part of $A^\star\cap B^\star$ that is disjoint from $\mathsf{bag}(x)$ should be at most the size of $(A\cap B)\setminus \mathsf{bag}(x)$. Therefore, $|A^\star\cap B^\star|\le k$.
     
    Hence, for the separation $(A^\star,B^\star)$ of $G_x$, we have that $|A^\star\cap B^\star|\le k$ and $|A^\star\cap \mathsf{bag}(x)|\ge |A\cap\mathsf{bag}(x)|>q$ and $|B^\star\cap \mathsf{bag}(x)|\ge |B\cap\mathsf{bag}(x)|>q$, which contradicts the fact that $\mathsf{bag}(x)$ is $(q,k)$-unbreakable in $G_x$.
\end{proof}
\subparagraph{Preserving topological-minor-freeness.}
We also need to show that when replacing with a representative of the same extended folio, topological-minor-freeness is preserved. A similar version of this result was shown in~\cite{SchirrmacherSSTV24mode} but we include the proof here for completeness.

\begin{lemma}[\!\!\cite{SchirrmacherSSTV24mode}]\label{lem_excluded_top_minor}
    Let $c,h\in\mathbb{N}$.
    Let $G$ be a graph and let $\mathcal{T}$ be a tree decomposition of~$G$.
    Let $t$ be a non-root node in $V(T)$.
    Let~$\bound{G}'$ be a boundaried graph of order at most $c$ such that~$\bound{G}_t$ and~$\bound{G}'$ are $h$-equivalent and let~$\hat{G}$ be the graph obtained from $G$ by replacing $G[\mathsf{cone}(t)]$ with $G'$.
    Then, if $G$ excludes a graph $H$ of size at most $h$ as a topological minor, then the same holds for $\hat{G}$.
\end{lemma}

\begin{proof}
    Suppose, towards a contradiction, that there is a topological minor model $\hat{M}$ of $H$ in~$\hat{G}$. We say that a pair $\{v,u\}$ of vertices from $\mathsf{adh}(t)$ is \emph{connected outside $G'$ in $\hat{M}$} if there is a $u$-$v$-path in $\hat{M}$ whose internal vertices are neither in $V(G')$ nor principal vertices of $\hat{M}$.

    Let $F$ be the set of all pairs $\{v,u\}$ of vertices from $\mathsf{adh}(t)$ that are connected outside~$G'$ in~$\hat{M}$. 
    We consider the graph $Q=(R,F)$ and note that since $\bound{G}_t$ and $\bound{G}'$ have the same extended $h$-folio, $G[\mathsf{cone}(t)]+Q$ and $G'+Q$ have the same $h$-folio. Now set $K':=\hat{M}[V(G')]+Q$ and observe that $K'$ is a topological minor model of a graph $J$ of size at most $h$. Since $G[\mathsf{cone}(t)]+Q$ and $G'+Q$ have the same $h$-folio, there is also a topological minor model $K$ of the same graph $J$ in $G$. Therefore, the subgraph of $G$ induced by the vertices of $K$ and $\hat{M}\setminus V(G')$ contains $H$ as a topological minor, a contradiction.
\end{proof}

\subsection{Dynamic programming over unbreakable decompositions -- Proof of~\cref{th_our_g_res_tp_bog}}

Assume that we are given a formula $\varphi$ of our logic in prenex normal form with quantifier rank $d_0$, $\DP$-rank $r_0$, and $\ttw$-rank $w$ and a graph $G$ that excludes a graph $H$ as a topological minor.

We set $k:=w+1$, $\ell:=q(k)$, $q:=q(k)$, where $q(\cdot)$ is the function of~\cref{th_strong_unbreakability}, and $\delta=\max\{\ell,|H|\}$.
We also set $d=\max\{d_0,g(\ell,\delta)\}$ and $r=\max\{r_0,g'(\ell,\delta)\}$, where~$g$ and~$g'$ are the functions of~\cref{obs_representative_properties}.
Also, we set $c_1:=f_{\size}(d,r,w,\ell)$, where $f_{\size}$ is the function of~\cref{lem_representatives}, $c_2:=f_{\mathsf{rep}}(\ell,d,r,w)$, where $f_{\mathsf{rep}}$ is the function of~\cref{lem_rep_tw}, and $c:=\max\{c_1,c_2\}$.

\subsubsection{The algorithm}
The algorithm of~\cref{th_our_g_res_tp_bog} works as follows. 
We first compute a regular strongly $(q,k)$-unbreakable decomposition $\mathcal{T}=(T,\mathsf{bag})$ of adhesion at most $\ell$ using the algorithm of~\cref{th_strong_unbreakability}.
Then, we consider a post-order traversal 
%
%
of $T$, i.e.,~an ordering of $V(T)$ such that each node is visited after all its (strict) descendants are visited.
We process the nodes by increasing order.
Before any iteration, we set $G^\bullet:=G$, $\bar{R}=V(G)^d$, and $\mathcal{T}^\bullet:=\mathcal{T}$.

\subparagraph{Iteration step.}
We work with the graph $G^\bullet$ and its tree decomposition $\mathcal{T}^\bullet$.
Let $t$ be the node of $\mathcal{T}^\bullet$ that is currently processed.

Let $\mathsf{MaxAdh}(t)$ be the set of all children $z$ of $t$
such that there is no child $z'\neq z$ of $t$ such that $\adh(z)\subseteq \adh(z')$.
Note that for every child $y$ of $t$ there is a $z\in\mathsf{MaxAdh}(t)$ such that $\adh(y)\subseteq \adh(z)$.
For every $z\in\mathsf{MaxAdh}(t)$, we set \[G_z^{\mathsf{Max}}:=\bigcup\{G^\bullet[\cone(y)]\mid \text{$y$ is a child of $t$ with $\adh(y)\subseteq \adh(z)$}\}\]
and $\bound{G}_z^\mathsf{Max}:=(G_z^{\mathsf{Max}},\adh(z),\rho)$, where $\rho$ is an arbitrary bijection from $\adh(z)$ to $|\adh(z)|$.

Using the algorithm of~\cref{lem_rep_tw}, we compute a collection of pairs \[\mathcal{H}:=\{(z,(\bound{H}_z,\bar{R}_z))\mid z\in\mathsf{MaxAdh}(x)\},\] where for each $z\in\mathsf{MaxAdh}(x)$, $(\bound{H}_z,\bar{R}_z)$ is a $(d,r,w)$-representative of $(\bound{G}_z^{\mathsf{Max}},\bar{R}^\bullet\cap V(G_z^\mathsf{Max}))$ of size at most~$c_1$.

Let $(G',\bar{R}')$ be the (annotated) graph obtained from $(G^\bullet,\bar{R}^\bullet)$ by replacing, for each $z\in \mathsf{MaxAdh}(t)$,  $(\bound{G}_z^{\mathsf{Max}},\bar{R}^\bullet\cap V(G_z^{\mathsf{Max}}))$ with $(\bound{H}_z,\bar{R}_z)$.
Also, let $\mathcal{T}'$ be the tree decomposition obtained from $\mathcal{T}^\bullet$ after removing every node in the subtree rooted at $t$, except of the nodes in $\mathsf{MaxAdh}(t)$ and by setting, for each $z\in\mathsf{MaxAdh}(t)$, $\mathsf{bag}(z)=V(H_z)$.

We set
\begin{itemize}
    \item $\bound{C}:= (G'\setminus \mathsf{comp}(t),\mathsf{adh}(t),\rho)$, where $\rho$ is an arbitrary bijection from $\mathsf{adh}(t)$ to $|\mathsf{adh}(t)|$; 
    \item $\bar{Q}:=\bar{R}'\setminus \mathsf{comp}(t)$;
    \item $\bound{G}_t':=(G_t',\mathsf{adh}(t),\rho)$, where $G_t'$ is the subgraph of $G'$ induced by the cone of $t$ in $\mathcal{T}'$; and
    \item $\bar{R}_t':=\bar{R}'\cap V(G_t')$.
\end{itemize}
Then, we apply the algorithm of~\cref{lem_representatives} for the graph $G'$, the tree decomposition $\mathcal{T}'$, and the node $t$.
If it outputs a $(d,r,w)$-representative $(\bound{G}^\star,\bar{R}^\star)$ of $(\bound{G}_t',\bar{R}')$, then we set $(G^\bullet,\bar{R}^\bullet):=(\bound{C},\bar{Q})\oplus(\bound{G}^\star,\bar{R}^\star)$.
Also, set~$\mathcal{T}^\bullet$ to be the tree decomposition obtained from~$\mathcal{T}'$ after removing all children of $t$. This finishes the iterative step.

\subparagraph{Final step.}
When all nodes have been processed, apply Courcelle's theorem~\cite{Courcelle90} to decide if~$G^\bullet$ satisfies~$\varphi$ when its quantified variables are interpreted in the corresponding sets~$R_i^\bullet$, and report the corresponding answer.

\subsubsection{Proof of correctness}
We will show that the following invariants are maintained after each iteration. Assume that we are after the $i$th iteration. Then,
    \begin{enumerate}
        \item $G^\bullet$ excludes $H$ as a topological minor.

        \item $\mathcal{T}^\bullet$ is a tree decomposition with the following properties:
        \begin{enumerate}
            \item $\mathcal{T}^\bullet$ is regular;
            \item every node of $\mathcal{T}^\bullet$ that is not yet processed has the strong $(q,k)$-unbreakability property;
            \item for every node $t$ of $\mathcal{T}^\bullet$ that is already processed, we have $|\mathsf{cone}(t)|\le c$; and
        \end{enumerate}
        \item $(G,V(G)^d)$ and $(G^\bullet,\bar{R}^\bullet)$ have the same annotated $(d,r,w)$-type.
    \end{enumerate}
    We show the invariants by induction.
    We start by observing that in the beginning of each iteration, the underlying graph $G_z^\mathsf{Max}$ of $\bound{G}_z^\mathsf{Max}$ has treewidth at most $c+\ell$. This follows directly from the fact that every connected component of $G_z^\mathsf{Max}\setminus\mathsf{adh}(z)$ is a subgraph of $G[\mathsf{comp}(y)]$, for some child $y$ of $t$ with $\mathsf{adh}(y)\subseteq \mathsf{adh}(z)$ and that $|\mathsf{comp}(y)|\le c$; recall that because of the considered ordering, all children of $t$ have already been processed and therefore, by induction hypothesis, their cones have size at most $c$.
    Also, because of~\cref{lem_compositionality} and the induction hypothesis,  $(G,V(G)^d)$ and $(G',\bar{R}')$  have the same annotated $(d,r,w)$-type.
    Also, because of~\cref{obs_representative_properties} and~\cref{lem_excluded_top_minor}, $G'$ excludes $H$ as a topological minor.
    Because of~\cref{obs_representative_properties} and~\cref{obs_maintain_regularity}, $\mathcal{T}'$ is regular and has adhesion at most $\ell$, and every node of $\mathcal{T}'$ that is not a child of $t$ has the strong $(q,k)$-unbreakability property. 
    Since $t$ has the strong $(q,k)$-unbreakability property in~$\mathcal{T}'$, by~\cref{lem_cone_unbreakability}, the subgraph of $G^\bullet$ induced by the cone of $t$ in $\mathcal{T}'$ is $(q',k)$-unbreakable, where 
    $q'=f_{\mathsf{unbr}}(\ell,c,q,k)$. 
    Therefore, the requirements of~\cref{lem_representatives} are satisfied and its application correctly gives a $(d,r,w)$-representative $(\bound{G}^\star,\bar{R}^\star)$ of $(\bound{G}_t',\bar{R}')$. Because of~\cref{lem_compositionality}, $(\bound{C},\bar{Q})\oplus(\bound{G}^\star,\bar{R}^\star)$ and $(G',\bar{R}')$ have the same annotated $(d,r,w)$-type, which shows invariant (3). 
    Invariant (1) follows from~\cref{lem_excluded_top_minor}. 
    Finally, invariant (2) follows from the fact that the corresponding properties are holding already for~$\mathcal{T}'$ and therefore, because of~\cref{obs_maintain_regularity}, they also hold for $\mathcal{T}'$.

\subsubsection{Running time}
The algorithm of~\cref{lem_rep_tw} is applied in total a linear (in $n$) number of times. This algorithm is always invoked on graphs of treewidth at most $c+\ell$, and therefore the running time of each application is $\mathcal{O}_{\ell,d,r,w,c}(n) = \mathcal{O}_{d,r,w}(n)$. As for the algorithm of~\cref{lem_representatives}, it runs in time $\mathcal{O}_{|H|,r,\ell,q,d,w}(|\mathsf{cone}(t)|^3)$. Since for every child $z$ of $t$, we have that $|\mathsf{cone}(z)|\le c$, we have that \[|\mathsf{cone}(t)|\le |\mathsf{mrg}(t)|+ c\cdot |\{z\in V(T)|\mid z\text{ is a child of $t$ in $T$}\}|.\]
Since the margins of the nodes in the initial tree decomposition $\mathcal{T}$ partition the vertex set of the input graph~$G$, we have that \[\sum_{t\in V(T)}|\mathsf{cone}(t)|\le n+c\cdot 2n = \mathcal{O}_{d,r,w}(n).\]
Therefore, the total running time of all calls to the algorithm of~\cref{lem_representatives} is $\mathcal{O}_{h,d,r,w}(n^3)$.
This implies that the total running time of the iteration step is $\mathcal{O}_{|H|,d,r,w}(n^3)$. In the final step, the running time is $\mathcal{O}_{d,r,w}(n)$.

\section{Expressive power}
\label{sec_power}

This section discusses the expressive power of \CMSOtw and \CMSOtwDP compared to other logics mentioned in this paper.


First, it can be easily observed that $\CMSO$ collapses to $\CMSOtw$ on graphs of bounded treewidth.

\begin{lemma}\label{lem_MSOcollapse_bdtw}
    On classes of bounded treewidth, $\CMSO$ collapses to \textup{\CMSOtw}.
\end{lemma}
\begin{proof}
    If a graph $G$ has bounded treewidth, $\tw(\torso(G,V(G)))$ is bounded.
    Therefore, \CMSOtw allows the quantification over all subsets of vertices.
\end{proof}


\subparagraph{$\CMSOtw$ and $\FODP$.}
We next show that the logics $\FODP$ and $\CMSOtw$ are incomparable, and that furthermore, the addition of the disjoint-paths predicate to the logic $\CMSOtw$ strictly increases its expressive power.

We start by showing that not all formulas of $\CMSOtw$ can be expressed in $\FODP$.

\begin{lemma}\label{lem_separating_MSOtw_FODP}
    There is a sentence of \textup{\CMSOtw} that verifies whether a cycle has even length. However, there is no such sentence in \textup{\FODP}.
\end{lemma}
\begin{proof}
    Let $\varphi\in\CMSOtw$ be the following sentence
    \[\varphi=\exists_2 X~\forall x~\forall y~(E(x,y) \rightarrow ((x\in X \wedge y\not\in X) \vee (x\not\in X \wedge y\in X))).\]
    On a cycle of length $2n$, there is such a set $X$ such that $\torso(G,X)$ is a cycle of length $n$, and hence of treewidth $2$. On the other hand, on a cycle of odd length, it is easy to verify that such a set $X$ does not exist.
    However, it is proved in \cite[Theorem 4.8]{SchirrmacherSV23} that \FODP does not distinguish even and odd length cycles.
\end{proof}

Note that in the above proof, we do not use the modulo counting of \CMSOtw.
Also, since the example separating $\CMSOtw$ from $\FODP$ is a property of cycles, which are of bounded treewidth, it is natural to ask whether there is a \CMSOtw property (even without using modulo counting) of graphs of \textit{unbounded treewidth} that is not expressible in $\FODP$.
Because of~\cref{lem_collapse_to_FO}, this should be a property of graphs that are not $(q,k)$-unbreakable, for appropriate values of $q$ and $k$.
This is in fact true. Let $\mathcal{C}_1$ (resp. $\mathcal{C}_2$) be the class of graphs obtained from a clique by subdividing each edge an odd (resp. even) number of times. It is then not hard to see (following the same approach as in~\cref{lem_separating_MSOtw_FODP}) that $\CMSOtw$ can distinguish between~$\mathcal{C}_1$ and~$\mathcal{C}_2$, while $\FODP$ cannot.

As a byproduct of the separation of $\CMSOtw$ from \FODP~ (\cref{lem_separating_MSOtw_FODP}), combined with~\cref{lem_MSOcollapse_bdtw} and some results of~\cite{SchirrmacherSV23}, we obtain the following expressiveness hierarchy on classes of bounded treewidth.

\begin{lemma}
    On classes of bounded treewidth, it holds that
    \[\FO\subsetneq\textup{\FODP}\subsetneq\textup{\CMSOtw}=\textup{\CMSOtwDP}=\CMSO.\]
\end{lemma}
\begin{proof}
    \FODP has a higher expressive power than~\FO by~\cite{SchirrmacherSV23}, and \FODP and \CMSOtw are separated by~\cref{lem_separating_MSOtw_FODP}.
    The collapse of \CMSOtw, \CMSOtwDP, and \CMSO follows by~\cref{lem_MSOcollapse_bdtw} and the observation that disjoint paths can be expressed in $\CMSO$.
\end{proof}

We also show the following inexpressibility result.
\begin{lemma}\label{lem_separating_FODP_MSOtw}
    There is a sentence of \textup{\FODP} that is not expressible in \textup{\CMSOtw}.
\end{lemma}
\begin{proof}
    For all $n\in\Nbbb$, let $G_n$ and $H_n$ be two $n\times n$-grids with the following additional edges.
    The vertices of the first (leftmost) column $g^{(1)}_1,\ldots,g^{(1)}_n$ of the grid in $G_n$ are connected to the vertices of the last (rightmost) column $g^{(n)}_1,\ldots,g^{(n)}_n$ by edges $(g^{(1)}_i,g^{(n)}_i)$ for every $i\leq n$.
    The vertices of the first (leftmost) column $h^{(1)}_1,\ldots,h^{(1)}_n$ of the grid in $H_n$ are connected to the vertices of the last (rightmost) column $h^{(n)}_1,\ldots,h^{(n)}_n$ by edges $(h^{(1)}_i,h^{(n)}_{n+1-i})$ for every $i\leq n$.
    Note that this construction of $G_n$ and $H_n$ ensures that $G_n$ is planar while~$H_n$ is not.
    Therefore, there exists a fix \FODP formula that distinguishes~$G_n$ from $H_n$ for any $n$ by~\cite{SchirrmacherSV23}.

    The rest of this proof is now devoted to prove that there is no \CMSOtw formula that distinguishes $G_n$ from $H_n$.
    Assume towards a contradiction the existence of a number $k\in\Nbbb$ and a formula $\phi\in\CMSOtw_k$ such that for every $n\in\Nbbb$, we have that $G_n\models \phi$ and $H_n \not\models \phi$. Thanks to \cref{lem_collapse_to_FO}, fix $\psi$ to be the \FO formula that is equivalent to $\phi$ over all $((k+1)^2,k+1)$-unbreakable graphs. We now define $\ell$ as the quantifier rank of $\psi$ and finally $r:=(3^\ell-1)/2$. 
    
    Thanks to the locality of \FO, we soon conclude that $\psi$ (and hence $\phi$) do not distinguish~$G_{2r+1}$ and~$H_{2r+1}$. 
    By construction and according to \cite[Theorem~4.12]{Libkin04}, $r$ is the Hanf locality radius of $\psi$. Now note that for every~$i,j\leq n$, the $r$-neighborhood of $g_j^{(i)}$ in $G_{2r+1}$ is isomorphic to the one of $h_j^{(i)}$ in $H_{2r+1}$.
    This implies that~$G_{2r+1}$ and $H_{2r+1}$ are $r$-equivalent (as defined in \cite[Definition~4.3]{Libkin04}), which in turn implies that $G_{2r+1}\models \psi$ if and only if $H_{2r+1}\models \psi$. Since the same holds for $\phi$, we reach a contradiction.
\end{proof}

It is also easy to show that the addition of the disjoint-paths predicate in our logic gives extra expressive power.
\begin{lemma}\label{lem_CMSOtwDP_CMSOtw}
    There is a sentence of \textup{\CMSOtwDP} that cannot be expressed in \textup{\CMSOtw}.
\end{lemma}
\begin{proof}
    $\CMSOtw$ cannot express planarity as shown in~\cref{lem_separating_FODP_MSOtw}, and therefore \CMSOtw cannot express disjoint paths which is possible in~$\CMSOtwDP$.
\end{proof}

Finally, we strengthen the result of~\cref{lem_separating_FODP_MSOtw} by showing that there is even a sentence in $\FOconn$ that is not definable in $\CMSOtw$. 

\begin{lemma}\label{lem_FOconn_CMSOttw}
    The logics \textup{\FOconn} and \textup{\CMSOtw} are incomparable.
\end{lemma}
\begin{proof}
    By \cref{lem_separating_MSOtw_FODP}, there is a property expressible in \textup{\CMSOtw} that is not expressible in \textup{\FODP}, hence in \textup{\FOconn} either.       

    For the other direction, fix for every $n$ two graphs: $G_n$ as the $n\times n$-grid, and $H_n$ as two disconnected copies of $G_n$. On one hand, there is a fix $\FOconn$ formula that distinguishes any $G_n$ from $H_n$, simply checking whether the graph is connected.
    On the other hand, note that, following the proof of \cref{lem_collapse_to_FO}, no $\CMSOtw_{k-1}$ formula can quantify over a set~$X$ of size exceeding $2k^2$ (in neither $G_n$ or $H_n$). 
    As such a set $X$ (over, say, $H_n$) would contain more than $k^2$ vertices of a copy of $G_n$, a $(k^2,k)$-unbreakable graph, contradicting that $\ttw(H_n,X)\le k-1$. Hence, \CMSOtw collapses to \FO on $\mathcal{C}:= (G_n)_{n\in\N} \cup (H_n)_{n\in\N}$ and for every \FO formula, there is an $n\in\N$ such that the formula does not distinguish~$G_n$ from $H_n$.
\end{proof}

Let us finally point out that full $\CMSO$ is strictly more
expressive than $\CMSOttwDP$, already on the class of complete
graphs. Indeed, full $\CMSO$ can express the parity of the number of
vertices by the counting atom
$\operatorname{card}(V)\equiv 0 \pmod 2$.
On the other hand, on a complete graph $K_n$, for every
$X\subseteq V(K_n)$, we have
$\ttw(K_n,X)=|X|-1$
for $X\neq\emptyset$. 
Hence every $\ttw$-bounded set quantifier ranges
only over sets of bounded size and can be simulated by a bounded tuple of
first-order variables. Moreover, the disjoint-paths predicate is
first-order definable on complete graphs, since its truth depends only on
the equality pattern of its arguments. Consequently, every
$\CMSOttwDP$-sentence is equivalent, on complete graphs, to a
first-order sentence. But first-order logic cannot define the parity of
$|V(K_n)|$ on the class of complete graphs: all sufficiently large
complete graphs satisfy the same first-order sentences of any fixed
quantifier rank, while parity alternates with $n$. Therefore parity is
not expressible in $\CMSOttwDP$ on complete graphs.

Combining all the previous results yields the following theorem; see also~\cref{fig_fodp_msotw}. 

\begin{theorem}
    In general, it holds that
    \[\FO\subsetneq\textup{\FOconn}\subsetneq\textup{\FODP}\subsetneq\textup{\CMSOtwDP}\subsetneq\CMSO,\]
    and
    \[\FO\subsetneq\textup{\CMSOtw}\subsetneq\textup{\CMSOtwDP}\subsetneq\CMSO.\]
    However, the expressive power of \textup{\FOconn} and \textup{\FODP} with \textup{\CMSOtw} is not comparable.
\end{theorem}

\cref{fig_fodp_msotw} depicts all known relations between the logics mentioned in this paper. In the figure, for two logics $\mathcal{L},\mathcal{L}'$, we denote by $\mathcal{L}\not\sim\mathcal{L}'$ the fact that $\mathcal{L}$ and $\mathcal{L}'$ are incomparable and by $\mathcal{L}\subseteq \mathcal{L}'$ (resp. $\mathcal{L}\subsetneq \mathcal{L}'$) the fact that $\mathcal{L}'$ is more expressive (resp. strictly more expressive) than $\mathcal{L}$.

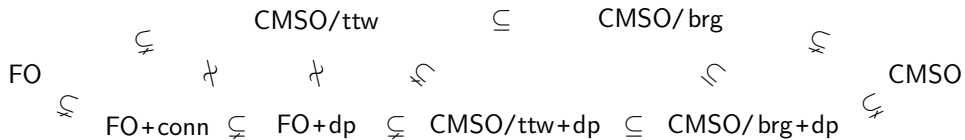
\begin{figure}[ht]
  \centering
  \begin{tikzpicture}[scale=.7]
    \node (fo) at (-1,0) {$\FO$};
    \node (foconn) at (1.5,-1) {$\FOconn$};
    \node (fodp) at (4.5,-1) {$\FODP$};
    \node (msotw) at (4.5,1) {$\CMSOtw$};
    \node (msotwdp) at (8.25,-1) {$\CMSOtwDP$};
    \node (msobrg) at (11,1) {$\CMSObrg$};
    \node (msoatwdp) at (12.75,-1) {$\CMSObrgDP$};
    \node (mso) at (16,0) {$\CMSO$};

    \node[rotate=-20] (u0) at (-.2,-.6) {$\subsetneq$};
    \node (u1) at (3,-1) {$\subsetneq$};
    \node[rotate=25] (u2) at (15,-.6) {$\subsetneq$};
    \node (u3) at (6,-1) {$\subsetneq$};
    \node (u4) at (10.5,-1) {$\subseteq$};
    \node (u5) at (8,1) {$\subseteq$};
    
    \node[rotate=10] (v00) at (1.25,.6) {$\subsetneq$};
    \node[rotate=40] (v0) at (2.5,0) {$\not\sim$};
    \node[rotate=40] (v1) at (4.5,0) {$\not\sim$};
    \node[rotate=-40] (v2) at (6.5,0) {$\subsetneq$};
    \node[rotate=-40] (v4) at (12,0) {$\subseteq$};
    \node[rotate=-20] (v5) at (14,.6) {$\subsetneq$};
  \end{tikzpicture}
  \caption{Comparison of the logics.}
  \label{fig_fodp_msotw}
\end{figure}

As one may observe, certain comparisons involving \CMSObrg and \CMSObrgDP are left open in the figure. We expect that \CMSObrgDP can be shown to be strictly more expressive than \CMSObrg using a proof similar to the one of~\cref{lem_CMSOtwDP_CMSOtw}. Since low bidimensionality are, in a sense, `local'' inside a grid~\cite[Lemma 17]{SauST25Parameterizing}, we also expect that~\cref{lem_FOconn_CMSOttw} can be strengthened to show that there are $\FOconn$ definable properties that cannot be expressed in \CMSObrg.
We further conjecture that $\CMSOtw \subsetneq \CMSObrg\subsetneq \CMSObrgDP$ and that $\CMSObrg\not\sim\CMSOttwDP$ but we opt for leaving these questions open for future research.

\section{Hardness on monotone classes that do not exclude a topological minor}\label{sec_hardness}

As mentioned in the introduction, on monotone classes of graphs not excluding any graph as a topological minor, hardness of model checking for $\CMSOttwDP$ follows directly from hardness of $\FOconn$~\cite{PilipczukSSTV22}, since $\FOconn$ is contained in $\CMSOttwDP$.
It is therefore natural to ask whether, in the absence of the disjoint-paths predicate, i.e.,~in the fragment $\CMSOttw$, the hardness still holds or it is possible to have efficient model checking for this fragment beyond topological-minor-free classes. 
In this section, we prove that it is not the case, namely by showing that $\CMSOttw$ can encode all graphs in any monotone class of graphs that does not exclude a fixed graph as a topological minor.
The key observation is that such a class contains an induced subdivision of every graph, and that these subdivisions can be undone with an $\CMSOttw$-formula. 
If the class additionally admits efficient encoding, then we can derive hardness of model checking for $\CMSOttw$.

\begin{definition}
A graph class $\Cc$ \emph{admits efficient encoding of topological minors} if for every graph $H$,
there exists a graph $G\in\Cc$ such that $H$ is a topological minor of~$G$, and moreover, given~$H$,
one can compute such a graph~$G$ together with an explicit topological-minor model of~$H$ in~$G$
in time polynomial in~$|H|$.
\end{definition}


\begin{theorem}
\label{th_aw_hard_ms_tw}
Let $\Cc$ be a monotone graph class that admits efficient encoding of topological minors.
Then \textup{\CMSOtw} model checking on $\Cc$, parameterized by the formula size, is AW$[*]$-hard.
\end{theorem}

\begin{proof}
We give an fpt-reduction from $\FO$ model checking on general graphs to $\CMSOtw$ model checking on~$\Cc$.

Let $H$ be an arbitrary input graph and let $\varphi$ be an $\FO$ formula. 
We can assume w.l.o.g.~that~$H$ does not contain vertices of degree $2$. If that is not the case, we define~$H_0$ as the graph obtained from $H$ by attaching two private leaves to every vertex of $H$.
In this way, we can distinguish the principal vertices (which do not have degree~$2$) of a minor model from the subdivision vertices (which all have degree~$2$).
This construction was already used in the hardness construction for separator logic (see~\cite[Lemma~9.1]{PilipczukSSTV22}). 

As $\Cc$ is monotone and admits efficient encoding of topological minors, we can efficiently compute a graph $G\in \Cc$ which is isomorphic to a subdivision of $H_0$. 
We show how to construct a formula $\psi\in \CMSOtw$ such that 
\[H\models \varphi \iff G\models \psi,\]
and thereby finishing the reduction. 
We define $\psi$ from $\varphi$ by the following syntactic rewriting.

\begin{enumerate}[label=(\arabic*)]
    \item \emph{Relativize all quantifiers to branch vertices.}
Let $\deg_{\ge 3}(x)$ be the FO-formula expressing that the vertex $x$ has degree at least~$3$.
Replace every first-order quantifier $\exists x\,\chi$ by $\exists x\,(\deg_{\ge 3}(x)\wedge \chi)$,
and every $\forall x\,\chi$ by $\forall x\,(\deg_{\ge 3}(x)\rightarrow \chi)$.
    \item \emph{Replace the edge predicate by “degree-2 path adjacency”.}
Replace each atomic predicate $E(x,y)$ by a formula~$E^\star(x,y)$ that holds iff either $x$ is adjacent
to $y$, or $x$ and $y$ are connected by a path whose internal vertices all have degree~$2$.

This predicate is definable in~$\CMSOtw$ by quantifying a set $P$ of vertices that is connected, only contains vertices of degree~$2$, and contains a neighbor of both~$x$ and~$y$.
As this set $P$ is an induced path in the graph~$G$, its torso is a cycle (connecting the two endpoints), which has treewidth~$2$, as required for the set quantification of~\CMSOtw. 
\end{enumerate}

Let $\psi$ be the result of applying (1) and (2) exhaustively to $\varphi$.
By definition, $G$ is isomorphic to a subdivision of~$H_0$, so there is a bijection $f$ from the vertices of degree~$3$ in~$G$ and the vertices of~$H$, such that $G\models E^\star(x,y)$ if and only if $\{f(x),f(y)\}$
is an edge in~$H$.
Therefore, $H\models \varphi \iff G\models \psi$.
\end{proof}

Note that the proof of \cref{th_aw_hard_ms_tw} holds not only for \textup{\CMSOtw}, but for $\CMSOp$ for any parameter $\p$ that is the \emph{torso extension} (defined in \cref{sec_conclusions}) of a parameter that is bounded on the class of cycles, such as treewidth or pathwidth (but not treedepth).

\section{Research directions and open questions}
\label{sec_conclusions}

In this paper, we studied two fragments of \CMSO, namely \CMSOttwDP\ and \CMSOatwDP,
generated by the annotated graph parameters $\ttw$ and $\atw$
(equivalently, $\bog$ and $\brg$), which represent two distinct parametric
extensions of treewidth.
On $H$-topological-minor-free graph classes, we establish an algorithmic
meta-theorem for \CMSOttwDP, complementing the known results of~\cite{SauST25Parameterizing} for
\CMSOatwDP\ on $H$-minor-free classes.
Together, these results support our view that \emph{annotated parameters}
provide a robust interface between logic and structure, as they nicely describe how complicated the quantified sets may be, while still being rich enough to cover many
algorithmic applications.

A natural next step is to broaden the framework beyond $\ttw$ and $\atw$.
One systematic way to do this starts from an arbitrary minor-monotone graph
parameter $\p$ and derives two canonical annotated variants.
We define the \emph{annotated extension} of $\p$, denoted by~$\ap$, by letting
$\ap(G,X)$ be the maximum value of $\p$ over all $X$-minors of $G$.
We also define the \emph{torso extension} of $\p$, denoted by $\tp$, by letting
$\tp(G,X)$ be the minimum value of $\p(\torso(G,X'))$ over all $X'\supseteq X$.
These two constructions are closely related but, in general, need not coincide, and 
understanding when they do coincide is a central structural question (see \cite{ProtopapasTW2025colorfulminors}).

\medskip
We conclude with a small list of concrete questions that we find particularly
compelling at the moment.

\begin{enumerate}
    \item \textbf{Model checking $\CMSOatwDP$ on topological-minor-free classes.}
    Can we extend the tractability of $\CMSOatwDP$ from minor-free classes~\cite{SauST2024parame} to topological-minor-free classes? We are not in position to conjecture whether our approach for model checking in the present paper can work also for $\CMSOatwDP$ or that  minor-freeness is the tractability boundary for this logic.
\medskip

    \item \textbf{Equivalence under functional equivalence of parameters.}
    Suppose $\p \sim \p'$ (in the sense used in \cref{lem_fun_min_same_logic}). Do the fragments $\CMSOp$ and $\CMSOp'$ already coincide in expressive power \emph{without} the $\DP$-predicate? Establishing this would clarify how much of the robustness of \cref{lem_fun_min_same_logic} relies on the presence of $\DP$.\medskip

    \item \textbf{Model checking beyond topological-minor-free classes for $\CMSO/\td$.}
    What is the precise tractability boundary for model checking $\CMSO/\td$?
    In particular, does fixed-parameter tractability extend meaningfully beyond topological-minor-free classes,
    or is there a natural hardness threshold?\medskip

    \item \textbf{Improving the running time of the meta-theorem.}
    Our current running time is dominated by the computation of an unbreakable decomposition, which is the main obstacle to getting below cubic dependence in the input size. Can this bottleneck be removed or replaced by a faster decomposition scheme?
    Similar considerations apply to the best known implementations of \FODP.\medskip

    \item \textbf{A dense analogue for hereditary classes.}
    Can one develop an analogue of the present framework for hereditary (possibly dense) graph classes?
    This likely requires annotated parameters that are \emph{not} minor-monotone, but still support a replacement/compositionality theory strong enough to yield algorithmic meta-theorems.
\end{enumerate}

%

%

\end{document}